\PassOptionsToPackage{table}{xcolor}
\documentclass[11pt]{article} 
\usepackage{fullpage}

\usepackage{times}
\usepackage{tikz}
\usepackage{mathtools}
\usepackage{amsthm}
\usepackage{bbm}
\usepackage{verbatim}
\usepackage{enumitem}
\usepackage{xspace}
\usepackage{todonotes}
\usepackage{amsfonts,mathrsfs,color}
\usepackage{xcolor}
\usepackage{graphicx}
\usepackage{booktabs}
\usepackage{nicefrac}
\usepackage{mdframed}
\usepackage{contour}
\contourlength{0.1pt}
\contournumber{10}
\usepackage[
    backend=biber,
    backref,
    style=alphabetic,
    maxcitenames=4,
    maxbibnames=99,
    natbib=true,
    url=false,
    doi=false, 
    isbn=false]{biblatex}
\usepackage{multirow}
\usepackage[ruled,linesnumbered,noend]{algorithm2e} 
\SetKwComment{Comment}{/* }{ */}
\SetKwProg{Fn}{function}
\SetAlFnt{\relax}
\SetAlCapFnt{\relax}
\SetAlCapNameFnt{\relax}
\SetAlCapHSkip{0pt}
\usepackage{tikz}
\usepackage{xcolor}

\newcommand{\declarecolor}[2]{\definecolor{#1}{RGB}{#2}\expandafter\newcommand\csname #1\endcsname[1]{\textcolor{#1}{##1}}}
\declarecolor{White}{255, 255, 255}
\declarecolor{Black}{0, 0, 0}
\declarecolor{Maroon}{128, 0, 0}
\declarecolor{Coral}{255, 127, 80}
\declarecolor{Red}{182, 21, 21}
\declarecolor{LimeGreen}{50, 205, 50}
\declarecolor{DarkGreen}{0, 80, 0}
\declarecolor{Purple}{146, 42, 158}
\declarecolor{Navy}{0, 0, 128}
\declarecolor{LightBlue}{84, 101, 202}
\usepackage{hyperref}
\hypersetup{
    colorlinks,
    citecolor=DarkGreen,
    linkcolor=LightBlue}
\usepackage[noabbrev, capitalise, nameinlink]{cleveref}
\newtheorem{definition}{Definition}
\newtheorem{theorem}{Theorem}
\newtheorem*{theorem*}{Theorem}
\newtheorem{lemma}{Lemma}
\newtheorem{fact}{Fact}
\newtheorem{corollary}{Corollary}
\newtheorem{example}{Example}
\newtheorem{assumption}{Assumption}

\newtheorem{claim}{Claim}
\newtheorem{proposition}{Proposition}
\newtheorem{remark}{Remark}

\newcommand{\reg}{\mathrm{Reg}}

\usepackage{pifont}
\newcommand{\cmark}{\ding{51}}%
\newcommand{\xmark}{\ding{55}}%

\DeclareMathOperator*{\argmin}{argmin}

\DeclareMathOperator{\poly}{poly}

\def\+#1{\mathcal{#1}}
\def\-#1{\mathbb{#1}}

\newcommand{\notshow}[1]{{}}
\newcommand{\AutoAdjust}[3]{{ \mathchoice{ \left #1 #2  \right #3}{#1 #2 #3}{#1 #2 #3}{#1 #2 #3} }}
\newcommand{\Xcomment}[1]{{}}

\newcommand{\InParentheses}[1]{\AutoAdjust{(}{#1}{)}}
\newcommand{\InBrackets}[1]{\AutoAdjust{[}{#1}{]}}
\newcommand{\InAngles}[1]{\AutoAdjust{\langle}{#1}{\rangle}}
\newcommand{\InNorms}[1]{\AutoAdjust{\|}{#1}{\|}}
\renewcommand{\part}[2]{\frac{\partial #1}{\partial #2}}

\newcommand{\R}{\mathbbm{R}}

\newcommand{\X}{\mathcal{X}}
\newcommand{\Y}{\mathcal{Y}}

\newcommand{\bbP}{\mathbb{P}}
\newcommand{\bbE}{\mathbb{E}}
\newcommand{\KL}{\mathrm{KL}}

\newcommand{\prox}{\mathrm{prox}}
\newcommand{\lscc}{\mathrm{lsc,c}}
\newcommand{\sm}{\mathrm{sm}}

\newcommand{\lce}{equilibrium\xspace}

\newcommand{\Int}{\mathrm{int}}
\newcommand{\Proj}{\mathrm{proj}}
\newcommand{\Beam}{\mathrm{beam}}
\newcommand{\Conv}{\mathrm{conv}}
\newcommand{\Phiint}{\Phi^{\X}_{\Int}}
\newcommand{\Phiintone}{\Phi^{\X}_{\Int^+}}
\newcommand{\Phiproj}{\Phi^{\X}_{\Proj}}
\newcommand{\Phiprojeq}{\Phi_{\Proj}}
\newcommand{\Phiprox}{\Phi^{\X}_{\prox}}
\newcommand{\Phiproxeq}{\Phi_{\prox}}
\newcommand{\Phibeam}{\Phi^{\X}_{\Beam}}
\newcommand{\Phiinteq}{\Phi_{\Int}}
\newcommand{\Phibeameq}{\Phi_{\Beam}}
\newcommand{\Phiintoneeq}{\Phi_{\Int^+}}

\newcommand{\regproj}{\reg_{\Proj,\delta}}

\newcommand{\Dx}{D_{\X}}

\title{On Tractable $\Phi$-Equilibria in Non-Concave Games\thanks{Authors are listed in alphabetical order. A preliminary version of the paper has been accepted to NeurIPS 2024.}}
\author{
\qquad 
Yang Cai\thanks{Yale University. Email: \texttt{yang.cai@yale.edu}}\\
\and 
Constantinos Daskalakis\thanks{MIT CSAIL. Email:\texttt{costis@csail.mit.edu}}\\
\and 
Haipeng Luo\thanks{University of Southern California. Email: \texttt{haipengl@usc.edu}} \qquad \\
\and
Chen-Yu Wei\thanks{University of Virginia. Email: \texttt{chenyu.wei@virginia.edu}}\\
\and 
Weiqiang Zheng\thanks{Yale University. Email: \texttt{weiqiang.zheng@yale.edu}}\\
}

\begin{document}

\maketitle

\begin{abstract}%
While Online Gradient Descent and other no-regret learning procedures are known to efficiently converge to a coarse correlated equilibrium in games where each agent's utility is concave in their own strategy, this is not the case when utilities are non-concave -- a common scenario in machine learning applications involving strategies parameterized by deep neural networks, or when agents' utilities are computed by neural networks, or both. Non-concave games introduce significant game-theoretic and optimization challenges: (i) Nash equilibria may not exist; (ii) local Nash equilibria, though they exist, are intractable; and (iii) mixed Nash, correlated, and coarse correlated equilibria generally have infinite support and are intractable. To sidestep these challenges, we revisit the classical solution concept of $\Phi$-equilibria introduced by~\citet{greenwald2003general}, which is guaranteed to exist for an arbitrary set of strategy modifications $\Phi$ even in non-concave games~\citep{stoltz2007learning}. However, the tractability of $\Phi$-equilibria in such games remains elusive. 

In this paper, we initiate the study of tractable $\Phi$-equilibria in non-concave games and examine several natural families of strategy modifications. We show that when $\Phi$ is finite, there exists an efficient uncoupled learning algorithm that converges to the corresponding $\Phi$-equilibria. Additionally, we explore cases where $\Phi$ is infinite but consists of local modifications. We show that approximating local $\Phi$-equilibria beyond the first-order stationary regime is computationally intractable. In contrast, within this regime, learning $\Phi$-equilibria reduces to achieving low $\Phi$-regret in online learning with convex loss functions, and we show Online Gradient Descent efficiently converges to $\Phi$-equilibria for several natural infinite families of modifications. 

A byproduct of our convergence analysis is a new structural family of modifications inspired by the well-studied proximal operator, and we refer to the corresponding regret as \emph{proximal regret}. This set of modifications is rich, and we show that small proximal regret implies not only small external regret but also other desirable properties, such as marginal coverage in online conformal prediction. To our knowledge, this notion has not been previously studied, even in online convex optimization. Despite the complexity of handling a rich set of modifications, we prove that Online Gradient Descent achieves sublinear \emph{proximal regret} fore convex loss functions.
\end{abstract}

\tableofcontents
\newpage
\section{Introduction}
Von Neumann's celebrated minimax theorem establishes the existence of Nash equilibrium in all two-player zero-sum games where the players' utilities are continuous as well as {\em concave} in their own strategy~\citep{v1928theorie}.\footnote{Throughout this paper, we model games using the standard convention in  Game Theory that each player has a utility function that they want to maximize. This is, of course, equivalent to modeling the players as loss minimizers, a convention more common in learning. When we say that a player's utility is concave (respectively non-concave) in their strategy, this is the same as saying that the player's loss is convex (respectively non-convex) in their strategy.} This assumption that players' utilities are concave, or quasi-concave, in their own strategies has been a cornerstone for the development of equilibrium theory in Economics, Game Theory, and a host of other theoretical and applied fields that make use of equilibrium concepts. In particular, (quasi-)concavity is key for showing the existence of many types of equilibrium, from generalizations of min-max equilibrium~\citep{fan1953minimax, sion1958general} to competitive equilibrium in exchange economies~\citep{arrow1954existence, mckenzie1954equilibrium}, mixed Nash equilibrium in finite normal-form games~\citep{nash1950equilibrium}, and, more generally, Nash equilibrium in (quasi-)concave games~\citep{debreu1952social, rosen1965existence}. 

Not only are equilibria guaranteed to exist in concave games, but it is also well-established---thanks to a long line of work at the interface of game theory,  learning and optimization whose origins can be traced to Dantzig's work on linear programming~\citep{george_b_dantzig_linear_1963}, Brown and Robinson's work on fictitious play~\citep{brown1951iterative,robinson1951iterative}, Blackwell's approachability theorem~\citep{blackwell_analog_1956} and Hannan's consistency theory~\citep{hannan1957approximation}---that several solution concepts are efficiently {computable} both centrally and via decentralized learning dynamics. For instance, it is well-known that the learning dynamics produced when the players of a game iteratively update their strategies using no-regret learning algorithms, such as online gradient descent, is guaranteed to converge to Nash equilibrium in two-player zero-sum concave games, and to coarse correlated equilibrium in multi-player general-sum concave games~\citep{cesa2006prediction}. The existence of such simple decentralized dynamics further justifies using these solution concepts to predict the outcome of real-life multi-agent interactions where agents deploy strategies, obtain feedback, and use that feedback to update their strategies.

While (quasi-)concave utilities have been instrumental in the development of equilibrium theory, as described above, they are also too restrictive an assumption. Several modern applications and outstanding challenges in Machine Learning, from training Generative Adversarial Networks (GANs) to Multi-Agent Reinforcement Learning (MARL) as well as generic multi-agent Deep Learning settings where the agents' strategies are parameterized by deep neural networks or their utilities are computed by deep neural networks, or both, give rise to games where the agents' utilities are {\em non-concave} in their own strategies. We call these games {\em non-concave}, following~\cite{daskalakis2022non}.

Unfortunately, classical equilibrium theory quickly hits a wall in non-concave games.
First, Nash equilibria are no longer guaranteed to exist. Second, while mixed Nash, correlated and coarse correlated equilibria do exist---under convexity and compactness of the strategy sets~\citep{glicksberg1952further}, which we have been assuming all along in our discussion so far, they have infinite support, in general~\citep{karlin1959mathematical}. Finally, they are computationally intractable; so, a fortiori, they are also intractable to attain via decentralized learning dynamics.


In view of the importance of non-concave games in emerging ML applications and the afore-described state-of-affairs, our investigation is motivated by the following broad and largely open question:

\textbf{Question from~\citep{daskalakis2022non}:} \textit{Is there a theory of non-concave games? What solution concepts are meaningful, universal, and tractable?}

\begin{table}[t]
\caption{A comparison between different solution concepts in multi-player non-concave games. We include definitions of Nash equilibrium, mixed Nash equilibrium, (coarse) correlated equilibrium, strict local Nash equilibrium, and second-order local Nash equilibrium in \Cref{app:preliminaries}. We also give a detailed discussion on the existence and computational/representation complexity of these solution concepts in \Cref{app:preliminaries}.}
\label{tab:equilibrium notions}
\resizebox{\textwidth}{!}{%
\begin{tabular}{cccc}
\hline
Solution Concept                                                     & Incentive Guarantee                                                                                                                             & Existence                      & Complexity (to compute or check existence)                                                                                                                          \\ \hline
Nash equilibrium                                                     &                                                                                                                                                 & \cellcolor[HTML]{C0C0C0}\xmark & \cellcolor[HTML]{C0C0C0}NP-hard                                                                                                                                     \\ \cline{1-1} \cline{3-4} 
Mixed Nash equilibrium                                               &                                                                                                                                                 & \cmark                         & \cellcolor[HTML]{C0C0C0}NP-hard                                                                                                                                     \\ \cline{1-1} \cline{3-4} 
(Coarse) Correlated equilibrium                                      & \multirow{-3}{*}{Global stability}                                                                                                              & \cmark                         & \cellcolor[HTML]{C0C0C0}NP-hard                                                                                                                                     \\ \hline
Strict local Nash equilibrium                                        & \cellcolor[HTML]{EFEFEF}Local stability                                                                                                         & \cellcolor[HTML]{C0C0C0}\xmark & \cellcolor[HTML]{C0C0C0}NP-hard                                                                                                                                     \\ \hline
Second-order local Nash equilibrium                                  & \cellcolor[HTML]{EFEFEF}Second-order stability                                                                                                  & \cellcolor[HTML]{C0C0C0}\xmark & \cellcolor[HTML]{C0C0C0}NP-hard                                                                                                                                     \\ \hline
Local Nash equilibrium                                               & \cellcolor[HTML]{EFEFEF}First-order stability                                                                                                   & \cmark                         & \cellcolor[HTML]{C0C0C0}PPAD-hard                                                                                                                                   \\ \hline
\textbf{$\Phi_{\mathrm{finite}}$-equilibrium (\cref{sec: finite-phi})}                        & \cellcolor[HTML]{EFEFEF}\begin{tabular}[c]{@{}c@{}}Stability against\\ finite deviations\end{tabular}                                           & \cmark                         & \begin{tabular}[c]{@{}c@{}}Effifient $\varepsilon$-approximation\\ for any $\varepsilon > 0$\end{tabular}                                                         \\ \hline
\textbf{$\Conv(\Phi_{\mathrm{finite}}(\delta))$-equilibrium (\cref{sec:convex-phi})} & \cellcolor[HTML]{EFEFEF}                                                                                                                        & \cmark                         & \cellcolor[HTML]{EFEFEF}\begin{tabular}[c]{@{}c@{}}Effifient $\varepsilon$-approximation\\ for $\varepsilon = \Omega(\delta^2)$\end{tabular}                        \\ \cline{1-1} \cline{3-4} 
\textbf{$\Phi_{\mathrm{prox}}(\delta)$-equilibrium (\cref{sec:proximal regret})}                  & \cellcolor[HTML]{EFEFEF}                                                                                                                        & \cmark                         & \cellcolor[HTML]{EFEFEF}\begin{tabular}[c]{@{}c@{}}Effifient $\varepsilon$-approximation\\ via GD/OG for $\varepsilon = \Omega(\delta^2)$\end{tabular}              \\ \cline{1-1} \cline{3-4} 
\textbf{$\Phiinteq(\delta)$-equilibrium (\cref{sec:phi-int-regret minization})}                   & \multirow{-3}{*}{\cellcolor[HTML]{EFEFEF}\begin{tabular}[c]{@{}c@{}}First-order stability\\ when $\varepsilon = \Omega(\delta^2)$\end{tabular}} & \cmark                         & \cellcolor[HTML]{EFEFEF}\begin{tabular}[c]{@{}c@{}}Effifient $\varepsilon$-approximation\\ via no-regret learning for $\varepsilon = \Omega(\delta^2)$\end{tabular} \\ \hline
\textbf{$\Phiintoneeq(\delta)$-equilibrium (\cref{sec:first-order regime hardness})}                 & \cellcolor[HTML]{EFEFEF}\begin{tabular}[c]{@{}c@{}}High-order stability\\ when $\varepsilon = o(\delta^2)$\end{tabular}                         & \cmark                         & \cellcolor[HTML]{C0C0C0}\begin{tabular}[c]{@{}c@{}}NP-hard $\varepsilon$-approximation\\ for $\varepsilon = o(\delta^2)$\end{tabular}                               \\ \hline
\end{tabular}%
}
\end{table}

\begin{figure}[ht]
    \centering
    \caption{The relationship between different solution concepts in non-concave games. An arrow from one solution concept to another means the former is contained in the latter. The dashed arrow from $\Conv(\Phi(\delta))$-equilibria to $\Phi_{\mathrm{finite}}$-equilibria means the former is contained in the latter when $\Phi(\delta) = \Phi_{\mathrm{finite}}$. }
    \includegraphics[width=0.5\textwidth]{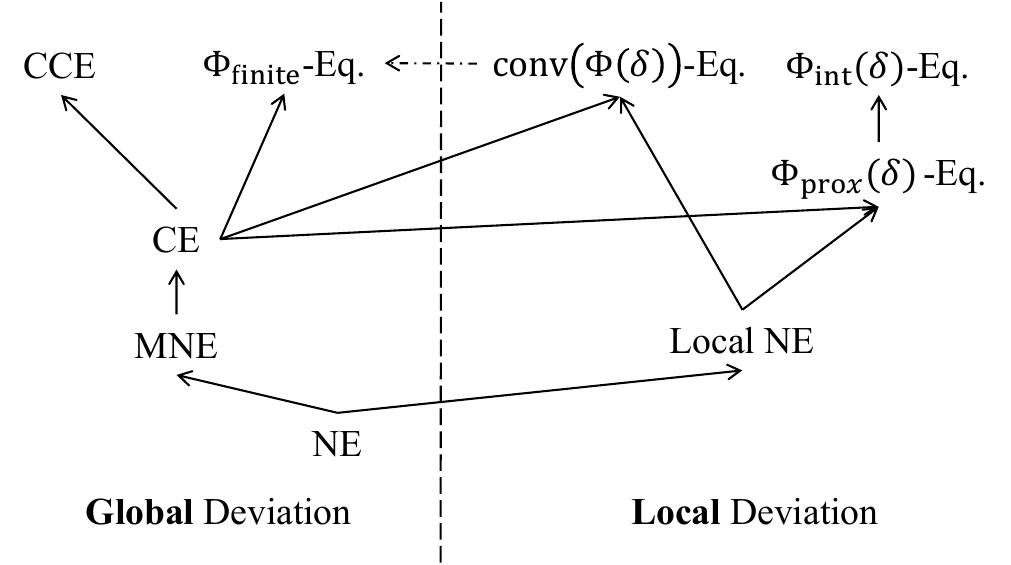}
    \label{fig:equilibria}
\end{figure}


\subsection{Contributions}
We study Daskalakis' question through the lens of the classical solution concept of $\Phi$-equilibria introduced by~\citet{greenwald2003general}. This concept is guaranteed to exist for virtually any set of strategy modifications $\Phi$, even in non-concave games, as demonstrated by~\citet{stoltz2007learning}.\footnote{\citet{stoltz2007learning} only require the elements of $\Phi$ to be measurable functions.} However, the tractability of $\Phi$-equilibria in such games remains elusive. In this paper, we initiate the study of tractable $\Phi$-equilibria in non-concave games and examine several natural families of strategy modifications. 

\paragraph{$\Phi$-Equilibrium.} 
The concept of $\Phi$-equilibrium generalizes (coarse) correlated equilibrium. A $\Phi$-equilibrium is a joint distribution over $\Pi_{i=1}^n \X_i$, the Cartesian product of all players' strategy sets, and is defined in terms of a set, $\Phi^{\X_i}$, of \emph{ strategy modifications}, for each player $i$. The set~$\Phi^{\X_i}$ contains functions mapping $\X_i$ to itself. A joint distribution over strategy profiles qualifies as a $\Phi = \Pi_{i=1}^n \Phi^{\mathcal{X}_i}$-equilibrium if no player $i$ can increase their expected utility by using any strategy modification function, $\phi_i \in \Phi^{\mathcal{X}_i}$, on the strategy sampled from the joint distribution. The larger the set $\Phi$, the stronger the incentive guarantee offered by the $\Phi$-equilibrium. For example, if $\Phi^{\X_i}$ contains all constant functions, the corresponding $\Phi$-equilibrium coincides with the notion of coarse correlated equilibrium. Throughout the paper, we also consider $\varepsilon$-approximate $\Phi$-equilibria, where no player can gain more than $\varepsilon$ by deviating using any function from $\Phi^{\mathcal{X}_i}$. We study several families of $\Phi$ and illustrate their relationships in \Cref{fig:equilibria}.

\paragraph{Finite Set of Global Deviations.} The first case we consider is when each player $i$'s set of strategy modifications, $\Phi^{\mathcal{X}_i}$, contains a finite number of arbitrary functions mapping $\mathcal{X}_i$ to itself. As shown in~\citep{greenwald2003general}, if there exists an online learning algorithm where each player $i$ is guaranteed to have sublinear $\Phi^{\mathcal{X}_i}$-regret, the empirical distribution of joint strategies played converges to a $\Phi = \Pi_{i=1}^n \Phi^{\mathcal{X}_i}$-equilibrium. \citet{gordon2008no} consider $\Phi$-regret minimization but for concave reward functions, and their results, therefore, do not apply to non-concave games. \citet{stoltz2007learning} provide an algorithm that achieves no $\Phi^{\mathcal{X}_i}$-regret in non-concave games; however, their algorithm requires a fixed-point computation per step, making it computationally inefficient.\footnote{The existence of the fixed point is guaranteed by the Schauder-Cauty fixed-point theorem~\citep{cauty_solution_2001}, a generalization of the Brouwer fixed-point theorem. Hence, it's unlikely such fixed points are tractable.} Our first contribution is to provide an efficient randomized algorithm that achieves no $\Phi^{\X_i}$-regret for each player $i$ with high probability. 

\medskip\noindent
    \hspace{.1in}\begin{minipage}{0.95\textwidth}\textbf{Contribution 1:} Let $\mathcal{X}$ be a strategy set (not necessarily compact or convex), and $\Psi$ an arbitrary finite set of strategy modification functions for $\mathcal{X}$. We design a randomized online learning algorithm that achieves $O\left(\sqrt{T \log |\Psi|}\right)$ $\Psi$-regret, with high probability, for \emph{arbitrary} bounded reward functions on $\mathcal{X}$ (\Cref{theorem:finite-phi-regret}). The algorithm operates in time $\sqrt{T}|\Psi|$ per iteration. If every player in an $n$-player \emph{non-concave} game adopts this algorithm, the empirical distribution of strategy profiles played forms an $\varepsilon$-approximate $\Phi = \Pi_{i=1}^n \Phi^{\mathcal{X}_i}$-equilibrium, with high probability, for any $\varepsilon > 0$, after $\frac{1}{\varepsilon^2} \cdot \poly\left(\log \left(\max_i |\Phi^{\mathcal{X}_i}|\right), \log n\right)$ iterations.
\end{minipage}
\medskip

A notable feature of our result is that it only requires the loss functions to be bounded. In contrast, existing results for efficient $\Phi$-regret minimization typically assume stronger conditions, such as continuity, convexity, or smoothness.

If players have a continuous action set with  infinitely many global strategy modifications, we can extend \Cref{alg:phi-reg} by discretizing the set of strategy modifications under mild assumptions, such as the modifications being Lipschitz (\Cref{cor:infinite global deviations via discretization}). The empirical distribution of the strategy profiles still converges to the corresponding $\Phi$-equilibrium, but at a much slower rate of $O(T^{-\frac{1}{d+2}})$, where $d$ is the dimension of the set of strategies. Additionally, the algorithm requires exponential time per iteration, making it inefficient. This inefficiency is unavoidable, as the problem remains intractable even when $\Phi$ contains only constant functions.

To address the limitations associated with infinitely large global strategy modifications, a natural approach is to focus on local deviations instead. The corresponding $\Phi$-equilibrium will guarantee local stability. The study of local equilibrium concepts in non-concave games has received significant attention in recent years---see e.g.,~\citep{ratliff2016characterization,hazan2017efficient,daskalakis2018limit,jin2020local,daskalakis2021complexity}. However, these solution concepts either are not guaranteed to exist, are restricted to sequential two-player zero-sum games~\citep{mangoubi2021greedy}, only establish local convergence guarantees for learning dynamics---see e.g.,~\citep{daskalakis2018limit,wang2020solving,fiez2020implicit}, only establish asymptotic convergence guarantees---see e.g.,~\citep{daskalakis2023stay}, or involve non-standard solution concepts where local stability is not with respect to a distribution over strategy profiles~\citep{hazan2017efficient}. 

We study the tractability of $\Phi$-equilibrium with infinitely large $\Phi$ sets that consist solely of local strategy modifications. A strategy modification $\phi \in \+X^\+X$ is $\delta$-local if $\InNorms{\phi(x)-x}\le \delta$ for all $x \in \+X$. These local solution concepts are guaranteed to exist in general multi-player non-concave games. Specifically, we focus on the following three families of natural deviations.

\begin{itemize}[leftmargin=*]
    \item[-] \textbf{Proximal Operator based Local Deviations:} 
    Each player $i$'s set of strategy deviations, denoted by $\Phi_\prox^{\+X_i}(\delta)$, contains all $\delta$-local deviations induced by the proximal operator \[\phi_f(x) = \prox_f(x):= \argmin_{x' \in \+X_i} \left\{ f(x') + \frac{1}{2}\InNorms{x'-x}^2 \right\},\] where $f$ is an arbitrary convex or smooth function. See \Cref{dfn:proximal operator} for the formal definition and more details on the notion of \emph{proximal regret} $\reg_f^T$. The deviation set \(\Phi_\prox^{\+X_i}(\delta)\) is rich, encompassing many interesting deviations. For instance, it contains the set of all $\delta$-local deviations based on projections $\Phi_{\text{Proj}}^{\X_i}(\delta)$. Specifically, $\Phi_{\text{Proj}}^{\X_i}(\delta)$ includes deviations of the form: $\phi_v(x) = \Pi_{\X_i}[x-v]$, where $\InNorms{v}\leq \delta$ and $\Pi_{\X_i}$ stands for the $\ell_2$-projection onto $\X_i$. In other words, $\Phi_{\text{Proj}}^{\X_i}(\delta)$ contains all deviations that attempt a small step from their input in a fixed direction and project if necessary. By choosing the function $f$ appropriately, we can show that $\Phi_\prox^{\+X_i}(\delta)$ also contains other interesting deviations. 
    
    \item[-] \textbf{Convex Combination of Finitely Many Local Deviations:} Each player $i$'s set of strategy modifications, denoted by $\Conv(\Phi^{\mathcal{X}_i}(\delta))$, contains all deviations that can be represented as a convex combination of a finite set of $\delta$-local strategy modifications, i.e., $\|\phi(x) - x\| \leq \delta$ for all $\phi \in \Phi^{\mathcal{X}_i}(\delta)$.
    
    \item [-] \textbf{Interpolation based Local Deviations:} each player $i$'s set of local strategy modifications, denoted by $\Phi^{\X_i}_{\Int, \Psi}(\delta)$, that contains all deviations that \emph{interpolate} between the input strategy and the strategy modification by $\Psi$. Formally, each element $\phi_{\lambda, \Psi}(x)$ of $\Phi^{\X_i}_{\Int, \Psi}(\delta)$ can be represented as $(1 - \lambda)x + \lambda \psi(x)$ for some $\psi \in \Psi$ and  $\lambda\leq \delta / D_{\X_i}$ ($D_{\X_i}$ is the diameter of $\X_i$). When $\Psi$ contains only constant functions, we denote the corresponding $\Phi^{\X_i}_{\Int, \Psi}(\delta)$ simply as $\Phi^{\X_i}_{\Int}(\delta)$. We remark that $\Phi^{\X_i}_{\Int}(\delta)\subseteq \Phi_{\text{prox}}^{\X_i}(\delta)$. See Section~\ref{sec:phi-int-regret minization} for details. 
\end{itemize}    

 For our three families of local strategy modifications, we explore the tractability of $\Phi$-equilibrium within a regime we term the \emph{first-order stationary regime}, where $\varepsilon = \Omega(\delta^2)$,\footnote{The regime $\varepsilon = \Omega(\delta)$ is trivial when the utility is Lipschitz.} with $\delta$ representing the maximum deviation allowed for a player. An $\varepsilon$-approximate $\Phi$-equilibrium in this regime ensures first-order stability. This regime is particularly interesting for two reasons:
(i) \citet{daskalakis2021complexity} have demonstrated that computing an $\varepsilon$-approximate $\delta$-local Nash equilibrium in this regime is intractable.\footnote{A strategy profile is considered an $\varepsilon$-approximate $\delta$-local Nash equilibrium if no player can gain more than $\varepsilon$ by deviating within a $\delta$ distance.} This poses an intriguing question: can correlating the players' strategies, as in a $\Phi$-equilibrium, potentially make the problem tractable?
(ii) Extending our algorithm, initially designed for finite sets of strategy modifications, to these three sets of local deviations results in inefficiency; specifically, the running time becomes exponential in one of the problem's natural parameters. Designing efficient algorithms for this regime thus presents challenges. 

To address these challenges, we first present a reduction showing that, in an $L$-smooth, possibly non-concave game (see \Cref{assumption:smooth games}), computing an $\varepsilon$-approximate $\Phi(\delta)$-equilibrium reduces to minimizing $\Phi^{\mathcal{X}_i}(\delta)$-regret against \emph{linear} losses, provided that $\varepsilon > \frac{\delta^2 L}{2}$ (i.e., within the first-order stationary regime). This reduction leverages a local first-order approximation of $L$-smooth functions, which incurs at most an approximation error of $\frac{\delta^2 L}{2}$. Thus, under the first-order stationary regime condition $\varepsilon > \frac{\delta^2 L}{2}$, it suffices to consider only linear losses when bounding $\Phi^{\mathcal{X}_i}(\delta)$-regret.

However, given that our strategy modifications are non-standard, it is a priori unclear how to minimize the corresponding $\Phi$-regret. For instance, to our knowledge, no algorithm is known to minimize $\Phiprox$-regret even when the loss functions are linear. Moreover, minimizing $\Phiprox$-regret is provably more challenging than minimizing external regret, in the sense that small $\Phiprox$-regret implies small external regret, but not vice versa. 

More specifically, proximal regret introduces a new notion of regret that lies between external regret and swap regret. Despite its generality, we prove that Online Gradient Descent achieves sublinear proximal regret.

\medskip\noindent
    \hspace{.1in}
    \begin{minipage}{0.95\textwidth}\textbf{Contribution 2:} We show that \hyperref[GD]{Online Gradient Descent (GD)} minimizes proximal regret in online convex optimization. Specifically, \hyperref[GD]{GD} guarantees $O(\sqrt{T})$ proximal regret for 
    for all convex / smooth functions $f$ simultaneously 
    (See \Cref{theorem:GD proximal regret} for formal statements).
\end{minipage}
\medskip
\noindent

Our result significantly extends the classical result of \hyperref[GD]{Online Gradient Descent}~\citep{zinkevich2003online} by showing that \hyperref[GD]{GD} not only minimizes external regret but also the more general proximal regret. This demonstrates that the sequence produced by \hyperref[GD]{GD} is not merely a no-external-regret sequence but also satisfies additional constraints. Moreover, our result implies that \hyperref[GD]{GD} dynamics in concave games converge to a refined subset of coarse correlated equilibria, which we refer to as \emph{proximal correlated equilibria}.\footnote{The general idea of providing tighter guarantees for \hyperref[GD]{GD} dynamics and refined classes of CCE is also explored in \citep{ahunbay2024first} but our results are different in various aspects. See \Cref{app:comparison} for details.} 

Our analysis introduces a novel perspective on \hyperref[GD]{Online Gradient Descent}, which we further extend to Optimistic Gradient Descent (\ref{OG}), yielding improved instance-dependent regret bounds (\Cref{thm:adversaril regret of OG}) and faster convergence in the game setting (\Cref{thm:game regret of OG}).

Using the non-concave to linear reduction, we further prove efficient equilibrium computation results for other classes of local strategy modifications.

\medskip\noindent
    \hspace{.1in}
    \begin{minipage}{0.95\textwidth}\textbf{Contribution 3:} For any $\delta > 0$, for each of the three families of infinite $\delta$-local strategy modifications mentioned above, there exists an efficient uncoupled learning algorithm that converges to an $\varepsilon$-approximate $\Phi$-equilibrium of the non-concave game in the first-order stationary regime, i.e., $\varepsilon = \Omega(\delta^2)$.
\end{minipage}
\medskip
\noindent 

We present our results for the proximal operator-based local deviation in \Cref{corollary:local proximal equilibrium} and \Cref{thm:game regret of OG}. Our result for the convex combination of local deviations can be found in \Cref{theorem:convex finite-phi-regret}. \Cref{thm:lce_int} contains our result for the interpolation-based local deviations. We remark that our consideration of strategy modifications based on proximal operators was influenced by the directions explored in~\citet{ahunbay2024first}. However, the two works study different and incomparable sets of modifications in several important respects. Moreover, our results and techniques are fundamentally distinct.\footnote{The first two versions of the paper only contain results for $\Phiproj(\delta)$ and $\Phiint(\delta)$. Following our initial version,~\citep{ahunbay2024first} proposed a generalized set of strategy modifications based on gradient fields, which includes both $\Phiproj(\delta)$ and $\Phiint(\delta)$ as special cases. His work encouraged us to consider a broader family of strategy modifications based on proximal operators, which we present here. While the two approaches are different and incomparable, the key distinction is that \citet{ahunbay2024first} studies a different notion of equilibrium and does not address $\Phi$-regret in the adversarial online learning setting or the approximation of standard $\Phi$-equilibria (\Cref{def:local CE}), which are the main focus of our work.  A detailed comparison with the results in \citet{ahunbay2024first} is provided in \Cref{app:comparison}.
}

Given that our results provide efficient uncoupled algorithms to compute $\varepsilon$-approximate $\Phi(\delta)$-equilibria in the first-order stationary regime $\varepsilon = \Omega(\delta^2)$, a natural question arises: \begin{equation*}
    \textnormal{\textit{Can we approximate $\Phi$-equilibria within $\varepsilon = o(\delta^2)$ efficiently?}}
\end{equation*} We answer this question in the negative: it is in fact NP-hard to achieve $\varepsilon = o(\delta^2)$. This hardness result further illustrates the necessity to consider the $\varepsilon=\Omega(\delta^2)$ regime. 

\medskip\noindent
    \hspace{.1in}
    \begin{minipage}{0.95\textwidth}\textbf{Contribution 4:} We show that the first-order stationary regime, i.e., $\varepsilon = \Omega(\delta^2)$, is the best one can hope for in terms of efficient computation of an \(\varepsilon\)-approximate \(\Phi(\delta)\)-equilibrium when \(\Phi(\delta)\) contains all \(\delta\)-local strategy modifications. Specifically, in \Cref{thm:hardnessFOSall swap}, we prove that, unless \(\textsc{P} = \textsc{NP}\), for any polynomial function $p(d,n,G,L)$ over $d$, $n$, $G$, and $L$, and for any \(\varepsilon \le p(d,n,G,L) \cdot \delta^{2+c}\)  with any constant \(c > 0\), there is no algorithm with running time \(\text{poly}(d,n, G, L,1/\varepsilon, 1/\delta)\) that can find an \(\varepsilon\)-approximate \(\Phi(\delta)\)-equilibrium. Here, $G$ is the Lipschitzness, and $L$ is the smoothness of the players' utilities. This result holds even for a single-player game over $\+X \subseteq [0,1]^d$. Moreover, we show in \Cref{thm:hardnessFOS restricted} that a similar hardness result holds even for a very restricted set of strategy modifications $\Phiintone(\delta)$ that is the union of $\Phiint(\delta)$ and one additional strategy modification.
\end{minipage}
\medskip
\noindent

We remark that existing hardness results for local maximizer \citep{daskalakis2021complexity} could be extended to approximate $\Phi(\delta)$-equilibrium. However, these results only apply to the ``global regime", where $\delta$ equals $D_\+X$, the diameter of the action set.\footnote{In \Cref{sec:hardness in global regime}, we also prove NP-hardness of $\varepsilon$-approximate $\Phiprojeq(\delta)$-equilibrium ($\Phiinteq(\delta)$-equilibrium) when $\delta = D_{\+X}$ and $\varepsilon = 1$.} In contrast, our hardness results hold for a range of $\delta$ and rule out efficient approximation for any $\varepsilon = o(\delta^2)$, which cannot be directly derived from results that concern the ``global regime'' \citep{daskalakis2021complexity}. Our analysis is also novel and different from \citep{daskalakis2021complexity}: we construct a reduction from the NP-hard maximum clique problem using the Motzkin-Straus Theorem~\citep{motzkin_maxima_1965}. As a byproduct of our analysis, we also obtain new hardness results for computing approximate local maximizer and approximate local Nash equilibrium (\Cref{corollary:local maximizer hardness}). See \Cref{fig:complexity} for a summary.

\begin{figure}[ht]
    \centering
    \includegraphics[width=.99\textwidth]{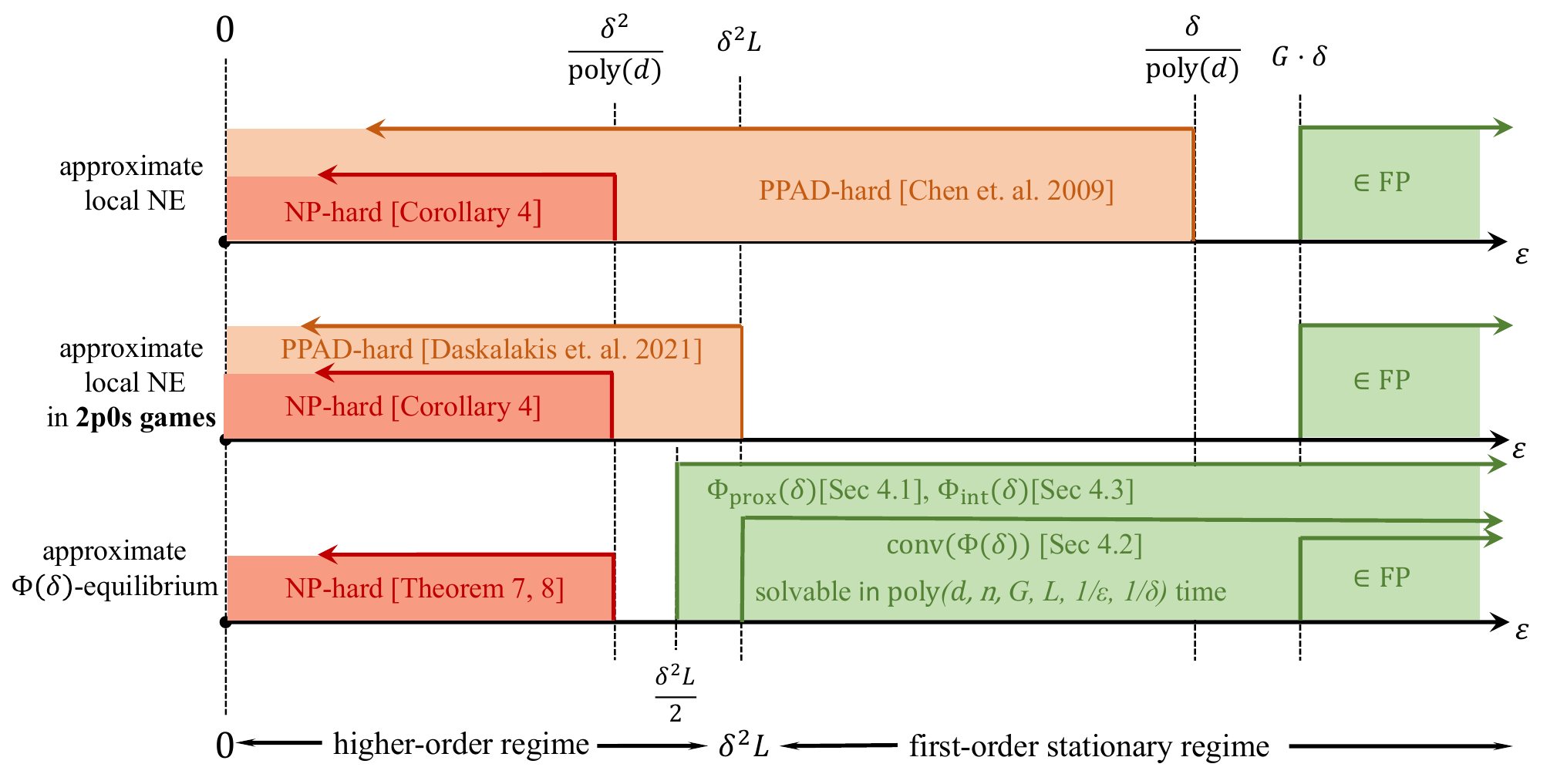}
    \caption{Complexity of computing an $\varepsilon$-approximate $\delta$-local Nash equilibrium and $\varepsilon$-approximate $\Phi(\delta)$-equilibrium in $G$-Lipschitz and $L$-smooth $d$-dimensional games. We consider cases where $G, L= O(\poly(d))$.
    The regime $\varepsilon \ge G\delta$ is trivial since the game is $G$-Lipschitz. The PPAD-hardness of approximate local Nash equilibrium follows from approximate (global) Nash equilibrium in bimatrix games due to linearity of the utility function~\citep{chen2009settling}. The PPAD-hardness of approximate local Nash equilibrium in two-player zero-sum games is proved in~\citep{daskalakis2021complexity}. The NP-hardness of $\varepsilon$-approximate $\Phi(\delta)$-equilibrium is proven for $\Phi_{\mathrm{All}}(\delta)$ (\Cref{thm:hardnessFOSall swap}) and $\Phi_{\mathrm{Int}^+}(\delta)$ (\Cref{thm:hardnessFOS restricted}) in \Cref{sec:first-order regime hardness}. The NP-hardness of $\varepsilon$-approximate $\delta$-local Nash equilibrium is implied by \Cref{corollary:local maximizer hardness}. The positive results for $\varepsilon$-approximate $\Phi(\delta)$-equilibrium in the first-order stationary regime hold for $\Phiproxeq(\delta)$ (\Cref{sec:proximal regret}),  $\Phiinteq(\delta)$ (\Cref{sec:phi-int-regret minization}), and $\Conv(\Phi(\delta))$ when $|\Phi(\delta)|$ is finite (\Cref{sec:convex-phi}).}
    \label{fig:complexity}
\end{figure}

Further related work is discussed in \Cref{app:related works}.

\section{Preliminaries}
We use $\|\cdot\|$ for $\ell_2$-norm throughout and specify other norms with subscripts. A ball of radius $r > 0$ centered at $x \in \-R^d$ is denoted by $B_d(x, r) := \{ x' \in \R^d: \InNorms{x - x'} \le r \}$.  We also write $B_d(\delta)$ for a ball centered at the origin with radius $\delta$. For $a \in \R$, we use $[a]^+$ to denote $\max\{0, a\}$. For a set $\+X$, we use $\Delta(\+X)$ to denote the set of distributions over $\+X$. For a set $\+X \subseteq \-R^d$,  we use $\Dx = \max_{x, x'}\InNorms{x-x'}$ to denote its diameter.

\paragraph{Games} An $n$-player game has a set of $n$ players $[n] := \{1,2,\ldots, n\}$. Each player $i$ has a strategy set $\+X_i$, and we use $x_i \in \+X_i$ to denote one strategy. We note that the strategy set could be discrete or continuous. We denote a joint strategy profile as $x = (x_i, x_{-i}) \in \Pi_{j=1}^n \+X_j$. The utility of each player $i$ is determined by a utility function $u_i: \Pi_{j=1}^n \+X_j \rightarrow [0,1]$. We denote a game instance as $\+G = \{ [n], \{\+X_i\}, \{u_i\} \}$. 

\paragraph{$\Phi$-\lce}
We introduce the concept of $\Phi$-\lce and its relationship with online learning and $\Phi$-regret minimization. For a strategy set $\+Y$, we use $\Phi^\+Y$ to denote a set of \emph{strategy modifications}, i.e., mappings from $\+Y$ to itself.

\begin{definition}[$\Phi$-\lce~\citep{greenwald2003general,stoltz2007learning}]\label{def:local CE} 
In a game $\+G = \{ [n], \{\+X_i\}, \{u_i\} \}$, let $\Phi = \Pi_{i=1}^n\Phi^{\+X_i}$ be a profile of strategy modifications. A distribution over joint strategy profiles $\sigma \in \Delta(\Pi_{i=1}^n \X_i)$ is an $\varepsilon$-approximate $\Phi$-\lce for $\varepsilon \ge 0$  if and only if for all player $i \in [n]$, 
\[\max_{\phi\in \Phi^{\+X_i}} 
\-E_{x \sim \sigma} \InBrackets{u_i(\phi(x_i), x_{-i})} \le \-E_{x \sim \sigma} \InBrackets{u_i(x)} + \varepsilon.\]
When $\varepsilon = 0$, we call $\sigma$ a $\Phi$-equilibrium.
\end{definition}

\paragraph{Online Learning and $\Phi$-regret} We consider the standard online learning setting: at each day $t \in [T]$, the learner chooses an action $x^t$ from an action set $\X$ and the adversary chooses a loss function $f^t: \X \rightarrow \R$, then the learner suffers a loss $f^t(x^t)$ and receives feedback. 
The classic goal of an online learning algorithm is to minimize the \emph{external regret} defined as $
\reg^T :=\max_{x\in \X}\sum_{t=1}^T (f^t(x^t) - f^t(x)).$ 
An algorithm is called \emph{no-regret} if its external regret is sublinear in $T$. The notion of $\Phi$-regret generalizes external regret by allowing more general strategy modifications.
\begin{definition}[$\Phi$-regret]
\label{def:phi-regret}
    Let $\Phi$ be a set of strategy modification functions $\{\phi: \X \rightarrow \X\}$. For $T \ge 1$, the $\Phi$-regret of an online learning algorithm is 
        $\reg_\Phi^T := \max_{\phi \in \Phi} \sum_{t=1}^T \InParentheses{ f^t(x^t) -  f^t(\phi(x^t))}.$
     An algorithm is called \emph{no $\Phi$-regret} if its $\Phi$-regret is sublinear in $T$.
\end{definition}
The framework of $\Phi$-regret generalizes many classical notions of regret. For example, the external regret is $\Phi_{\mathrm{ext}}$-regret where $\Phi_{\mathrm{ext}}$ contains all constant strategy modifications $\phi_{x^*}(x) = x^*$ for all $x^*\in \X$. The classic notion of \emph{swap regret}~\citep{blum_external_2007} defined on the simplex $\Delta([m])$, corresponds to $\Phi_{\mathrm{swap}}$-regret where $\Phi_{\mathrm{swap}}$ contains all linear transformations. In this paper, we also use swap regret to refer to the more general notion of $\Phi_{\mathrm{All}}$-regret where $\Phi_{\mathrm{All}}$ contains all the strategy modifications over any strategy set $\+X$. 

A fundamental result for learning in games is that when each player $i$ employs a no-$\Phi_i$-regret learning algorithm in a game, then the empirical distribution of their strategy profiles converges to the set of approximate $\Phi$-equilibria~\citep{greenwald2003general}. As a result, computing an approximate $\Phi$-equilibrium often reduces to designing a no $\Phi$-regret online learning algorithm. 
\begin{theorem}[\citep{greenwald2003general}]
\label{theorem: no-phi-regret-2-phi-eq}
    If each player $i \in [n]$ has $\Phi_i$-regret that is upper bounded by $\reg^T_{\Phi_i}$, then their empirical distribution of strategy profiles played is an $(\max_{i\in [n]} \reg^T_{\Phi_i} / T)$-approximate $\Phi$-equilibrium.
\end{theorem}

\section{Tractable $\Phi$-Equilibrium for Finite $\Phi$ via Sampling}
\label{sec: finite-phi}
In this section, we revisit the problem of computing and learning an $\Phi$-equilibrium in non-concave games when each player's set of strategy modifications  $\Phi^{\X_i}$ is finite. 

The pioneering work of \citet{stoltz2007learning} gives a no-$\Phi$-regret algorithm for this case where each player chooses a distribution over strategies in each round. This result also implies convergence to $\Phi$-equilibrium. However, the algorithm by \citet{stoltz2007learning} is not computationally efficient. In each iteration, their algorithm requires computing a distribution that is stationary under a transformation that can be represented as a mixture of the modifications in $\Phi$. The existence of such a stationary distribution is guaranteed by the Schauder-Cauty fixed-point theorem~\cite{cauty_solution_2001}, but the distribution might require exponential support and be intractable to find.

Our main result in this section is an efficient $\Phi$-regret minimization algorithm (\Cref{alg:phi-reg}) that circumvents the step of the exact computation of a stationary distribution. Consequently, our algorithm also ensures efficient convergence to a $\Phi$-equilibrium when adopted by all players. 

\begin{algorithm}[!ht]
    \KwIn{$x_{\text{root}} \in \X$, $h \ge 2$, an external regret minimization algorithm $\mathfrak{R}_{\Phi}$ over $\Phi$} 
    \KwOut{A $\Phi$-regret minimization algorithm for $\X$}
    \caption{$\Phi$-regret minimization for non-concave reward via sampling}
    \label{alg:phi-reg}
    \Fn{\textsc{NextStrategy()}}{
    $p^t \leftarrow$ $\mathfrak{R}_\Phi$.\textsc{NextStrategy}(). Note that $p^t$ is a distribution over $\Phi$.\\
    \textbf{return} $x^t \leftarrow$ \textsc{SampleStrategy}($x_{\text{root}}, h, p^t$) using \Cref{alg:sample}.}
    \Fn{\textsc{ObserveReward}$(u^t(\cdot))$}{
    Set $u^t_{\Phi}(\phi) = u^t(\phi(x^t))$ for all $\phi \in \Phi$. \\
    $\mathfrak{R}_{\Phi}$.\textsc{ObserveReward}($u^t_\Phi(\cdot)$).
    }
\end{algorithm}
\begin{algorithm}[!ht]
    \KwIn{$x_{\text{root}} \in \X$, $h \ge 2$, $p^t \in \Delta(\Phi)$ } 
    \KwOut{$x \in \X$}
    \caption{\textsc{SampleStrategy}}
    \label{alg:sample}
    $x_1 \leftarrow x_{\text{root}}$.\\
    \For{$2 \le k \le h$}{
    $\phi \leftarrow$ sample form $\Phi$ according to $p^t$. \\
    $x_{k} = \phi(x_{k-1})$.
    }
    \textbf{return} $x$ from $\{x_1, \ldots, x_h\}$ uniformly at random.
\end{algorithm}

\begin{theorem}
\label{theorem:finite-phi-regret} Let $\X$ be an arbitrary set, $\Phi$ be an arbitrary finite set of strategy modification functions for $\X$, and $u^1(\cdot),\ldots, u^T(\cdot)$ be an arbitrary sequence of possibly non-concave reward functions from $\X$ to $[0,1]$. If we instantiate \Cref{alg:phi-reg} with the Hedge algorithm as the regret minimization algorithm $\mathfrak{R}_{\Phi}$ over $\Phi$ and $h=\sqrt{T}$, the algorithm guarantees that, with probability at least $1 - \beta$, it produces a sequence of strategies $x^1, \ldots, x^T$ with $\Phi$-regret at most 
\[
\max_{\phi \in \Phi} \sum_{t=1}^T u^t(\phi(x^t)) - \sum_{t=1}^T u^t(x^t) \le 8\sqrt{T (\log |\Phi| + \log(1/\beta))}.
\]
Moreover, the algorithm runs in time $O(\sqrt{T}|\Phi|)$ per iteration.

If all players in a non-concave continuous game employ \Cref{alg:phi-reg}, then with probability at least $1-\beta$, for any $\varepsilon > 0$, the empirical distribution of strategy profiles played forms an $\varepsilon$-approximate $\Phi = \Pi_{i=1}^n \Phi^{\mathcal{X}_i}$-equilibrium, after $\poly\left(\frac{1}{\varepsilon}, \log \left(\max_i |\Phi^{\mathcal{X}_i}|\right), \log \frac{n}{\beta}\right)$ iterations.
\end{theorem}

\begin{remark}
    We remark that the regret bound $O(\sqrt{T \log |\Phi|})$ is tight: for the $K$-expert problem where $\+X = \{1, 2, \ldots, K\}$ and $\Phi = \Phi_{\mathrm{ext}}$ contains all constant strategy modifications, since $|\Phi|=K$, our algorithm gives $O(\sqrt{T\log K})$ regret. This matches the lower bound $\Omega(\sqrt{T\log K})$ on the external regret~\citep{cesa2006prediction}.
\end{remark}

\paragraph{High-level ideas} We adopt the framework in \citep{stoltz2007learning}. The framework contains two steps in each iteration $t$: (1) the learner runs a no-external-regret algorithm over $\Phi$ which outputs $p^t \in \Delta(\Phi)$ in each iteration $t$; (2) the learner chooses a stationary distribution $\mu^t = \sum_{\phi \in \Phi} p^t \phi(\mu^t)$, where we slightly abuse notation to use $\phi(\mu^t)$ to denote the image measure of $\mu$ by $\phi$. The idea is we can decompose the $\Phi$-regret into two parts:
\begin{align*}
     \reg^T_{\Phi}= \underbrace{\max_{\phi \in \Phi}\left\{ \sum_{t=1}^T u^t(\phi(x^t)) - \-E_{\phi' \sim p^t} \InBrackets{ u^t(\phi'(x^t)) } \right\}}_{\text{I: external regret over $\Phi$}} +  \underbrace{\sum_{t=1}^T \-E_{\phi' \sim p^t} \InBrackets{ u^t(\phi'(x^t)) } - u^t(x^t)}_{\text{II: approximation error of stationary distribution}}.
\end{align*}
The first term can be bounded by the external regret over $\Phi$. If we play an action $x^t \sim \mu^t$ sampled from a stationary distribution $\mu^t$, then in expectation, the second term would be zero. However, how to compute the stationary distribution $\mu^t$ efficiently is unclear. We essentially provide a computationally efficient way to carry out step (2) without computing this stationary distribution.
\begin{itemize}[leftmargin=*]
    \item We first construct an $\varepsilon$-approximate stationary distribution by recursively applying strategy modifications from $\Phi$. The constructed distribution can be viewed as a tree. Our construction is inspired by the recent work of~\citet{zhang2024efficient} for concave games. The main difference here is that for non-concave games, the distribution needs to be approximately stationary with respect to a \emph{mixture} of strategy modifications rather than a single one as in concave games. Consequently, this leads to an approximate stationary distribution with prohibitively high support size $(|\Phi|)^{\sqrt{T}}$, as opposed to $\sqrt{T}$ in \citep{zhang2024efficient} for concave games.
    \item Despite the exponentially large support size of the distribution, we utilize its tree structure to design a simple and efficient sampling procedure that runs in time $\sqrt{T}$. Equipped with such a sampling procedure, we provide an efficient randomized algorithm that generates a sequence of strategies so that, with high probability, the $\Phi$-regret for this sequence of strategies is at most $O(\sqrt{T\log|\Phi|})$.
\end{itemize}

Before we present the proof of \Cref{theorem:finite-phi-regret} in \Cref{sec:proof of finite-phi-regret}, we remark that an extension of \Cref{theorem:finite-phi-regret} to infinite $\Phi$ holds when the rewards $\{u^t\}_{t \in [T]}$ are $G$-Lipschitz and $\Phi$ admits an $\alpha$-cover with size $N(\alpha)$. In particular, when $\Phi$ is the set of all $M$-Lipschitz functions over $[0,1]^{d}$, $\Phi$ admits an $\alpha$-cover with $\log N(\alpha)$ of the order $(1/\alpha)^{d}$~\citep{stoltz2007learning}. In this case, we can run our algorithm over the finite cover of $\Phi$ and get 
 \begin{corollary}\label{cor:infinite global deviations via discretization}
           There is a randomized algorithm such that, with probability at least $1-\beta$, the $\Phi$-regret is bounded by $c \cdot T^{\frac{d+1}{d+2}} \cdot \log(1/\beta)$, where $c$ only depends on $G$ and $M$. The per-iteration complexity is $O(\sqrt{T} N(T^{-\frac{1}{d+2}})) = O(\sqrt{T}\exp(T^{\frac{d}{d+2}}))$.
 \end{corollary}

\subsection{Proof of \Cref{theorem:finite-phi-regret}}
\label{sec:proof of finite-phi-regret}
For a distribution $\mu \in \Delta(\X)$ over strategy space $\X$, we slightly abuse notation and define its expected utility 
as
\begin{align*}
    u^t(\mu) := \-E_{x \sim \mu} \InBrackets{u^t(x)} \in [0,1].
\end{align*}
We define $\phi(\mu)$ the image measure of $\mu$ under transformation $\phi$. In each iteration $t$, the learner chooses their strategy $x^t \in \+X$ according to the distribution $\mu^t$. For a sequence of strategies $\{x^t\}_{t\in[T]}$, the $\Phi$-regret is
\begin{align*}
    \reg^T_{\Phi}:=  \max_{\phi \in \Phi}\left\{ \sum_{t=1}^T \InParentheses{ u^t(\phi(x^t)) - u^t(x^t)} \right\}.
\end{align*}
Algorithm~\ref{alg:phi-reg} uses an external regret minimization algorithm $R_\Phi$ over $\Phi$ which outputs a distribution $p^t \in \Delta(\Phi)$. We can then decompose the $\Phi$-regret into two parts.
\begin{align*}
     \reg^T_{\Phi}= \underbrace{\max_{\phi \in \Phi}\left\{ \sum_{t=1}^T u^t(\phi(x^t)) - \-E_{\phi' \sim p^t} \InBrackets{ u^t(\phi'(x^t)) } \right\}}_{\text{I: external regret over $\Phi$}} +  \underbrace{\sum_{t=1}^T \-E_{\phi' \sim p^t} \InBrackets{ u^t(\phi'(x^t)) } - u^t(x^t)}_{\text{II: approximation error of stationary distribution}}.
\end{align*}
\paragraph{I: Bounding the external regret over $\Phi$.} The external regret over $\Phi$ can be bounded directly. This is equivalent to an online expert problem: in each iteration $t$, the external regret minimizer $\mathfrak{R}_{\Phi}$ chooses $p^t \in \Delta(\Phi)$ and the adversary then determines the utility of each expert $\phi \in \Phi$ as $u^t(\phi(x^t))$. We choose the external regret minimizer $\mathfrak{R}_{\Phi}$ to be the Hedge algorithm~\citep{freund_adaptive_1999}, which gives
\begin{align}\label{eq: external regret}
    \max_{\phi \in \Phi}\left\{ \sum_{t=1}^T u^t(\phi(x^t)) - \-E_{\phi' \sim p^t} \InBrackets{ u^t(\phi'(x^t)) } \right\} \le 2\sqrt{T \log |\Phi|},
\end{align}
where we use the fact that the utility function $u^t$ is bounded in $[0,1]$.

\paragraph{II: Bounding error due to sampling from an approximately stationary distribution.} We first construct an approximately stationary distribution $\mu^t$ using a complete $|\Phi|$-ary tree with depth $h$. The root of this tree is an arbitrary strategy $x_{\text{root}} \in \mathcal{X}$. Each internal node $x$ has exactly $|\Phi|$ children, denoted as $\{\phi(x)\}_{\phi \in \Phi}$. The distribution $\mu^t$ is supported on the nodes of this tree. Next, we define the probability for each node under the distribution $\mu^t$.
The root node $x_{\text{root}}$ receives probability $\frac{1}{h}$. The probability of other nodes is defined in a recursive manner. For every node $x = \phi(x_p)$ where $x_p$ is its parent, $x$ receives probability $\Pr_{\mu^t}[x] = \Pr_{\mu^t}[x_p] \cdot p^t(\phi)$. It is then clear that the total probability of the children of a node $x_p$ is exactly $\Pr_{\mu^t}[x~\text{is $x_p$'s child}] = \sum_{\phi} \Pr[x_p] \cdot p^t(\phi) = \Pr[x_p]$. Denote the set of nodes in depth $k$ as $N_k$. We have $\Pr_{\mu^t}[x\in N_k] = \frac{1}{h}$ for every depth $1 \le k \le h$. Thus the distribution $\mu^t$ supports on $\frac{|\Phi|^h-1}{|\Phi| -1}$ points and is well-defined. By the construction above, we know $x^t$ output by \Cref{alg:sample} is a sample from $\mu^t$.
\begin{claim}
    \Cref{alg:sample} generates a sample from the distribution $\mu^t$.
\end{claim}

Now we show that the approximation error of $\mu^t$ is bounded by $O(\frac{1}{h})$. 
We can evaluate the approximation error of $\mu^t$:
\begin{align*}
    &\-E_{\phi \sim p^t} \InBrackets{ u^t(\phi(\mu^t)) } - u^t(\mu^t) \\
    &=\-E_{\phi \sim p^t} \InBrackets{ 
    \sum_{k=1}^h \sum_{x \in N_k} \Pr_{\mu^t}[x] u^t(\phi(x)) } - \InBrackets{ 
    \sum_{k=1}^h \sum_{x \in N_k} \Pr_{\mu^t}[x] u^t(x)} \\
    &= \sum_{k=1}^h \sum_{x\in N_k} \InParentheses{ \-E_{\phi \sim p^t} \InBrackets{  \Pr_{\mu^t}[x]  u^t(\phi(x))} -  \Pr_{\mu^t}[x] u^t(x)    }.
\end{align*}
We recall that for a node $x = \phi(x_p)$ with $x_p$ being its parent, we have $\Pr_{\mu^t}[x] = \Pr_{\mu^t}[x_p] \cdot p^t(\phi)$. Thus for any $1 \le k \le h-1$, we have
\begin{align*}
    &\sum_{x \in N_k}\InParentheses{\-E_{\phi \sim p^t} \InBrackets{  \Pr_{\mu^t}[x]  u^t(\phi(x))} -  \Pr_{\mu^t}[x] u^t(x)} \\
    & = \sum_{x \in N_k}\InParentheses{ \sum_{\phi \in \Phi} p^t(\phi) \Pr_{\mu^t}[x]  u^t(\phi(x)) -  \Pr_{\mu^t}[x] u^t(x)} \\
    & = \sum_{x \in N_{k+1}} \Pr_{\mu^t}[x] u^t(x) - \sum_{x \in N_{k}} \Pr_{\mu^t}[x] u^t(x).
\end{align*}
Using the above equality, we get
\begin{align*}
    &\-E_{\phi \sim p^t} \InBrackets{ u^t(\phi(\mu^t)) } - u^t(\mu^t) \\
    &= \sum_{k=1}^{h-1} \sum_{x \in N_{k+1}} \Pr_{\mu^t}[x] u^t(x)+\sum_{x\in N_h}\sum_{\phi\in \Phi}p^t(\phi)\Pr_{\mu^t}[x]u^t(\phi(x)) - \sum_{k=2}^h \sum_{x \in N_{k}} \Pr_{\mu^t}[x] u^t(x) -\Pr_{\mu^t}[x_{\text{root}}]u^t(x_{\text{root}})\\
    &= \sum_{x\in N_h}\sum_{\phi\in \Phi}p^t(\phi)\Pr_{\mu^t}[x]u^t(\phi(x)) - \Pr_{\mu^t}[x_{\text{root}}]u^t(x_{\text{root}}) \\
    &\le \frac{1}{h},
\end{align*}
where in the last inequality we use the fact that $\sum_{x \in N_k} \Pr_{\mu^t}[x] = \frac{1}{h}$ for all $1 \le k \le h$ and the utility function $u^t$ is bounded in $[0,1]$. Therefore, for $x^t \sim \mu^t$, the sequence of random variables 
\begin{align*}
    \sum_{t=1}^\tau \InParentheses{\-E_{\phi \sim p^t} \InBrackets{ u^t(\phi(x^t)) } - u^t(x^t) - \frac{1}{h}}, \text{for~}\tau \ge 1.
\end{align*}
is a super-martingale. Thanks to the boundedness of the utility function, we can apply the Hoeffding-Azuma Inequality and get for any $\varepsilon > 0$.
\begin{align}\label{eq:approx fix distribution}
    \Pr\InBrackets{\sum_{t=1}^T \InParentheses{\-E_{\phi \sim p^t} \InBrackets{ u^t(\phi(x^t)) } - u^t(x^t) - \frac{1}{h}} \ge \varepsilon} \le \exp\InParentheses{ - \frac{\varepsilon^2}{8T}}. 
\end{align}
Combining \eqref{eq: external regret} and \eqref{eq:approx fix distribution} with $\varepsilon = \sqrt{8 T \log (1/\beta)}$ and $h = \sqrt{T}$, we have that, with probability $1 - \beta$, that
\begin{align*}
    \reg^T_{\Phi}&= \max_{\phi \in \Phi}\left\{ \sum_{t=1}^T u^t(\phi(x^t)) - \-E_{\phi' \sim p^t} \InBrackets{ u^t(\phi'(x^t)) } \right\} +  \sum_{t=1}^T \-E_{\phi' \sim p^t} \InBrackets{ u^t(\phi'(x^t)) } - u^t(x^t) \\
    &\le 2\sqrt{T\log |\Phi|} + \frac{T}{h} + \sqrt{8 T \log (1/\beta)} \\
    &\le 8\sqrt{T (\log |\Phi| + \log(1/\beta))}.
\end{align*}

\paragraph{Convergence to $\Phi$-equilibrium} If all players in a non-concave continuous game employ \Cref{alg:phi-reg}, then we know for each player $i$,  with probability $1 - \frac{\beta}{n}$, its $\Phi^{\+X_i}$-regret is upper bounded by
\begin{align*}
    8\sqrt{T (\log |\Phi^{\+X_i}| + \log(n/\beta))}.
\end{align*}
By a union bound over all $n$ players, we get with probability $1 - \beta$, every player $i$'s $\Phi^{\+X_i}$-regret is upper bounded by $8\sqrt{T (\log |\Phi^{\+X_i}| + \log(n/\beta))}$. Now by \Cref{theorem: no-phi-regret-2-phi-eq}, we know the empirical distribution of strategy profiles played forms an $\varepsilon$-approximate $\Phi = \Pi_{i=1}^n \Phi^{\mathcal{X}_i}$-equilibrium, as long as $T \ge  \frac{64(\log |\Phi^{\+X_i}| + \log(n/\beta))}{\varepsilon^2}$ iterations. This completes the proof of \Cref{theorem:finite-phi-regret}

\section{Approximate $\Phi$-Equilibria under  Infinite Local Strategy Modifications}
This section studies $\Phi$-equilibrium when $|\Phi|$ is infinite. We focus on \emph{continuous games} where each agent has infinite many strategies. 
We first formally define the class of games of interest. 
\paragraph{Continuous / Smooth Games.} 
An $n$-player \emph{continuous game} has a set of $n$ players $[n]:= \{1,2, \ldots, n\}$. Each player $i \in [n]$ has a nonempty convex and compact strategy set $\X_i \subseteq \R^{d_i}$.  For a joint strategy profile $x = (x_i, x_{-i}) \in \prod_{j=1}^n \X_j$, the reward of player $i$ is determined by a utility function $u_i: \prod_{j=1}^n \X_j \rightarrow [0,1]$. We denote by $d = \sum_{i=1}^n d_i$ the dimensionality of the game and assume $\max_{i\in[n]}\{D_{\X_i}\} \le D$. Examples of continuous games include mixed strategy extensions of normal-form games where $\+X_i = \Delta^{d_i}$, nonconvex-nonconcave minimax optimization, and many other machine learning applications with high-dimensional strategies. When each player's utility is further differentiable, we also call the game a \emph{differentiable game}.

A \emph{smooth game} is a continuous game whose utility functions further satisfy the following assumption.

\begin{assumption}[Smooth Games]
\label{assumption:smooth games}
    The utility function $u_i(x_i, x_{-i})$ for any player $i \in [n]$ is differentiable and satisfies
    \begin{itemize}[leftmargin=*]
        \item[1.] ($G$-Lipschitzness) $\InNorms{\nabla_{x_i} u_i(x)} \le G$ for all $i$ and  $x\in\prod_{j=1}^n \X_j$;  
        \item[2.] ($L$-smoothness) there exists $L_i > 0$ such that $\InNorms{\nabla_{x_i} u_i(x_i, x_{-i}) - \nabla_{x_i} u_i(x_i', x_{-i})} \le  L_i\InNorms{x_i  - x_i'} $ for all $x_i, x_i' \in \X_i$ and $x_{-i} \in \Pi_{j\ne i}\X_j$. We denote $L = \max_{i}L_i$ as the smoothness of the game.
    \end{itemize}
\end{assumption}

Crucially, we make no assumption on the concavity of $u_i(x_i, x_{-i})$. It is, in general, computationally hard to compute a $\Phi$-equilibrium even if $\Phi$ contains all constant deviations since it would imply finding the global optimum of a nonconcave function. Instead, we focus on $\Phi$ that consists solely of local strategy modifications.  

\begin{definition}[$\delta$-local strategy modification]
    For each agent $i$, we call a set of strategy modifications $\Phi^{\X_i}$ \emph{$\delta$-local} if for all $x \in \X_i$ and $\phi_i \in \Phi^{\X_i}$, $\InNorms{\phi_i(x)- x}\le \delta$. We use notation $\Phi^{\X_i}(\delta)$ to denote a $\delta$-local strategy modification set for agent $i$. We also use $\Phi(\delta) = \Pi_{i=1}^n \Phi^{\X_i}(\delta)$ to denote a profile of $\delta$-local strategy modification sets.
\end{definition}

\paragraph{A Reduction in the First-Order Stationary Regime} Below we present a simple yet useful reduction from computing an $\varepsilon$-approximate $\Phi(\delta)$-\lce in \emph{non-concave} smooth games to $\Phi^{\X_i}(\delta)$-regret minimization against \emph{linear} losses for any $\varepsilon \ge \frac{\delta^2 L}{2}$. The key observation here is that the $L$-smoothness of the utility function permits the approximation of a non-concave function with a linear function within a $\delta$-neighborhood. This approximation yields an error of at most $\frac{\delta^2L}{2}$ that lies in the first-order stationary regime. Formally, by $L$-smoothness of the utility function $u_i$, we have
\begin{align*}
    u_i(\phi(x_i), x_{-i}) - u_i(x) \le \InAngles{\nabla_{x_i} u_i(x),\phi(x_i) -  x_i } + \frac{L}{2}\InNorms{\phi(x_i) -x_i}^2.
\end{align*}
As a result, computing an $(\varepsilon+\frac{\delta^2L}{2})$-approximate $\Phi(\delta)$-equilibrium reduces to computing a joint strategy distribution $\sigma \in \Delta(\Pi_{j=1}^n \+X_j)$ such that for every player $i$\footnote{We note that we may define local $\Phi$-equilibrium directly using \eqref{eq:first-order phi-equilibrium} by the reduction, since any joint distribution satisfies \eqref{eq:first-order phi-equilibrium} is also an approximate $\Phi(\delta)$-equilibrium with error $\varepsilon+\frac{\delta^2L}{2}$ in $L$-smooth games. Several follow-up works study local $\Phi$-equilibrium notions using similar first-order approximation ideas~\citep{ahunbay2024first,zhang2025expected}.},
\begin{align}\label{eq:first-order phi-equilibrium}
    \max_{\phi \in \Phi^{\+X_i}(\delta)}\-E_{x \sim \sigma} \InBrackets{\InAngles{\nabla_{x_i} u_i(x), \phi(x_i) -  x_i} }  \le \varepsilon. 
\end{align}
This also implies when every agent runs an no $\Phi^{\+X_i}(\delta)$ algorithm in a smooth game against the linearized reward $\InAngles{\nabla_{x_i}u_i(x),\cdot}$, the empirical distribution of their played strategy profiles is a $(\nicefrac{\max_{i\in[n]}\{\reg_{\Phi^{\X_i}(\delta)}^T\}}{T}+\frac{\delta^2L}{2})$-approximate $\Phi(\delta)$-equilibrium. We summarize the results in the following lemma.

\begin{lemma}[From Non-Concave to Linear Losses in First-Order Stationary Regime] 
\label{lemma:no-regret-2-CE}
Let $\mathcal{G} = \{[n], \{\+X_i\}, \{u_i\}\}$ be a $L$-smooth game. Then

1. A distribution $\sigma \in \Delta(\Pi_{j=1}^n \+X_j)$ is an $(\varepsilon+\frac{\delta^2L}{2})$-approximate $\Phi(\delta)$-equilibrium if 
\begin{align*}
    \-E_{x \sim \sigma} \InBrackets{\InAngles{\nabla_{x_i} u_i(x), \phi(x_i) -  x_i }}  \le \varepsilon, \forall i\in [n], \forall \phi \in \Phi^{\+X_i}(\delta).
\end{align*}

2. For any $T \ge 1$ and $\delta > 0$, let $\+A$ be an algorithm that guarantees to achieve no more than $\reg_{\Phi^{\X_i}(\delta)}^T$ $\Phi^{\X_i}(\delta)$-regret for \emph{linear} loss functions for each agent $i\in[n]$. Then when every agent employs $\+A$ in a $L$-smooth game, their empirical distribution of the joint strategies played converges to a $( \nicefrac{\max_{i\in[n]}\{\reg_{\Phi^{\X_i}(\delta)}^T\}}{T} + \frac{\delta^2L}{2})$-approximate $ \Phi(\delta)$-equilibrium after $T$ iterations.
\end{lemma}

In subsequent sections,  we introduce several natural classes of local strategy modifications and provide efficient online learning algorithms that converge to $\varepsilon$-approximate $\Phi$-equilibrium in the first-order stationary regime where $\varepsilon = \Omega(\delta^2 L)$. These approximate $\Phi$-equilibria guarantee first-order stability. On the other hand, in \Cref{sec:first-order regime hardness}, we show that a better approximation beyond the first-order stationary regime such as $\varepsilon = o(\delta^2)$ is impossible unless P = NP.

\subsection{Proximal Regret: A Tractable Notion Between External and Swap Regret}
\label{sec:proximal regret}
In this section, we introduce a class of $\Phi$-regret called \emph{proximal regret} parameterized by a convex or smooth function. The main result of this section shows that the classic \hyperref[GD]{Online Gradient Descent} algorithm~\citep{zinkevich2003online} minimizes proximal regret. Since external regret is an instantiation of proximal regret (see discussion below), our result shows that the proximal regret is an efficient regret notion that lies between external regret and swap regret.

\paragraph{Proximal Operator and Proximal Regret} The strategy modification in proximal regret is defined using the proximal operator.
\begin{definition}[Proximal Operator]\label{dfn:proximal operator}
    The \emph{proximal operator} of a  proper functions $f : \+X \rightarrow [-\infty, +\infty]$ is defined as\footnote{A proper function is not identically equal to $+\infty$ over its domain and $-\infty$ is not in its image.}
    \begin{align*}
    \prox_{f}(x) = \argmin_{x' \in \+X} \left\{ f(x') + \frac{1}{2}\InNorms{x'-x}^2 \right\}.
    \end{align*}
\end{definition}
Certain conditions on the function $f$ are required to ensure the proximal operator is well-defined, i.e., the optimization problem has a unique solution. Two sufficient conditions are:
\begin{itemize}
    \item[1.] $f$ is lower semicontinuous and convex. We use $\+F_{\lscc}$ to denote this class of functions.
    \item[2.] $f$ is $1^{-}$-smooth, i.e., $L$-smooth with $L < 1$.  We use $\+F_{\sm}$ to denote this class of functions.
\end{itemize}
The proximal operator has been extensively studied in the optimization literature. We list one example here for illustration and refer the readers to the book \citep{parikh2014proximal} for a comprehensive review. When $f = \mathbb{I}_{\+X}$ is the indicator function for a convex set $\+X \subseteq\-R^d$ (that is, $f(x) = 0$ if $x \in \+X$ and $f(x) = +\infty$ otherwise),  $f$ is convex and $\prox_f(x) = \Pi_{\+X}[x]$ is the Euclidean projection onto $\+X$. Proximal operator naturally induces the following proximal regret notion. 

\begin{definition}[Proximal Regret]\label{dfn:proximal regret}
    The \emph{proximal regret} associated with $f$ is 
    \begin{align*}
    \reg^T_f:= \sum_{t=1}^T \ell^t(x^t) - \ell^t(\prox_f(x^t)).
\end{align*}
We say an algorithm is \emph{no-proximal regret} w.r.t. $\+F$ if $\reg^T_f = o(T)$ for all $f \in \+F$.
\end{definition}
Note that if we choose $\+F_{\mathrm{ext}}$ to be all the indicator functions for one point in $\+X$, then $\max_{f \in \+F_{\mathrm{ext}}}\reg_f^T = \reg^T_{\mathrm{ext}}$ is exactly the external regret. The proximal regret is more general than the external regret as it simultaneously guarantees regret defined by any other convex or smooth function $f \in \+F_{\lscc} \cup \+F_{\sm}$.

\subsubsection{Online Gradient Descent Minimizes Proximal Regret}

\begin{algorithm}[!ht]\label{GD}
    \KwIn{strategy space $\+X$, non-increasing step sizes $\{\eta_t\}$} 
    \KwOut{strategies $\{x^t\}$}
    \caption{Online Gradient Descent (GD)}
    Initialize $x^1 \in \+X$ arbitrarily\\
    \For{$t = 1,2, \ldots,$}{
    play $x^t$ and receive $\nabla \ell^t(x^t)$.\footnotemark\\
    update $x^{t+1} = \Pi_\+X\InBrackets{x^t - \eta_t \nabla \ell^t(x^t)}$.
    }
\end{algorithm}\footnotetext{When the function is non-differentiable, the algorithm receives a subgradient in $\partial \ell^t(x^t)$.}

We show a surprising result that \hyperref[GD]{Online Gradient Descent (GD)} achieves $\reg_f^T = O(\sqrt{T})$ \emph{simultaneously} for all $f \in \+F_{\lscc} \cup \+F_{\sm}$, without any modification. Our result significantly extends the classical result of \hyperref[GD]{GD}~\citep{zinkevich2003online} by showing that \hyperref[GD]{GD} not only minimizes external regret but simultaneously the more general proximal regret. It shows that the sequence produced by \hyperref[GD]{GD} is not only a no-external-regret sequence but has additional features. It also implies that the \hyperref[GD]{GD} dynamics in concave games converges to a refined subset of coarse correlated equilibria, which we may refer to as \emph{proximal correlated equilibria}. 

\begin{theorem}\label{theorem:GD proximal regret}[Online Gradient Descent Minimizes Proximal Regret]
    Let $\+X \subseteq \-R^d$ be a closed convex set and $\{\ell^t: \+X \rightarrow \-R\}$ be a sequence of convex loss functions. 

    1. For any convex function $f \in \+F_{\mathrm{lsc,c}}(\+X)$, denote $p^t = \prox_f(x^t)$, \hyperref[GD]{GD} guarantees for all $T \ge 1$,
    \begin{align*}
        \sum_{t=1}^T \InAngles{\nabla \ell^t(x^t), x^t - \prox_f(x^t)} \le \frac{D^2 + 2B_f - \InNorms{x^{T+1}-p^T}^2}{2\eta_T} + \sum_{t=1}^T\frac{\eta_t}{2}\InNorms{\nabla \ell^t(x^t)}^2 - \sum_{t=1}^{T-1} \frac{1}{2\eta_t} \InNorms{p^t - p^{t+1}}^2,
    \end{align*}
    where $D = \max_{t \in T} \InNorms{x^t-\prox_f(x^t)}$ and $B_f = \max_{t \in [T]}f(\prox_f(x^t)) - \min_{t \in [T]} f(\prox_f(x^t))$. 

    2. For any smooth function $f \in \+F_{\mathrm{sm}}(\+X)$, \hyperref[GD]{GD} guarantees for all $T \ge 1$,
    \begin{align*}
        \sum_{t=1}^T \InAngles{\nabla \ell^t(x^t), x^t - \prox_f(x^t)} \le \frac{D^2 + 2B_f - \InNorms{x^{T+1}-p^T}^2}{2\eta_T} + \sum_{t=1}^T\frac{\eta_t}{2}\InNorms{\nabla \ell^t(x^t)}^2,
    \end{align*}
    where $D = \max_{t \in T} \InNorms{x^t-\prox_f(x^t)}$ and $B_f = \max_{t \in [T]}f(\prox_f(x^t)) - \min_{t \in [T]} f(\prox_f(x^t))$. 

    3. If the step size is constant $\eta_t = \eta$, then the above two bounds hold with $D = \InNorms{x^1 - \prox_f(x^1)}$ and $B_f = f(\prox_f(x^1)) - f(\prox_f(x^T))$.
\end{theorem}
\begin{remark}
    If the set $\+X$ is bounded, we can directly apply the bound with $D$ being the diameter of $\+X$. When the loss functions are $G$-Lipschitz, we can substitute the summation $\sum_{t=1}^T\frac{\eta_t}{2}\InNorms{\nabla \ell^t(x^t)}^2$ with $\frac{G^2}{2} \sum_{t=1}^T \eta_t$.  Without any knowledge of $B_f$, we can choose decreasing step size $\eta_t = \frac{1}{\sqrt{t}}$ or fixed step size $\eta_t = \frac{1}{\sqrt{T}}$, then 
    \begin{align*}
        \reg_f^T \le (D^2 + B_f + G^2) \cdot \sqrt{T}.
    \end{align*}
    If we know $B_f$ and choose the optimized step size $\eta_t = \sqrt{\frac{D^2+2B_f}{G^2 T}}$, then
    \begin{align*}
        \reg_f^T \le G\sqrt{D^2 + 2B_f}\cdot\sqrt{T}.
    \end{align*}
    We note that when $f$ is an indicator function as for the external regret, then $B_f = 0$. Then we recover the optimal $O(GD\sqrt{T})$ bound for the external regret~\citep{zinkevich2003online}.
\end{remark}
\begin{remark}
    Our results also generalize to the Bregman setup where the proximal operator is defined using general Bregman divergence instead of the Euclidean distance. In that case, we show that \hyperref[MD]{Mirror Descent} minimizes the Bregman proximal regret.  See \Cref{sec:bregman proximal} for details.
\end{remark}

\begin{proof}[Proof of \Cref{theorem:GD proximal regret}]
    We define $p^t := \prox_f(x^t)$ for $t \in [1,T]$. By standard analysis of \hyperref[GD]{Online Gradient Descent} (see also the proof of \citealp[Theorem 3.2]{bubeck_convex_2015}), we have
    \begin{align*}
        \sum_{t=1}^T \ell^t(x^t) - \ell^t(p^t) \le  \sum_{t=1}^T \InAngles{\nabla \ell^t(x^t), x^t - p^t}\le \sum_{t=1}^T \frac{1}{2\eta_t} \InParentheses{\InNorms{x^t- p^t}^2 -\InNorms{x^{t+1} - p^t}^2} + \frac{\eta_t}{2}\InNorms{\nabla \ell^t(x^t)}^2.
    \end{align*}
    Since $\eta_t$ is non-increasing, we simplify the above inequality
    \begin{align}
       &\sum_{t=1}^T \InAngles{\nabla \ell^t(x^t), x^t - p^t} \nonumber \\
        &\le \sum_{t=1}^{T} \frac{1}{2\eta_t} \InParentheses{\InNorms{x^{t}- p^{t}}^2 -\InNorms{x^{t+1} - p^t}^2} + \sum_{t=1}^T\frac{\eta_t}{2}\InNorms{\nabla\ell^t(x^t)}^2 \nonumber\\
        &= \frac{\InNorms{x^1-p^1}^2}{2\eta_1} + \sum_{t=1}^{T-1} \InParentheses{ \frac{1}{2\eta_t} \InParentheses{\InNorms{x^{t+1}- p^{t+1}}^2 -\InNorms{x^{t+1} - p^t}^2}  + \InNorms{x^{t+1} - p^{t+1}}^2 \InParentheses{ \frac{1}{2\eta_{t+1}}-\frac{1}{2\eta_t}}  }\nonumber \\
        & \quad -\frac{1}{2\eta_T} \InNorms{x^{T+1} - p^T}^2 + \sum_{t=1}^T\frac{\eta_t}{2}\InNorms{\nabla\ell^t(x^t)}^2 \nonumber \\
        & \le \frac{D^2 - \InNorms{x^{T+1} - p^T}^2}{2\eta_T} + \sum_{t=1}^T\frac{\eta_t}{2}\InNorms{\nabla\ell^t(x^t)}^2 + \sum_{t=1}^{T-1} \frac{1}{2\eta_t} \InParentheses{ \InNorms{x^{t+1} - p^{t+1}}^2 - \InNorms{x^{t+1} - p^t}^2}. \label{eq:gd-1}
    \end{align}
    where in the last step we use $\max_{t \in [T]} \InNorms{x^t -  p^t} \le D$. Also note that for constant step size $\eta_t = \eta$, the above inequality holds with $D = \InNorms{x^1 - p^1}^2$. 
    \paragraph{Key Step}  Now we focus on the term $\InNorms{x^{t+1} - p^{t+1}}^2 - \InNorms{x^{t+1} - p^t}^2$ which does not telescope. We can use \Cref{lemma:telescope} to replace it with an upper bound that telescopes. We present proofs for $f \in \+F_{\mathrm{lsc,c}}$ and $f \in \+F_{\mathrm{smooth}}$ separately.
    
    \paragraph{When $f \in \+F_{\mathrm{lsc,c}}$:} 
    Recall that for any $t$, $p^{t+1} = \prox_f(x^{t+1})$. By item 1 in \Cref{lemma:telescope}, we have
    \begin{align}
        \InNorms{x^{t+1} - p^{t+1}}^2 - \InNorms{x^{t+1} - p^t}^2   \le 2(f(p^t) - f(p^{t+1})) - \InNorms{p^t - p^{t+1}}^2. \label{eq:telescope inequality-1}
    \end{align}
    We can then plug \eqref{eq:telescope inequality-1} back to \eqref{eq:gd-1}. Let $f^* = \min_{t \in [T]} f(p^t)$, we have
    \begin{align*}
        &\sum_{t=1}^T \InAngles{\nabla \ell^t(x^t), x^t - p^t}\\ &\le \frac{D^2 - \InNorms{x^{T+1} - p^T}^2}{2\eta_T} + \sum_{t=1}^T\frac{\eta_t}{2}\InNorms{\nabla\ell^t(x^t)}^2 + \sum_{t=1}^{T-1} \frac{1}{\eta_t} \InParentheses{ f(p^t) - f^* - (f(p^{t+1}) - f^*)} - \sum_{t=1}^{T-1} \frac{1}{2\eta_t} \InNorms{p^t- p^{t+1}}^2 \\
        &= \frac{D^2 - \InNorms{x^{T+1} - p^T}^2}{2\eta_T} + \sum_{t=1}^T\frac{\eta_t}{2}\InNorms{\nabla\ell^t(x^t)}^2 + \sum_{t=1}^{T-1}\InParentheses{ \frac{f(p^t)-f^*}{\eta_t} - \frac{f(p^{t+1})-f^*}{\eta_{t+1}}}  \\
        & \quad \quad + \sum_{t=1}^{T-1} (f(p^{t+1})-f^*)\InParentheses{\frac{1}{\eta_{t+1}} - \frac{1}{\eta_t}} - \sum_{t=1}^{T-1} \frac{1}{2\eta_t} \InNorms{p^t- p^{t+1}}^2\\
        &\le \frac{D^2 + 2B_f - \InNorms{x^{T+1} - p^T}^2}{2\eta_T} + \sum_{t=1}^T\frac{\eta_t}{2}\InNorms{\nabla\ell^t(x^t)}^2 - \sum_{t=1}^{T-1} \frac{1}{2\eta_t} \InNorms{p^t- p^{t+1}}^2,
    \end{align*}
    where in the last inequality, we use $\{\eta_t\}$ is non-increasing again and $f(p^t) -f^*\le B_f$ for all $p^t$ since $B_f  = \max_{t \in [T]} f(p^t) - \min_{t \in [T]} f(p^t)$. 
    Note that when the step size is constant $\eta_t = \eta$, the same inequality holds for $B_f = f(p^1) - f(p^T)$.

    \paragraph{When $f \in \+F_{\mathrm{sm}}$:} Recall that for any $t$, $p^{t+1} = \prox_f(x^{t+1})$. By item 2 in \eqref{lemma:telescope}, we have
    \begin{equation}
        \InNorms{x^{t+1} - p^{t+1}}^2 - \InNorms{x^{t+1} - p^t}^2  \le 2(f(p^t) - f(p^{t+1})). \label{eq:telescope inequality-2}
    \end{equation}
    Now we can plug \eqref{eq:telescope inequality-2} back to \eqref{eq:gd-1} and telescope. By the same analysis above in the convex case, we have
    \begin{align*}
        \sum_{t=1}^T \InAngles{\nabla \ell^t(x^t), x^t - p^t} \le  \frac{D^2 + 2B_f - \InNorms{x^{T+1} - p^T}^2}{2\eta_T} + \sum_{t=1}^T\frac{\eta_t}{2}\InNorms{\nabla\ell^t(x^t)}^2.
    \end{align*}
    Similarly, when the step size is constant $\eta_t = \eta$, the same inequality holds for $B_f = f(p^1) - f(p^T)$.
\end{proof}
The following lemma is used in the proof of \Cref{theorem:GD proximal regret}.
\begin{lemma}[Key Inequality]
\label{lemma:telescope}
Let $x, p \in \+X \subseteq \-R^d$. 

    1. For any convex function $f \in \+F_{\mathrm{lsc,c}}(\+X)$, denote $p_x = \prox_f(x)$, then 
    \begin{align*}
        \InNorms{x - p_x}^2 - \InNorms{x - p}^2  \le 2f(p) -2f(p_x) - \InNorms{p - p_x}^2.
    \end{align*}
    
    2. For any smooth function $f \in \+F_{\mathrm{sm}}(\+X)$, denote $p_x = \prox_f(x)$, then
    \begin{align*}
        \InNorms{x - p_x}^2 - \InNorms{x - p}^2  \le 2f(p) -2f(p_x).
    \end{align*}
\end{lemma}
\begin{proof}
    Since $p_x = \prox_f(x)$, we know there exists a subgradient $v \in \partial f(p_x)$ (if $f$ is differentiable, then $v = \nabla f(p_x)$ is its gradient) such that $\InAngles{p - p_x, x - v - p_x} \le 0$. Then we have 
    \begin{align*}
        &\InNorms{x - p_x}^2 - \InNorms{x - p}^2 \\
        &= 2 \InAngles{p - p_x, x - p_x} - \InNorms{p - p_x}^2 \\
        &= 2 \InAngles{p - p_x, v} + 2\InAngles{p - p_x, x -v - p_x} - \InNorms{p - p_x}^2 \\
        & \le 2 \InAngles{p - p_x, v}- \InNorms{p - p_x}^2.
    \end{align*}
    
    If $f$ is convex, we further have $2 \InAngles{p - p_x, v} \le 2f(p) - 2f(p_x)$, which proves the first inequality. If $f$ is $1$-smooth, we further have $
        2 \InAngles{p - p_x, v}- \InNorms{p - p_x}^2 \le 2f(p) - 2f(p_x)$, which proves the second inequality.
\end{proof}

\subsubsection{Proximal Operator based Local Strategy Modifications}
We consider the set of local strategy modifications based on the proximal operator. Formally, the set $\Phiprox(\delta)$ encompasses all strategy modifications $\prox_f$ that are $\delta$-local.
\begin{align*}
    \Phiprox(\delta) = \{\prox_f: f \in \+F_{\lscc} \cup \+F_{\sm}, \InNorms{\prox_f(x) -x} \le \delta,\forall x \in \+X\}
\end{align*}
We define $\Phiproxeq(\delta) = \Pi_{i=1}^n \Phi^{\+X_i}_{\prox}(\delta)$. By \Cref{lemma:no-regret-2-CE} and \Cref{theorem:GD proximal regret}, we immediately get the following corollary on learning approximate $\Phiproxeq(\delta)$-equilibrium. 
\begin{corollary}
\label{corollary:local proximal equilibrium}
    For any $\delta, \varepsilon > 0$, when each players in a $L$-smooth and $G$-Lipschitz game employ \hyperref[GD]{Online Gradient Descent}, their empirical distribution of played strategy profiles converges to an $(\varepsilon+\frac{L\delta^2}{2})$-approximate $\Phiproxeq(\delta)$-equilibrium in $O(1/\varepsilon^2)$ iterations.
\end{corollary}

\subsubsection{Examples of Proximal Regret}
We have shown that external regret is an instantiation of the proximal regret with indicator functions for single points. In this subsection, we give several other instantiations of the proximal regret to illustrate its generality.
\begin{example}[Projection-based deviations]\label{ex:projection}
    Let $f = \InAngles{v,x}$ be a linear function. Then $\prox_f(x) = \Pi_\+X[x - v]$ for all $x \in \+X$. This deviation adds a fixed displacement vector $v$ to the input strategy and projects back to the feasible set if necessary. Let $\Phiproj:=\{\prox_f: f(x) = \InAngles{v, x}\}$ and $\Phiproj(\delta)$ be the corresponding subset of $\delta$-local strategy modifications (all the functions $\InAngles{v,x}$ with $\InNorms{v}\le \delta$). We also denote $\regproj$ the corresponding $\Phiproj(\delta)$-regret. In an earlier version of the paper, we proved $O(\sqrt{T})$ regret for this class of ``projection-based strategy modifications", which is now implied by the more general \Cref{theorem:GD proximal regret}. 
    
    Moreover, we provide examples in \Cref{sec:diff external-proj} showing that the external regret is incomparable with $\Phiproj(\delta)$-regret, which implies that the proximal regret is strictly more general than the external regret. Also, in \Cref{sec:lower bound proj regret}, we give two lower bounds showing that (1) $\regproj^T =\Omega(\sqrt{\delta T})$ against linear losses for all algorithms; (2) $\regproj^T =\Omega(\delta^2LT)$ against $L$-smooth non-concave losses for all algorithms that satisfy the ``linear span" assumption which includes \hyperref[GD]{GD} and \ref{OG}. Since $\Phiproj(\delta)$-regret is a special case of the proximal regret, these lower bounds shows the near-optimality of our upper bounds in \Cref{theorem:GD proximal regret} and \Cref{corollary:local proximal equilibrium}, respectively.

\end{example}

\begin{example}[Application of Online Gradient Descent in Conformal Prediction]
    We note that when $\+X = \-R^d$ is unconstrained, then \Cref{theorem:GD proximal regret} with $\Phiproj$ implies that \hyperref[GD]{GD} guarantees the following property: the norm of the averaged gradient $\InNorms{\frac{1}{T}\sum_{t=1}^T \nabla\ell^t(x^t)} \le \frac{\InNorms{x^1-x^{T+1}}}{\eta T}$. So as long as $\InNorms{x^T}$ grows sublinear in $T$,  the norm of the averaged gradient $\InNorms{\frac{1}{T}\sum_{t=1}^T \nabla\ell^t(x^t)} \rightarrow 0$. This property is called \emph{gradient equilibrium}~\citep{angelopoulos2025gradient}. The property of gradient equilibrium implies marginal coverage for online conformal prediction, while the property of no external regret is not sufficient~\citep{angelopoulos2025gradient, ramalingam2025relationship}. This implication shows the benefit of analyzing \hyperref[GD]{GD} through the lens of the more general proximal regret notion.

    Although the property directly follows from the fact that $\sum_{t=1}^T \eta \nabla \ell^t(x^t) = x^1 - x^{T+1}$, we show this property as a direct implication of proximal regret guarantee and the bound provided in \Cref{theorem:GD proximal regret} to illustrate the generality of proximal regret and tightness of \Cref{theorem:GD proximal regret}.
\begin{corollary}
    Let $\+X = \-R^d$ and $\{\ell^t: \-R^d \rightarrow \-R\}$ be a sequence of convex loss functions. Then \hyperref[GD]{GD} with any constant step size $\eta$ satisfies for any $T \ge 1$ and any vector $v \in \-R^d$,
    \begin{align*}
        \InAngles{\frac{1}{T}\sum_{t=1}^T \nabla \ell^t(x^t), v} \le \frac{\InAngles{v, x^1-x^{T+1}}}{\eta T}.
    \end{align*}
    In particular, this implies
    \begin{align*}
        \InNorms{\frac{1}{T}\sum_{t=1}^T \nabla \ell^t(x^t)} \le \frac{\InNorms{x^1-x^{T+1}}}{\eta T}.
    \end{align*}
\end{corollary}
\begin{proof}
    Fix any vector $v$. We can apply \Cref{theorem:GD proximal regret} with the linear $f(x) = \InAngles{v,x}$, which defines $\prox_f(x) = x - v$ for all $x$. Also note that $x^{t+1} = x^t - \eta \nabla \ell^t(x^t)$ and $p^{t}:=\prox_f(x^t) = x^t - v$ for all $t \ge 1$. Moreover, we have $D = \InNorms{x^1 - p^1} = \InNorms{v}$ and $B_f = f(p^1) - f(p^T) = \InNorms{x^1 - x^T, v}$.
    Then by item 1 in \Cref{theorem:GD proximal regret}, we have
    \begin{align*}
        \sum_{t=1}^T \InAngles{\nabla\ell^t(x^t), v} &\le \frac{D^2 + 2B_f - \InNorms{x^{T+1}-p^t}^2}{2\eta} + \sum_{t=1}^T\frac{\eta}{2}\InNorms{\nabla \ell^t(x^t)}^2 - \sum_{t=1}^{T-1} \frac{1}{2\eta} \InNorms{p^t - p^{t+1}}^2, \\
        &=\frac{D^2 + 2B_f - \InNorms{x^{T+1}-x^T + v}^2}{2\eta} + \sum_{t=1}^T\frac{1}{2\eta}\InNorms{x^t - x^{t+1}}^2 - \sum_{t=1}^{T-1} \frac{1}{2\eta} \InNorms{x^t - x^{t+1}}^2 \\
        & = \frac{D^2 + 2B_f -\InNorms{x^{T+1}-x^T + v}^2 + \InNorms{x^T - x^{T+1}}^2 }{2\eta} \\
        & = \frac{\InNorms{v}^2 + 2 \InAngles{x^1 - x^T, v} - 2 \InAngles{x^{T+1} - x^{T}, v} - \InNorms{v}^2}{2 \eta}\\
        & = \frac{\InAngles{x^1 - x^{T+1}, v}}{\eta}.
    \end{align*}
    This completes the proof for the first inequality. 

    The second inequality immediately follows by choosing $v = \sum_{t=1}^T \nabla\ell^t(x^t)$ and then applying Cauchy-Schwarz.
\end{proof}
\end{example}
\begin{example}[Symmetric Linear Swap Regret]
    Linear swap regret is $\Phi$-regret where $\Phi$ contains all the affine \emph{endomorphisms} over $\+X$: any affine strategy modification $\phi(x) = Ax + b$ is an \emph{endomorphism} that maps from $\+X$ to $\+X$ itself.  We consider \emph{symmetric} linear swap regret where we impose the additional assumption that $A$ is symmetric. We show that symmetric linear swap regret can be instantiated as proximal regret so \hyperref[GD]{GD} minimizes symmetric linear swap regret (after some equivalent transformation, see \Cref{sec:symmetric linear swap regret} for details). We remark that very recently~\citep{daskalakis2024efficient} gives an efficient learning algorithm that minimizes the general linear swap regret, but their algorithm is considerably much more complicated than \hyperref[GD]{GD}. Moreover, \hyperref[GD]{GD} simultaneously minimizes other proximal regret $\reg^T_f$ for different convex/smooth functions $f$.  Since \hyperref[GD]{GD} is widely applied in the practice, we view providing stronger guarantees for this classic algorithm as an important contribution.
\end{example}

\subsubsection{Faster Convergence with Optimism in the Game Setting}
\label{sec:OG game setting}
Any online algorithm suffers an $\Omega(\sqrt{T})$ $\Phiprox(\delta)$-regret even against linear loss functions since $\Phiprox(\delta)$-regret covers the external regret as a special case.  This lower bound, however, holds only in the \emph{adversarial} setting.  
In this section, we show an improved $O(T^{\frac{1}{4}})$ individual $\Phiprox(\delta)$-regret bound under a slightly stronger smoothness assumption (\Cref{assumption:stronger smoothness}) in the \emph{game} setting, where players interact with each other using the same algorithm. This assumption is naturally satisfied by finite normal-form games and is also made for results about concave games~\citep{farina2022near}.
\begin{assumption}\label{assumption:stronger smoothness}
For any player $i \in [n]$, the utility $u_i(x)$ satisfies $\InNorms{\nabla_{x_i} u_i(x) - \nabla_{x_i} u_i(x')} \le L\InNorms{x - x'}$ for all $x, x' \in \X$. 
\end{assumption}
We remark that improved regret guarantees have been shown for the external regret~\citep{syrgkanis2015fast,chen2020hedging,daskalakis2021near-optimal, farina2022near} in general concave games or the swap regret in the subclass of normal-form games~\citep{chen2020hedging, anagnostides2022near-optimal, anagnostides2022uncoupled}. Our result fits in the line of literature as an extension from external to proximal regret for general concave games, which also implies a faster $O(\varepsilon^{-\frac{4}{3}})$ convergence to an $(\varepsilon+\frac{L\delta^2}{2})$-approximate $\Phiprox(\delta)$-equilibria for non-concave smooth games in the first-order stationary regime.

We study the Optimistic Gradient \eqref{OG} algorithm~\citep{rakhlin2013optimization}, an optimistic variant of \hyperref[GD]{GD} that has been shown to have improved individual \emph{external} regret guarantee in the game setting~\citep{syrgkanis2015fast}. The \ref{OG} algorithm initializes $w^0 \in \X$ arbitrarily and $g^0 = 0$. In each step $t \ge 1$, the algorithm plays $x^t$, receives gradient feedback $g^t := \nabla f^t(x^t)$, and updates $w^t$, as follows:
\begin{equation}
\label{OG}
\tag{OG}
    \begin{aligned}
        x^t = \Pi_\X \InBrackets{ w^{t-1} - \eta g^{t-1}},  \quad w^t = \Pi_\X \InBrackets{ w^{t-1} - \eta g^t}. 
    \end{aligned}
\end{equation}
We show that for any $f \in \+F_{\lscc}$, the adversarial $\Phiprox(\delta)$-regret of \ref{OG} can be $O(\sqrt{P^t})$   where $P^t:=\sum_{t=1}^T \InNorms{g^t- g^{t-1}}^2$ is the total sum of gradient variation (\Cref{thm:adversaril regret of OG} in \Cref{app:proofs proj regret}). In the game setting, the gradient variation can be further bounded as players' strategies are more stable. This improved gradient variation dependent regret bound leads to a fast $O(T^{1/4})$ $\Phiprox(\delta)$-regret and convergence to approximate $\Phiprox(\delta)$-equilibrium in games.

\begin{theorem}[Improved Individual $\Phiproj(\delta)$-Regret of \ref{OG} in the Game Setting]
\label{thm:game regret of OG}
    In a $G$-Lipschitz $L$-smooth (in the sense of \Cref{assumption:stronger smoothness}) game $\+G = \{[n], \{\+X_i\}, \{u_i\}\}$, when all players employ \ref{OG} with step size $\eta > 0$, then for each player $i$, any $f \in \+F_{\lscc}$, and $T \ge 1$, we have
    \begin{align*}
        \sum_{t=1}^T \InAngles{\nabla_{x_i} u_i(x^t), \prox_f(x^t_i)-x^t_i }\le \frac{D_{\+X_i}^2 + 2B_f}{\eta} +  2\eta G^2 + 3n L^2 G^2 \eta^3 T, \forall f \in \+F_{\lscc}. 
    \end{align*}
    Choosing $\eta = T^{-\frac{1}{4}}$, it is bounded by $(D_{\+X_i}^2 + 2B_f + 4n L^2 G^2)T^{\frac{1}{4}}$. Furthermore, for any $\delta> 0$ and any $\varepsilon> 0 $, their empirical distribution of played strategy profiles converges to an $(\varepsilon + \frac{\delta^2L}{2})$-approximate $\Phiproxeq(\delta)$-\lce in 
    $O(1/\varepsilon^{\frac{4}{3}})$ iterations.
\end{theorem}

\subsection{Convex Combination of Finite Local Strategy Modifications}
\label{sec:convex-phi}
This section considers $\Conv(\Phi)$ where $\Phi$ is a finite set of local strategy modifications. The set of infinite strategy modifications $\Conv(\Phi)$ is defined as $
    \Conv(\Phi) =\{ \phi_p(x) = \sum_{\phi \in \Phi} p(\phi) \phi(x): p \in \Delta(\Phi)\}.$ Our main result is an efficient algorithm (\Cref{alg:convex phi-reg}) that guarantees convergence to an $\varepsilon$-approximate $\Conv(\Phi)$-equilibrium in a smooth game satisfying \Cref{assumption:smooth games} for any $\varepsilon > \delta^2 L$. Due to space constraints, we defer \Cref{alg:convex phi-reg} and the proof to \Cref{app:convex-phi}.
\begin{theorem}
\label{theorem:convex finite-phi-regret} Let $\X$ be a convex and compact set, $\Phi$ be an arbitrary finite set of $\delta$-local strategy modification functions for $\X$, and $u^1(\cdot),\ldots, u^T(\cdot)$ be a sequence of $G$-Lipschitz and $L$-smooth but possibly non-concave reward functions from $\X$ to $[0,1]$. If we instantiate \Cref{alg:convex phi-reg} with the Hedge algorithm as the regret minimization algorithm $\mathfrak{R}_{\Phi}$ over $\Delta(\Phi)$ and $K=\sqrt{T}$, the algorithm guarantees that, with probability at least $1 - \beta$, it produces a sequence of strategies $x^1, \ldots, x^T$ with $\Conv(\Phi)$-regret at most $8\sqrt{T} \InParentheses{ G\delta \sqrt{\log |\Phi|} + \sqrt{\log(1/\beta)}} + \delta^2L T$. The algorithm runs in time $\sqrt{T}|\Phi|$ per iteration.

If all players in a non-concave smooth game employ \Cref{alg:convex phi-reg}, then with probability $1-\beta$, for any $\varepsilon > 0$, the empirical distribution of strategy profiles played forms an $(\varepsilon+\delta^2 L)$-approximate $\Phi = \Pi_{i=1}^n \Phi^{\mathcal{X}_i}$-equilibrium,, after $\poly\left(\frac{1}{\varepsilon}, G, \log \left(\max_i |\Phi^{\mathcal{X}_i}|\right), \log \frac{n}{\beta}\right)$ iterations.
\end{theorem}
\paragraph{Proof Sketch} We adopt the framework in \citep{stoltz2007learning,gordon2008no} (as described in \Cref{sec: finite-phi}) with two main modifications. First, we utilize the $L$-smoothness of the utilities to transform the problem of external regret over $\Delta(\Phi)$ against non-concave rewards into a linear optimization problem.  Second, we use the technique of ``expected fixed point"~\citep{zhang2024efficient} to circumvent the intractable problem of finding a fixed point.

\subsection{Interpolation-Based Local Strategy Modifications}
\label{sec:phi-int-regret minization}
We introduce a natural set of local strategy modifications and the corresponding local equilibrium notion. Given any set of (possibly non-local) strategy modifications $\Psi = \{ \psi: \X \rightarrow \X\}$, we define a set of \emph{local} strategy modifications as follows: for $\delta \le D_\X$ and $\lambda \in [0,1]$, each strategy modification $\phi_{\lambda, \psi}$ interpolates the input strategy $x$ with the modified strategy $\psi(x)$: formally, 
\begin{align*}
    \Phi^\X_{\Int, \Psi}(\delta):=\left\{ \phi_{\lambda, \psi}(x) := (1- \lambda)x + \lambda \psi(x): \psi \in \Psi, \lambda \le \delta / D_\X\right\}.
\end{align*}
Note that for any $\psi \in \Psi$ and $\lambda \le \frac{\delta}{D_\X}$, we have $\InNorms{\phi_{\lambda, \psi}(x) - x} = \lambda \InNorms{x - \psi(x)} \le \delta$, respecting the locality constraint. The induced $\Phi^\X_{\Int, \Psi}(\delta)$-regret can be written as 
\[ \reg^T_{\Int,\Psi, \delta} :=   \max_{\psi \in \Psi, \lambda \le \frac{\delta}{D_\X}} \sum_{t=1}^T \InParentheses{\ell^t(x^t) - \ell^t\InParentheses{(1-\lambda)x^t + \lambda \psi(x^t)}}.\]
To guarantee convergence to the corresponding $\Phi_{\Int,\Psi}(\delta)$-equilibrium in the first-order stationary regime, it suffices to minimize $\Phi^\X_{\Int, \Psi}(\delta)$-regret against linear losses, which we show further reduces to $\Psi$-regret minimization against linear losses (\Cref{thm:int reduction} in \Cref{app:phi-int-regret minization}). Therefore, we can apply any efficient algorithms for computing/learning $\Psi$-equilibrium in concave setting for $\Phi_{\Int,\Psi}(\delta)$-equilibrium in smooth non-concave games in the first-order stationary regime.

\paragraph{CCE-like Instantiation} In the special case where $\Psi$ contains only \emph{constant} strategy modifications (i.e., $\psi(x) = x^*$ for all $x$), we get a coarse correlated equilibrium (CCE)-like instantiation of local equilibrium, which limits the gain by interpolating with any \emph{fixed} strategy. We denote the resulting set of local strategy modification simply as $\Phiint(\delta)$\footnote{We note that $\Phiint(\delta)$ is in fact contained in $\Phiprox(\delta)$. Let $f_{x'}(x) = \frac{1}{2}\InNorms{x-x'}^2$, then $\prox_{\frac{\lambda}{1-\lambda} f_{x'}} = (1-\lambda)x + \lambda x'$. }. We can apply any no-external regret algorithm for efficient $\Phiint(\delta)$-regret minimization and computation of $\varepsilon$-approximate $\Phiinteq(\delta)$-\lce in the first-order stationary regime as summarized in \Cref{thm:lce_int}. We also discuss faster convergence rates in the game setting in \Cref{app:phi-int-regret minization}.

\begin{theorem}\label{thm:lce_int}
     For the \hyperref[GD]{Online Gradient Descent} algorithm ~\citep{zinkevich2003online} with step size $\eta = \frac{\Dx}{G\sqrt{T}}$, its $\Phiint(\delta)$-regret is at most $2\delta G\sqrt{T}$. Furthermore, for any $\delta> 0$ and any $\varepsilon>\frac{\delta^2L}{2}$, when all players employ the \hyperref[GD]{GD}  algorithm in a smooth game, their empirical distribution of played strategy profiles converges to an $(\varepsilon + \frac{\delta^2L}{2})$-approximate $\Phiinteq(\delta)$-\lce in $O(1/\varepsilon^2)$ iterations. 
\end{theorem}

\section{Hardness Beyond First-Order Stationary Regime}
\label{sec:first-order regime hardness}

For an infinite set of $\delta$-local strategy modifications $\Phi(\delta)$, we focus on the complexity of computing $\varepsilon$-approximate $\Phi(\delta)$-equilibrium when $\varepsilon = \Omega(\delta^2)$. We call $\varepsilon = \Omega(\delta^2)$ the \emph{first-order stationary regime} since a deviation bounded by $\delta$ only gives $\delta^2$ gain in utility. We show that by considering $\Phi(\delta)$-equilibria that correlate players' strategy, efficient algorithms exist for $\varepsilon = O( L\delta^2)$ for certain classes of $\Phi(\delta)$ (\Cref{sec:proximal regret,sec:convex-phi,sec:phi-int-regret minization}). A natural question on the complexity of computing an $\varepsilon$-approximate $\Phi(\delta)$-equilibrium is  
\begin{equation}
    \text{Is first-order stationary regime the best we can hope for? Can we get $\varepsilon = o(\delta^2)$ efficiently?}
\end{equation}
We show that unless \textsc{P}$=$\textsc{NP}, there is no algorithm that computes an $\varepsilon$-approximate $\Phi(\delta)$-equilibrium for $\varepsilon = o(\delta^2)$ and runs in time $\poly(d, G, L, 1/\delta,1/\varepsilon)$. This hardness result complements our positive results where we design algorithms that run in time $\poly(d, G, L, 1/\delta,1/\varepsilon)$ in the first-order stationary regime $\varepsilon = \Omega(\delta^2L)$. Together, they present a more complete picture of the complexity landscape of approximate $\Phi(\delta)$-equilibrium. 

We first present a weaker hardness that holds for $\Phi(\delta) = \Phi_{\mathrm{All}}(\delta)$ that contains all local strategy modifications. Then we present a hardness result that holds even for a restricted set of strategy modifications  $\Phi_{\mathrm{Int}^+}(\delta)$ is the combination of $\Phiinteq(\delta)$ and one more local strategy modification. 

\subsection{Hardness When $\Phi(\delta)$ Contains All Local Strategy Modifications}
We start with the case when $\Phi(\delta)$ contains all $\delta$-local strategy modifications. We consider a single-player game with an action set $\+X := \{x \in \-R^d_{+}: \InNorms{x}_1 \le 1\}$ and quadratic utility functions. As a byproduct of our analysis, we also prove the NP-hardness of computing a $(\varepsilon, \delta)$-local maximizer of a smooth and Lipschitz quadratic function over a polytope.\footnote{A point $x$ is an $(\varepsilon, \delta)$-local maximizer of a function $f: \+X \rightarrow \mathbb{R}$ if and only if $f(x) \ge f(x') - \varepsilon, \forall x' \in B_\delta(x) \cap \+X$.} This is a generalization of the hardness of computing an \emph{exact} local maximizer of a quadratic function over a polytope~\cite{ahmadi2022complexity}. 

We will construct a reduction from the Maximum Clique problem denoted as \textsc{Max Clique}. 
\begin{definition}[\textsc{Max Clique}]
    Given a graph $G = (V, E)$ with vertices $V$ and edges $E$. The \textsc{Max Clique} problem asks to compute the clique number $\omega(G)$ defined as the size of the maximum clique in the graph $G$.
\end{definition}
The \textsc{Max Clique} problem is NP-hard. Moreover, even approximating \textsc{Max Clique} with a factor of $O(d^{1-o(1)})$ is NP-hard~\citep{zuckerman2006linear}, where $d$ is the number of nodes in the graph.

\paragraph{Hard Functions} Fix any graph $G = (V, E)$ with $|V| = d \ge 3$. We denote $A \in \{0,1\}^{d \times d}$ as $G$'s adjacent matrix. Then by then Motzkin-Straus Theorem~\citep{motzkin_maxima_1965}, we have
\begin{align*}
    \max_{x \in \Delta^d} x^\top A x = 1 - \frac{1}{\omega(G)}.
\end{align*}
We denote $x^* \in \Delta^d$ as an optimal solution to the above optimization problem. 

We define a family of functions $f_k: \+X \rightarrow \-R$ indexed by $k \in [1,d]$:
\begin{align}
     f_k(x) := \InParentheses{ x^\top A x - \InParentheses{1 - \frac{1}{k}}\cdot \InNorms{x}_1^2}, \quad, 1 \le k \le d. \label{dfn:f_k}
\end{align}
We can verify that each $f_k$ is $G$-Lipschitz and $L$-smooth with $G, L = O(\poly(d))$.  
\begin{lemma}\label{lem:properties of f_k}
    The following holds for any $f_k$ over $\+X = \{x \in \-R^d_+: \InNorms{x}_1\le 1\}$:
    \begin{itemize}
        \item[1.] \textbf{Boundedness:} $f_k(x) \in [-1, 1]$ for all $x\in [0,1]^d$. If $k> \omega(G)$, $f_k(x)\le 0$ for all $x\in \+X$.
        \item[2.] \textbf{(Local) Lipschitzness:} For any $x \in [0,1]^d$, $\InNorms{\nabla_x f_k(x)}_2 \le 3\sqrt{d} \InNorms{x}_1$. Therefore, $f_k$ is $3\sqrt{d}$-Lipschitz.
        \item[3.] \textbf{Smoothness:} $f_k$ is $2d$-smooth.
    \end{itemize}
\end{lemma}
\begin{proof}
    For any $x \in \+X$ such that $x \ne 0$, we have
    \begin{align}
        f_k(x) &= x^\top A x - \InParentheses{1 - \frac{1}{k}}\cdot \InNorms{x}_1^2 \nonumber \\
        &= \InNorms{x}_1^2 \cdot \InParentheses{ \frac{x}{\InNorms{x}_1}^\top A  \frac{x}{\InNorms{x}_1} - \InParentheses{1 - \frac{1}{k}}}\nonumber \\
        &\le \InNorms{x}_1^2 \InParentheses{ 1 - \frac{1}{\omega(G)} - \InParentheses{1 - \frac{1}{k}}}\nonumber \tag{Motzkin-Straus Theorem} \\
        & = \InNorms{x}_1^2 \InParentheses{ \frac{1}{k} - \frac{1}{\omega(G)}}. \label{eq:f_k}
    \end{align}
    We note that the inequality becomes equality when $\frac{x}{\InNorms{x}_1} = x^*$ the optimal solution for $\max_{x \in \Delta^d} x^\top Ax$. Thus $f_k(x) \le \InNorms{x}_1^2 \le 1$. Moreover, we have $f_k(x) \ge -(1-\frac{1}{k})\InNorms{x}_1^2\ge -\InNorms{x}_1^2 \ge -1$. Thus $f_k(x) \in [-1, 1]$. If $k>\omega(G)$, then we have $f_k(x)\leq \InNorms{x}_1^2 \InParentheses{ \frac{1}{k} - \frac{1}{\omega(G)}}<0$.
    
    For the second part, we can bound the $\ell_2$-norm of the gradient $\nabla_x f_k(x)$ by
    \begin{align*}
        \InNorms{\nabla_x f_k(x)}_2 &= \InNorms{Ax - 2\InParentheses{1-\frac{1}{k}}\InNorms{x}_1 \cdot 1_d}_2  \\
        &\le \InNorms{Ax}_2 + 2\InNorms{{\InNorms{x}_1 \cdot 1_d}}_2 \\
        &\le \sqrt{d}\InNorms{x}_1 + 2\sqrt{d} \InNorms{x}_1 \\
        &= 3\sqrt{d}\InNorms{x}_1.
    \end{align*}
    where in the first inequality, we apply the triangle inequality; in the second inequality, we use the fact that $A \in [0,1]^d$. Therefore, $f_k$ is $3\sqrt{d}$-Lipschitz over $\+X$.

    For the third part, we compute the the Hessian of $f_k$
    \begin{align*}
        \nabla^2_x f(x) = A - 2(1-\frac{1}{k}) \cdot 1_{d \times d},
    \end{align*}
    where we use $1_{d \times d}$ to denote the all-one matrix. Since $A \in \{0,1\}^{d \times d}$, we know the absolute eigenvalues of $\nabla^2_x f(x)$ is bounded by $2d$. Thus $f_k$ is $2d$-smooth.
\end{proof}

The following technical lemma relates the local maximizer of $f_k$ for $1 \le k \le d$ and the clique number $\omega(G)$. Specifically, \Cref{lem:localMax} shows that when $k\ge \omega(G) + 1$, all $(\frac{\delta^2}{8d^2}, \delta)$-local maximizer must have small $\ell_2$-norms, i.e., $\InNorms{x} < \frac{\delta}{2}$; when $k < \omega(G)$, all $(\frac{\delta^2}{8d^2}, \delta)$-local maximizer must has large $\ell_2$-norm, i.e., $\InNorms{x} > \frac{\delta}{2}$. 
\begin{lemma}
\label{lem:localMax}
    The following holds for all $\delta \in (0,1]$:
    \begin{itemize}
        \item[1.] If $k \ge \omega(G)+1$, then any $x$ with $\InNorms{x}_1 \ge \frac{\delta}{2}$ is not a $(\frac{\delta^2}{8d^2}, \delta)$-local maximizer.
        \item[2.] If $k < \omega(G)$, then any $x$ with $\InNorms{x}_1 \le \frac{\delta}{2}$ is not a $(\frac{\delta^2}{8d^2}, \delta)$-local maximizer.
    \end{itemize}
\end{lemma}
\begin{proof} 
We first prove the case when $k \ge \omega(G) + 1$.
\paragraph{When $k \ge \omega(G) + 1$:} Let $x \in \+X$ be any point such that $\InNorms{x}_1 \ge \frac{\delta}{2}$. Define $x' := \InParentheses{1- \frac{\delta}{2 \InNorms{x}_1}}x$ the point by moving $x$ towards the origin. Then we know the distance between $x$ and $x'$ is bounded by
\begin{align*}
    \InNorms{x - x'}_2 = \frac{\delta}{2 \InNorms{x}_1}\InNorms{x}_2 \le \frac{\delta}{2} \le \delta.
\end{align*}
Moreover, we have $\InNorms{x'}_1 = \InParentheses{1- \frac{\delta}{2 \InNorms{x}_1}} \cdot \InNorms{x}_1$ and thus
\begin{align*}
    f_k(x') - f_k(x) &= \InParentheses{\InParentheses{1- \frac{\delta}{2 \InNorms{x}_1}}^2 - 1} \cdot f_k(x) \\
    & = \InParentheses{\frac{\delta^2}{4 \InNorms{x}^2_1} -2 \cdot \frac{\delta}{ 2\InNorms{x}_1} } \cdot f_k(x) \\
    & \ge \frac{\delta}{2\InNorms{x}_1} \cdot (-f_k(x)) \tag{ $0 < \frac{\delta}{2\InNorms{x}_1} \le 1$ and $f_k(x)<0$} \\
    & \ge \frac{\delta}{2\InNorms{x}_1} \cdot \InNorms{x}_1^2 \cdot \InParentheses{-\frac{1}{k} + \frac{1}{\omega(G)}} \tag{by \eqref{eq:f_k}} \\
    & \ge \frac{\delta^2}{4}\cdot \InParentheses{\frac{k-\omega(G)}{k \cdot \omega(G)}} \tag{$\InNorms{x}_1 \ge \frac{\delta}{2}$} \\
    & > \frac{\delta^2}{4d(d+1)}. \tag{$\omega(G)+1 \le k \le d+1$}
\end{align*}
Thus $x$ is not a $(\frac{\delta^2}{8d^2}, \delta)$-local maximizer.

\paragraph{When $k < \omega(G)$:} Let $x \in \+X$ be any point such that $\InNorms{x}_1 \le \frac{\delta}{2}$. We further define a threshold $v_x$ for $x$:
\begin{align*}
    v_x := \frac{f_k(x)}{\InNorms{x}_1^2}.
\end{align*}
and proceed by two cases depending on $v_x$.
\paragraph{Case 1: $v_x \ge \frac{1}{2d^2}$.} In this case, we consider a deviation $x' := x + \frac{\delta}{2\InNorms{x}_1} x$. We note that $x'$ still lies in $\+X$ since $\InNorms{x'}_1 = (1+\frac{\delta}{2\InNorms{x}_1})\InNorms{x}_1 = \InNorms{x}_1 + \frac{\delta}{2} \le 1$. The magnitude of deviation is bounded by $\InNorms{x'-x}_2 = \frac{\delta}{2 \InNorms{x}_1}\InNorms{x}_2 \le \frac{\delta}{2} $. Moreover, the function value increases at least
\begin{align*}
    f_k(x') - f_k(x) &= \InParentheses{\InParentheses{1+\frac{\delta}{ 2\InNorms{x}_1}}^2 - 1} \cdot f_k(x) \\
    & > \frac{\delta^2}{4\InNorms{x}_1^2} \cdot f_k(x) \\
    & \ge \frac{\delta^2}{8d^2} \tag{$v_x = \frac{f_k(x)}{\InNorms{x}_1^2} \ge \frac{1}{2d^2}$}.
\end{align*}
Thus $x$ is not a  $(\frac{\delta^2}{8d^2}, \delta)$-local maximizer.

\paragraph{Case 2: $v_x < \frac{1}{2d^2}$.} We consider deviation towards the optimal solution $x^* \in \Delta^d$ with function value $f_k(x^*) = \frac{1}{k} - \frac{1}{\omega(G)} \ge \frac{1}{d^2}$. Specifically, we let $x' = \frac{\delta}{2} x^*$. It is clear that $\InNorms{x'} = \frac{\delta}{2} \le 1$ so $x'$ is valid deviation in $\+X$.  The deviation distance from $x$ to $x'$ is bounded by $\InNorms{x' - x}_2 \le \InNorms{x - x'}_1 \le \InNorms{x}_1 + \InNorms{x'}_1 \le \delta$. Moreover, deviation from $x$ to $x'$ increases the function value by 
\begin{align*}
    f_k(x') - f_k(x) &= \frac{\delta^2}{4} \cdot f_k(x^*) - f_k(x)\\
    &>  \frac{\delta^2}{4d^2} - f_k(x).
\end{align*}
If $f_k(x) \le 0$, then we know $f_k(x') - f_k(x) \ge \frac{\delta^2}{4d^2}$. Otherwise, we can further lower bound the increase in function value by
\begin{align*}
    f_k(x') - f_k(x) &\ge \frac{\delta^2}{4d^2}- f_k(x)\\
    & \ge \frac{\delta^2}{4d^2} - \frac{\delta^2}{4} \cdot \frac{1}{\InNorms{x}_1^2} \cdot f_k(x)\tag{$\InNorms{x}_1^2 \le \frac{\delta^2}{4}$} \\
    & > \frac{\delta^2}{8d^2}. \tag{$v_x = \frac{f_k(x)}{\InNorms{x}_1^2} < \frac{1}{2d^2}$}
\end{align*}
Combining the above, we know $x$ is not a $(\frac{\delta^2}{8d^2}, \delta)$-local maximizer.
\end{proof}

Denote $\Phi_{\mathrm{All}}(\delta)$ the set of all $\delta$-local strategy modifications. Using \Cref{lem:localMax}, we can prove NP-hardness of $\varepsilon$-approximate $\Phi_{\mathrm{All}}(\delta)$-equilibrium for two cases: (1) $\varepsilon \le \frac{\delta^2}{16d}$; (2) $\varepsilon \le \poly(G,L,d) \cdot \delta^{2+c}$ for any $\poly(G,L,d)$ and $c > 0$.
\begin{theorem}[Hardness of Approximate $\Phi_\mathrm{All}(\delta)$-Equilibrium]
\label{thm:hardnessFOSall swap}
   There exists a family of single-player games with utility functions $f:\{x\in \R^d_+: \InNorms{x}_1\le 1\} \rightarrow [-1,1]$ that is $3\sqrt{d}$-Lipschitz and  $2d$-smooth, such that the followings are true.
    \begin{itemize}
        \item[1.] For any $0 < \delta \le 1$, if there is an algorithm that computes an $\varepsilon$-approximate $\Phi_{\mathrm{All}}(\delta)$-equilibrium for $\varepsilon \le \frac{\delta^2}{16d^2}$ in time $\poly(d)$, then \textsc{P} $=$ \textsc{NP}.
        \item[2.] For any $g = \poly(G,L,d)$ and $c > 0$, if  there is an algorithm that computes an $\varepsilon$-approximate $\Phi_{\mathrm{All}}(\delta)$-equilibrium for any $\delta \in \left(0, \InParentheses{\frac{1}{16d^2\cdot g}}^{\frac{1}{c}}\right]$ and $\varepsilon = g \cdot \delta^{2+c}$ in time $\poly(d)$, then \textsc{P} $=$ \textsc{NP}.  
    \end{itemize}
\end{theorem}
\begin{proof}
    We consider utility functions $f_k: \+X \rightarrow [-1, 1]$ for $1 \le k \le d$ as defined in \eqref{dfn:f_k}. By \Cref{lem:properties of f_k}, we know $f_k$ is $G$-Lipschitz and $L$-smooth with $G = 3\sqrt{d}$ and $L = 2d$. 
    
    \paragraph{Part 1:} We fix any $\delta \in (0,1]$ and let $\varepsilon = \frac{\delta^2}{16d^2}$. Suppose $\+A$ is an algorithm  that computes an $\varepsilon$-approximate $\Phi_{\mathrm{All}}(\delta)$-equilibrium in time $\poly(d)$. Now consider the following algorithm for approximating \textsc{Max Clique}.
    \begin{itemize}
        \item[1.]  Run $\+A$ for the game $f_k$ and denote the outcome as $\sigma_k$ for $1 \le k \le d$.
        \item[2.]  Return the smallest $1 \le k \le d$ such that $\Pr_{x\sim \sigma_k}[\-I[\InNorms{x}_2 \le \frac{\delta}{2}]] \ge \frac{1}{2}$.
    \end{itemize}
    We claim that the above algorithm outputs either $\omega(G)$ or $\omega(G)+1$. We first prove that $\Pr_{x\sim \sigma_k}[\-I[\InNorms{x}_1 \le \frac{\delta}{2}]] < \frac{1}{2}$ for all $k < \omega(G)$. Suppose the claim is false, then consider the following deviation for all $x$ in $\sigma_k$'s support: (1) if $\InNorms{x}_1 \le \frac{\delta}{2}$, then deviate $x$ to the point $x'$ with $\InNorms{x'-x}\le \delta$ and $f_k(x') - f_k(x) > \frac{\delta^2}{8d^2}$. This deviation is possible since $x$ is not a $(\frac{\delta^2}{8d^2}, \delta)$-local maximizer by \Cref{lem:localMax}. (2) If $\InNorms{x}_1 > \frac{\delta}{2}$, then keep $x$ unchanged. We denote the resulting distribution after deviation as $\sigma_k'$. Then $f_k(\sigma_k') - f_k(\sigma_k) > \frac{\delta^2}{8d^2} \cdot \frac{1}{2} = \frac{\delta^2}{16d^2}$, which contradicts the fact that $\sigma_k$ is an $\frac{\delta^2}{16d^2}$-approximate $\Phi(\delta)$-equilibrium. Using a similar argument, we can prove that $\Pr_{x\sim \sigma_k}[\-I[\InNorms{x}_1 \le \frac{\delta}{2}]] \ge \frac{1}{2}$ for all $k \ge \omega(G) + 1$. Combining the above, we know the algorithm either outputs $\omega(G)$ or $\omega(G) + 1$. Moreover, the algorithm runs in time $\poly(d)$. But since it is NP-hard to approximate \textsc{Max Clique} even within $O(d^{1-o(1)})$ factor~\citep{zuckerman2006linear}, the above implies \textsc{P} $=$ \textsc{NP}.

    \paragraph{Part 2:} We fix any $g = \poly(G, L, d) = \poly(d)$ and $c > 0$. We choose $\delta = (\frac{1}{16d^2 \cdot g})^{\frac{1}{c}} = \poly(\frac{1}{d})$ so that $\varepsilon = g \cdot \delta^{2+c} \le \frac{\delta^2}{16d^2}$. Then by Part 1, we know if there is an algorithm that computes an $\varepsilon$-approximate $\Phi_{\mathrm{All}}(\delta)$-equilibrium in time $\poly(d)$ and, then \textsc{P} $=$ \textsc{NP}.
\end{proof}

As a byproduct of \Cref{lem:localMax}, we can prove hardness results for $(\varepsilon, \delta)$-local maximizer similar to \Cref{thm:hardnessFOSall swap}. The proof is simpler than that of \Cref{thm:hardnessFOSall swap}, and we omitted it here.
\begin{corollary}[Hardness of $(\varepsilon, \delta)$-Local maximizer]
\label{corollary:local maximizer hardness}
    Consider a class of functions $f:\{x\in \R^d_+: \InNorms{x}_1\le 1\} \rightarrow [-1,1]$ that is $3\sqrt{d}$-Lipschitz and  $2d$-smooth Then the following hold.
    \begin{itemize}
        \item[1.] For any $0 < \delta \le 1$ , if there is an algorithm that computes an $(\varepsilon, \delta)$-local maximizer for $\varepsilon \le \frac{\delta^2}{16d^2}$ in time $\poly(d)$, then \textsc{P} $=$ \textsc{NP}.
        \item[2.] For any $g = \poly(G,L,d)$ and $c > 0$, if 
        there is an algorithm that computes an $(\varepsilon, \delta)$-local maximizer for any $\delta \in \left(0, \InParentheses{\frac{1}{16d^2\cdot g}}^{\frac{1}{c}}\right]$ and $\varepsilon = g \cdot \delta^{2+c}$ in time $\poly(d)$, then \textsc{P} $=$ \textsc{NP}.
    \end{itemize}
\end{corollary}
Since we can define a function $f: \+X \times \+Y \rightarrow \-R$ such that $f'(x,y) = f(x)$. \Cref{corollary:local maximizer hardness} also shows the hardness of computing an $(\varepsilon, \delta)$-local Nash equilibrium in a two-player zero-sum non-concave game. 

\subsection{Hardness under Restricted Deviations}
In this section, we show the hardness of approximating a local maximizer and approximating a $\Phi(\delta)$-equilibrium when the class of strategy modifications $\Phi(\delta)$ is $\Phiint(\delta)$ plus one additional deviation. We let $\+X = \{\-R^d_+: \InNorms{x}_1 \le 1\}$ where $D_\+X = \sqrt{d}$. Recall that each strategy modification $\phi_{x'}$ in $\Phiint(\delta)$ is of the form 
\begin{align*}
    \phi_{\lambda, x'}(x) = \InParentheses{1 - \lambda} x + \lambda  x',\quad \forall \lambda \in [0, \frac{\delta}{\sqrt{d}}] \text{ and } x' \in \+X
\end{align*}
The additional strategy modification we consider is 
\begin{align*}
    \phi(x) = \begin{cases}
        \InParentheses{1- \frac{\delta}{12d^3 \InNorms{x}_1}}x, & \InNorms{x}_1 \ge \frac{\delta}{12d^3}, \\
        x , & \mathrm{Otherwise}.
    \end{cases} 
\end{align*}
We can check that $\phi$ is a well-defined $\delta$-local strategy modification for $\+X$. We denote the union of $\{\phi\}$ and $\Phiint(\delta)$ as $\Phiintone(\delta)$.

The hard functions we consider is still $f_k$ as defined in \eqref{dfn:f_k} .The following technical lemma is similar to \Cref{lem:localMax}. The main difference is that \Cref{lem:localMax} allows arbitrary $\delta$-local strategy modifications, while in \Cref{lem:restricted-localMax}, we are restricted to strategy modifications in $\Phiintone(\delta)$.
\begin{lemma}
\label{lem:restricted-localMax}
    For all $\delta\in (0,1]$, the following holds for $f_k$ over $\+X = \{x\in \R^d_+: \InNorms{x}_1\le 1\}$
    \begin{itemize}
        \item[1.] If $k \ge \omega(G)+1$, there exists a \emph{single} $\phi \in \Phiintone(\delta)$ such that $f_k(\phi(x)) - f_k(x) > \frac{\delta^2}{144d^8}$ \emph{for all} $x$ with $\InNorms{x}_1 \ge \frac{\delta}{12d^3}$. 
        \item[2.] If $k < \omega(G)$, there exists a \emph{single} $\phi \in \Phiintone(\delta)$ such that $f_k(\phi(x)) - f_k(x) \ge \frac{\delta^2}{2d^3} > \frac{\delta^2}{144d^8}$ \emph{for all $x$} with $\InNorms{x}_1 \le \frac{\delta}{12d^3}$, .
    \end{itemize}
\end{lemma}
\begin{proof} We first prove the claim for $k \ge \omega(G) + 1$.
    \paragraph{When $k \ge \omega(G) + 1$:} Let $x \in \+X$ be any point such that $\InNorms{x}_1 \ge \frac{\delta}{12d^3}$. Define $x' = \InParentheses{1- \frac{\delta}{12d^3 \InNorms{x}_1}}x$. We note that $x' \in [0,1]^d$ is well-defined and the distance between $x$ and $x'$ is bounded by
    \begin{align*}
        \InNorms{x - x'}_2 = \frac{\delta}{12d^3\InNorms{x}_1}\InNorms{x}_2 \le \frac{\delta}{12d^3} \le \delta.
    \end{align*}
    Moreover, we have $\InNorms{x'}_1 = \InParentheses{1- \frac{\delta}{12d^3\InNorms{x}_2}} \cdot \InNorms{x}_1$
    \begin{align*}
        f_k(x') - f_k(x) &= \InParentheses{\InParentheses{1- \frac{\delta}{12d^3 \InNorms{x}_1}}^2 - 1} \cdot f_k(x) \\
        & = \InParentheses{\InParentheses{\frac{\delta}{12 d^3 \InNorms{x}_1}}^2 -2 \cdot \frac{\delta}{12 d^3 \InNorms{x}_1} } \cdot f_k(x) \\
        & \ge \frac{\delta}{12 d^3\InNorms{x}_1} \cdot (-f_k(x)) \tag{ $0 < \frac{\delta}{12 d^3\InNorms{x}_1} \le 1$ and $f_k(x)<0$} \\
        & \ge \frac{\delta}{12 d^3\InNorms{x}_1} \cdot \InNorms{x}_1^2 \cdot \InParentheses{-\frac{1}{k} + \frac{1}{\omega(G)}} \tag{by \eqref{eq:f_k}}\\
        & \ge \frac{\delta^2}{144d^{6}}\cdot \InParentheses{-\frac{1}{k} + \frac{1}{\omega(G)}} \tag{$\InNorms{x}_1 \ge \frac{\delta}{12 d^3}$} \\
        & > \frac{\delta^2}{144d^8}. \tag{$d \ge k \ge \omega(G)+1$}
    \end{align*}
    Thus $x$ is not a $(\frac{\delta^2}{144d^8}, \delta)$-local maximizer.

    \paragraph{When $k < \omega(G)$:} Let $x \in \+X$ be any point such that $\InNorms{x}_1 \le \frac{\delta}{12d^3}$. We consider the deviation that interpolates with the optimal solution $x^* \in \Delta^d$, whose function value is $f_k(x^*) = \frac{1}{k} - \frac{1}{\omega(G)} \ge \frac{1}{d^2}$. Specifically, we let $x' = (1-\frac{\delta}{\sqrt{d}})x + \frac{\delta}{\sqrt{d}}x^*$. We know the deviation distance from $x$ to $x'$ is bounded by $\InNorms{x' - x}_2 \le \delta$. Moreover, for any $y$ lies in the line segment between $x$ and $x'$, we have $\InNorms{y}_1 \le \InNorms{x}_1 +\InNorms{ \frac{\delta}{\sqrt{d}}x^*}_1 \le \frac{2\delta}{\sqrt{d}}$. Thus by the mean value theorem, there exists a $y$ in the line segment between $x$ and $x'$, such that \[f_k(x') = f_k\InParentheses{\frac{\delta}{\sqrt{d}} \cdot x^*} +  \InAngles{\nabla f_k(y), x' - \frac{\delta}{\sqrt{d}} x^*}\geq f_k\InParentheses{\frac{\delta}{\sqrt{d}} \cdot x^*} - \InNorms{\nabla f_k(y)}_2 \cdot \InNorms{x' - \frac{\delta}{\sqrt{d}} x^*}_2.\] By the local Lipschitzness of $f_k$, the function value of $f_k(x')$ is at least
    \begin{align*}
        f_k(x') &\ge f_k(\frac{\delta}{\sqrt{d}} \cdot x^*) - \InNorms{\nabla f_k(y)}_2 \cdot \InNorms{x' - \frac{\delta}{\sqrt{d}} x^*}_2\\
        &\ge \frac{\delta^2}{d} \cdot \frac{1}{d^2} - 3\sqrt{d} \cdot \frac{2\delta}{\sqrt{d}} \cdot \InNorms{(1- \frac{\delta}{\sqrt{d}})x}_2 \\
        &\ge \frac{\delta^2}{d^3} - 6\delta \cdot \InNorms{x}_1\\
        &\ge \frac{\delta^2}{2d^3},
    \end{align*}
In the second inequality, we use $f_k(x^*) \ge \frac{1}{d^2}$ and $\InNorms{\nabla f_k(y)}_2 \le 3\sqrt{d}\InNorms{y}_1 \le 6\delta$ (\Cref{lem:properties of f_k}); in the third inequality, we use the fact that the $\ell_1$-norm upper bounds the $\ell_2$-norm.
    
    On the other hand, we have 
    \begin{align*}
        f_k(x) &= \InNorms{x}_1^2 \cdot f_k\InParentheses{\frac{x}{\InNorms{x}_1}} \\
        &\le \frac{\delta^2}{144 d^6} \cdot f_k(x^*) \\
        &\le \frac{\delta^2}{144 d^6}.
    \end{align*}
    Combining the above, we have $f_k(x') - f(x) \ge \frac{\delta^2}{3d^3}$ and we know $x$ is not a $(\frac{\delta^2}{3d^3}, \delta)$-local maximizer.
\end{proof}

Combining \Cref{lem:restricted-localMax} and similar analysis as \Cref{thm:hardnessFOSall swap}, we have the following hardness results for computing an $\varepsilon$-approximate $\Phiintoneeq(\delta)$-equilibrium.
\begin{theorem}[Hardness of Approximate $\Phiintoneeq(\delta)$-Equilibrium]
\label{thm:hardnessFOS restricted}
   There exists a family of single-player games with utility functions $f:\{x\in \R^d_+: \InNorms{x}_1\le 1\} \rightarrow [-1,1]$ that is $3\sqrt{d}$-Lipschitz and  $2d$-smooth, such that the followings are true.
    \begin{itemize}
        \item[1.] For any $0 < \delta \le 1$ such that $\delta = \poly(1/d)$, if there is an algorithm that computes an $\varepsilon$-approximate $\Phiintoneeq(\delta)$-equilibrium for $\varepsilon \le \frac{\delta^2}{576d^8}$ in time $\poly(d)$, then \textsc{P} $=$ \textsc{NP}. 
        \item[2.] For any $g = \poly(G,L,d)$ and $c > 0$, if  there is an algorithm that computes an $\varepsilon$-approximate $\Phiintoneeq(\delta)$-equilibrium for any $\delta \in \left(0, \InParentheses{\frac{1}{576d^8\cdot g}}^{\frac{1}{c}}\right]$ and $\varepsilon = g \cdot \delta^{2+c}$ in time $\poly(d)$, then \textsc{P} $=$ \textsc{NP}.  
    \end{itemize}
\end{theorem}
\begin{proof}
    We consider utility functions $f_k$ as defined in \eqref{dfn:f_k}. By \Cref{lem:properties of f_k}, we have $f_k$ is $3\sqrt{d}$-Lipschitz and $2d$-smooth.

    \paragraph{Part I:} We fix any $\delta \in (0,1]$ and let $\varepsilon = \frac{\delta^2}{576d^8}$. Suppose $\+A$ is an algorithm  that computes an $\varepsilon$-approximate $\Phi_{\mathrm{All}}(\delta)$-equilibrium in time $\poly(d)$. Now consider the following algorithm for approximating \textsc{Max Clique}: define $p = \frac{\delta}{1728 d^{10}} = \frac{1}{\poly(d)}$,
    \begin{itemize}
        \item[1.]  Run $\+A$ for the game $f_k$ and denote the outcome as $\sigma_k$ for $1 \le k \le d$.
        \item[2.]  Return the smallest $1 \le k \le d$ such that $\Pr_{x\sim \sigma_k}[\-I[\InNorms{x}_1 \le \frac{\delta}{12d^3}]] \ge 1 - p$.
    \end{itemize}
    We claim that the above algorithm outputs either $\omega(G)$ or $\omega(G)+1$. We first prove that $\Pr_{x\sim \sigma_k}[\-I[\InNorms{x}_1 \le \frac{\delta}{12d^3}]] < 1-p$ for all $k < \omega(G)$. Suppose the claim is false. By \Cref{lem:restricted-localMax}, we know there is a single deviation $\phi \in \Phiintone(\delta)$ such that $f_k(\phi(x)) - f_k(x) > \frac{\delta^2}{144d^8}$ for all $x$ with $\InNorms{x}_1 \ge \frac{\delta}{12d^3}$. Moreover, for any $x \in \+X$, since $f_k$ is $3\sqrt{d}$-Lipschitz and $\InNorms{\phi(x) -x}\le \delta$, we have $f_k(\phi(x)) - f_k(x) \ge -3\sqrt{d}\delta$. Combining the above gives
    \begin{align*}
        f_k(\phi(\sigma_k)) - f_k(\sigma_k) &\ge \frac{\delta^2}{144d^8}\cdot (1 - p) - 3\sqrt{d}\delta \cdot  p \\
        & > \frac{\delta^2}{288d^8} - \frac{\delta^2}{576d^8} = \varepsilon.
    \end{align*}
    This contradicts the fact that $\sigma_k$ is an $\varepsilon$-approximate $\Phi(\delta)$-equilibrium. Using a similar argument, we can prove that $\Pr_{x\sim \sigma_k}[\-I[\InNorms{x}_1 \le \frac{\delta}{12d^3}]] \ge 1- p$ for all $k \ge \omega(G) + 1$. Combining the above, we know the algorithm either outputs $\omega(G)$ or $\omega(G) + 1$. Moreover, the algorithm runs in time $\poly(d)$. But since it is NP-hard to approximate \textsc{Max Clique} even within $O(d^{1-o(1)})$ factor~\citep{zuckerman2006linear}, the above implies \textsc{P} $=$ \textsc{NP}.

    \paragraph{Part 2:} We fix any $g = \poly(G, L, d) = \poly(d)$ and $c > 0$. We choose $\delta = (\frac{1}{576 d^8 \cdot g})^{\frac{1}{c}} = \frac{1}{\poly(d)}$ so that $\varepsilon := g \cdot \delta^{2+c} \le \frac{\delta^2}{576d^2}$. Then, by Part 1, we know if there is an algorithm that computes an $\varepsilon$-approximate $\Phi_{\mathrm{All}}(\delta)$-equilibrium in time $\poly(d)$, then \textsc{P} $=$ \textsc{NP}.
\end{proof}

\section{Discussion and Future Directions}
In this paper, we initiate the study of tractable $\Phi$-equilibria in non-concave games and examine several natural families of strategy modifications. For any $\Phi$ that contains only a finite number of strategy modifications, we design an efficient randomized $\Phi$-regret minimization algorithm, which provides efficient uncoupled learning dynamics that converge to the corresponding $\Phi$-equilibria. Additionally, we study several classes of $\Phi(\delta)$ that contain an infinite number of $\delta$-local strategy modifications and show efficient uncoupled learning dynamics that converge to an $\varepsilon$-approximate $\Phi(\delta)$-equilibrium in the first-order stationary regime, where $\varepsilon = \Omega(\delta^2)$. We justify our focus on the first-order stationary regime by proving an \textsc{NP}-hardness result for achieving an approximation error $\varepsilon = o(\delta^2)$, even when the set of strategy modifications is simple, such as $\Phi(\delta) = \Phiintoneeq(\delta)$, which includes only $\Phiinteq$ and one additional strategy modification. 

Below, we discuss several future directions on $\Phi$-regret minimization and $\Phi$-equilibria computation.

\paragraph{Application of Proximal Regret} We propose the notion of proximal regret and show that it is a notion that lies between external regret and swap regret. We give a general analysis of the classic \hyperref[GD]{Gradient Descent (GD)} algorithm showing that it efficiently minimizes the proximal regret $\reg_f^T$ for all convex or smooth functions $f$ (See \Cref{sec:bregman proximal} for generalization on Bregman proximal regret and results for \hyperref[MD]{Mirror Descent (MD)}). Proximal regret has several interesting implications, such as gradient equilibrium and symmetric swap regret. Since \hyperref[GD]{GD} is widely applied in the practice, it would be interesting to show more implications of the proximal regret and our results on \hyperref[GD]{GD}. In particular, our results imply that \hyperref[GD]{GD} dynamics in concave games converge to a refined set of coarse correlated equilibria we call proximal correlated equilibria. Exploring further applications of proximal correlated equilibria in game theory may be a fruitful future direction. 

\paragraph{Other efficiently minimizable $\Phi$-regret.} In this paper, we propose several natural sets of $\Phi$'s that admit efficient $\Phi$-regret minimization. It will be interesting to investigate for which other strategy modifications $\Phi$, the corresponding $\Phi$-regret can be minimized efficiently. 
\paragraph{Improved $\Phiprox(\delta)$-regret in games.} We show in \Cref{thm:game regret of OG} that the optimistic gradient \eqref{OG} dynamics guarantees that an improved individual $\Phiprox(\delta)$-regret of $O(T^{1/4})$.  Could we design uncoupled learning dynamics with better individual regret guarantees, consequently leading to faster convergence to an approximate $\Phiproxeq(\delta))$-equilibrium?

\printbibliography

\appendix

\crefalias{section}{appendix} 

\newpage
\thispagestyle{empty}
\setcounter{tocdepth}{3}
\tableofcontents
\addtocounter{page}{-1}
\thispagestyle{empty}

\section{Related Work}
\label{app:related works}
\paragraph{Non-Concave Games.} An important special case of multi-player games are two-player zero-sum games, which are defined in terms of some function $f: \X \times \Y \rightarrow \R$ that one of the two players say the one choosing $x \in \X$, wants to minimize, while the other player, the one choosing $y \in \Y$, wants to maximize. Finding Nash equilibrium in such games is tractable in the \emph{convex-concave} setting, i.e.~when $f(x,y)$ is convex with respect to the minimizing player's strategy, $x$, and concave with respect to the maximizing player's strategy, $y$, but it is computationally intractable in the general \emph{nonconvex-nonconcave} setting. Namely, a Nash equilibrium may not exist, and it is NP-hard to determine if one exists and, if so, find it. Moreover, in this case, stable limit points of gradient-based dynamics are not necessarily Nash equilibria, not even local Nash equilibria~\citep{daskalakis2018limit, mazumdar2020gradient}. Moreover, there are examples including the “Polar Game”~\citep{pethick2022escaping} and the “Forsaken Matching Pennies”~\citep{hsieh2021limits} showing that for \hyperref[GD]{GD} / \ref{OG} and many other no-regret learning algorithms in nonconvex-nonconcave min-max optimization, the last-iterate does not converge and even the average-iterate fails to be a stationary point. We emphasize that the convergence guarantees we provide for \hyperref[GD]{GD}  / \ref{OG} in \Cref{sec:proximal regret} and \Cref{sec:phi-int-regret minization} holds
for the empirical distribution of play, not the average-iterate or the last-iterate.

A line of work focuses on computing Nash equilibrium under additional structure in the game.  This encompasses settings where the game satisfies the (weak) Minty variational inequality~\citep{mertikopoulos2019learning, diakonikolas2021efficient,pethick2022escaping, cai2023accelerated}, or is sufficiently close to being bilinear~\citep{abernethy2021last}. However, the study of universal solution concepts in the nonconvex-nonconcave setting is sparse. \citet{daskalakis2021complexity} proved the existence and computational hardness of local Nash equilibrium. In a more recent work, \citep{daskalakis2023stay} proposes second-order algorithms with asymptotic convergence to local Nash equilibrium. Several works study sequential two-player zero-sum games with additional assumptions about the player who goes second. They propose equilibrium concepts such as \emph{local minimax points}~\citep{jin2020local}, \emph{differentiable Stackelberg equilibrium}~\citep{fiez2020implicit}, and \emph{greedy adversarial equilibrium}~\citep{mangoubi2021greedy}. Notably, local minimax points are stable limit points of Gradient-Descent-Ascent (GDA) dynamics~\citep{jin2020local, wang2020solving, fiez2021local} while greedy adversarial equilibrium can be computed efficiently using second-order algorithms in the unconstrained setting~\citep{mangoubi2021greedy}. In contrast to these studies, we focus on the more general case of multi-player non-concave games.

\paragraph{Local Equilibrium.} To address the limitations associated with classical, global equilibrium concepts, a natural approach is to focus on developing equilibrium concepts that guarantee local stability instead. One definition of interest is the strict local Nash equilibrium, wherein each player's strategy corresponds to a local maximizer of their utility function, given the other players' strategies. Unfortunately, a strict local Nash equilibrium may not always exist, as demonstrated in \Cref{ex:strict local NE}. Furthermore, a weaker notion—the second-order local Nash equilibrium, where each player has no incentive to deviate based on the second-order Taylor expansion estimate of their utility, is also not guaranteed to exist as illustrated in \Cref{ex:strict local NE}. What's more, it is NP-hard to check whether a given strategy profile is a strict local Nash equilibrium or a second-order local Nash equilibrium, as implied by the result of~\citet{murty_np-complete_1987} and \citet{ahmadi2022complexity}.\footnote{\citet{murty_np-complete_1987} shows that checking whether a point is a local maximizer of a multi-variate quadratic function is NP-hard, while \citet{ahmadi2022complexity} shows that computing a local maximizer of a multi-variate quadratic function over a polytope is NP-hard.} Finally, one can consider \emph{local Nash equilibrium}, a first-order stationary solution, which is guaranteed to exist~\citep{daskalakis2021complexity}. Unlike non-convex optimization, where targeting first-order local optima sidesteps the intractability of global optima, this first-order local Nash equilibrium has been recently shown to be intractable, even in two-player zero-sum non-concave games with joint feasibility constraints~\citep{daskalakis2021complexity}.\footnote{In general sum games, it is not hard to see that the intractability results~\citep{daskalakis2009complexity, chen2009settling} for computing \emph{global} Nash equilibria in bimatrix games imply intractability for computing \emph{local} Nash equilibria.} See \Cref{tab:equilibrium notions} for a summary of solution concepts in non-concave games.

\paragraph{Online Learning with Non-Convex Losses.} A line of work has studied online learning against non-convex losses. To circumvent the computational intractability of this problem, various approaches have been pursued: some works assume a restricted set of non-convex loss functions~\citep{gao2018online}, while others assume access to a sampling oracle~\citep{maillard2010online,krichene2015hedge} or access to an offline optimization oracle~\citep{agarwal2019learning, suggala2020online, heliou2020online} or a weaker notion of regret~\citep{hazan2017efficient, aydore2019dynamic, hallak2021regret,guan2023online}. The work most closely related to ours is \citep{hazan2017efficient}. The authors propose a notion of \emph{$w$-smoothed local regret} against non-convex losses, and they also define a local equilibrium concept for non-concave games. They use the idea of \emph{smoothing} to average the loss functions in the previous $w$ iterations and design algorithms with optimal $w$-smoothed local regret. The concept of regret they introduce suggests a local equilibrium concept. However, their local equilibrium concept is non-standard in that its local stability is not with respect to a distribution over strategy profiles sampled by this equilibrium concept.
Moreover, the path to attaining this local equilibrium through decentralized learning dynamics remains unclear. The algorithms provided in~\citep{hazan2017efficient,guan2023online} require that every agent $i$ experiences (over several rounds) the average utility function of the previous 
 $w$ iterates, denoted as $F^t_{i,w}:=\frac{1}{w} \sum_{\ell=0}^{w-1} u_i^{t-\ell}(\cdot,x_{-i}^{t-\ell})$. Implementing this imposes a significant coordination burden on the agents.
In contrast, we focus on a natural concept of $\Phi(\delta)$-\lce, which is incomparable to that of~\cite{hazan2017efficient}, and we also show that efficient convergence to this concept is achieved via decentralized gradient-based learning dynamics.

\paragraph{$\Phi$-regret and $\Phi$-equilibrium.} The concept of $\Phi$-regret and the associated $\Phi$-equilibrium is introduced by~\citet{greenwald2003general} and has been broadly investigated in the context of concave games~\citep{greenwald2003general,stoltz2007learning, gordon2008no, rakhlin2011online, piliouras2022evolutionary, bernasconi2023constrained, dagan2024external,peng2024fast} and extensive-form games~\citep{von2008extensive,morrill2021hindsight, morrill2021efficient,Farina22:Simple, bai2022efficient, song2022sample, Anagnostides23:NearOptimal, zhang2024efficient}. The work of~\cite{sharma2024no} studies internal regret minimization against non-convex losses. To our knowledge, no efficient algorithm exists for the classes of $\Phi$-equilibria we consider for non-concave games. Specifically, all of these algorithms, when applied to compute a $(\varepsilon, \Phi(\delta))$-equilibrium for a general $\delta$-local strategy modification set $\Phi(\delta)$ (using \Cref{lemma:no-regret-2-CE}), require running time exponential in either $1/\varepsilon$ or the dimension $d$. In contrast, we show that for several natural choices of $\Phi(\delta)$, $\varepsilon$-approximate $\Phi(\delta)$-equilibrium can be computed efficiently, i.e. polynomial in $1/\varepsilon$ and $d$, using simple algorithms.

\subsection{Comparison with (First-Order) Local Coarse Correlated Equilibrium~\citep{ahunbay2024first}}\label{app:comparison}
In the first two versions of the paper, we analyzed two sets of strategy modifications $\Phiproj$ (\Cref{ex:projection}) and $\Phiint$ (\Cref{sec:phi-int-regret minization}), and their localized versions $\Phiproj(\delta)$ and $\Phiint(\delta)$. In this version, we propose a general strategy modification set $\Phiprox$ (\Cref{dfn:proximal regret}) that contains both $\Phiproj$ and $\Phiint$ as special cases. More specifically, we show that the prox operator $\prox_f(x) = \argmin_{x'\in \+X} \{f(x) + \frac{1}{2}\InNorms{x -x'}^2\}$ can be instantiated with $f(x) = \InAngles{v,x}$ to recover $\phi_v$ and $f(x) = \frac{\delta}{2(1-\delta)}\InNorms{x - x^*}^2$ to recover $\phi_{x^*,\delta}$. We show that \hyperref[GD]{GD} minimizes the proximal regret $\reg_f^T$ for all convex or smooth functions $f$ simultaneously in \Cref{theorem:GD proximal regret}.

Following our initial version,~\citep{ahunbay2024first} proposed a generalized set of strategy modifications based on gradient fields, which includes both $\Phiproj(\delta)$ and $\Phiint(\delta)$ as special cases. Specifically, given a gradient field $\nabla f$ of a continuously differentiable function $f$,  \cite{ahunbay2024first} consider the deviation $\phi_f(x) = \Pi_\+X[x - \nabla f(x)]$. ~\cite{ahunbay2024first}  show that 
\begin{itemize}
    \item a projection-based deviation $\phi_v(x) = \Pi_\+X[x - v]$ is equivalent to adding the gradient of the linear function $-\InAngles{v,x}$ to $x$: $\phi_v(x) = \Pi_\+X[x - \nabla_x\InAngles{v,x}] $
    \item an interpolation-based deviation $\phi_{x^*,\delta}(x) = (1-\delta)x + \delta x^*$ is equivalent to adding the gradient of the quadratic function $-\frac{\delta}{2}\InNorms{x - x^*}^2$ to $x$: $\phi_{x^*, \delta}(x) = x  - \nabla_x(\frac{\delta}{2}\InNorms{x - x^*}^2)$
\end{itemize}
\citep{ahunbay2024first} propose two notions of \emph{first-order} coarse corrected equilibria (CCE) based on the strategy modifications generated by gradient fields. 
Consider a $G$-Lipschitz and $L$-smooth game $\+G$, a family of Lipschit continuous gradient fields $\{\nabla f: f \in F\}$ where $F$ is a set of continuous differentiable functions with bounded gradient norm $G_f = \max_{x \in \+X} \InNorms{\nabla f(x)}$ and Lipschitz constant $L_f$. A distribution $\sigma$ over $\+X$ is an $\varepsilon$-\emph{local CCE} if 
\begin{align}\label{eq:Local CCE}
    \sum_{i\in[n]} \-E_{x \sim \sigma} \InBrackets{\InAngles{\Pi_{T_{\+X_i}(x_i)}[\nabla f_i(x)], \nabla_i u_i(x)   }}\le \varepsilon \cdot \poly(G,L, G_f, L_f), \forall f \in F.\footnotemark
\end{align}\footnotetext{$T_\+X(x)$ is the tangent cone of $x \in \+X$}
Let $\Phi_F = \{\phi_f: f \in F\}$. The definition of $\varepsilon$-\emph{local CCE} can be seen as the first-order approximation of the $\Phi_F$-equilibria.  A distribution $\sigma$ over $\+X$ is an $\varepsilon$-\emph{stationary CCE} if 
\begin{align}\label{eq:stationary CCE}
    \left | \sum_{i\in[n]} \-E_{x \sim \sigma} \InBrackets{\InAngles{\nabla f_i(x), \Pi_{T_{\+X_i}(x_i)}[\nabla_i u_i(x)]   }} \right|\le \varepsilon \cdot \poly(G,L, G_f, L_f), \forall f \in F.
\end{align}

\citep{ahunbay2024first} show that when every player runs gradient descent (\hyperref[GD]{GD}) with the same step sizes, sampling from a suitably-defined continuous curve extension of their generated strategies is an approximate local/stationary CCE. These results do not hold for all convex sets $\+X$ but require $\+X$ to be either (1) a compact convex set with a smooth boundary or (2) an acute polyhedron. 

We give a comparison of our results (\Cref{theorem:GD proximal regret}) and \citep{ahunbay2024first} and highlight the differences in various aspects. 
\paragraph{Strategy Modifications} Our strategy modifications based on the proximal operator are \emph{different} and \emph{incomparable} from their strategy modifications based on gradient fields of Lipschitz and smooth functions. For example, our $\Phiprox$-regret covers projection-based deviations such that $\phi_C(x) = \Pi_C[x]$ for any convex subset $C \subseteq \+X$. The deviation is generated by $\prox_{\mathbb{I}_C}$ where $\mathbb{I}_C$ is the indicator function for the set $C$. Such strategy modifications can not be generated using gradient fields. We also remark that there are strategy modifications based on gradient fields that proximal operators can not recover.

\paragraph{Incentive Guarantees} Our results hold for the standard $\Phi$-regret and $\Phi$-equilibria setting, in which no player can gain more than $\varepsilon$ utility by deviating within $\Phi$ (\Cref{def:local CE}). In contrast, the results in \citet{ahunbay2024first} concern local coarse correlated equilibria (local CCE) and stationary CCE, as defined in \eqref{eq:Local CCE} and \eqref{eq:stationary CCE}, which are first-order approximations of $\Phi$-equilibria. While these guarantees coincide for certain classes of $\Phi$, such as $\Phiint$, they are generally different and incomparable.

\paragraph{Constrained Sets} Our results hold for general convex sets while \citep{ahunbay2024first}'s results require additional assumptions on the constrained sets as discussed above. 

\paragraph{Adversarial Setting} We show that \hyperref[GD]{GD} minimizes $\Phiprox$-regret in the adversarial setting, which implies results when all players run \hyperref[GD]{GD}.  \citep{ahunbay2024first}'s results do not give the corresponding sublinear $\Phi$-regret bounds of \hyperref[GD]{GD} in the adversarial setting. 

\paragraph{Average-Iterate Convergence} Our results in the adversarial setting directly implies when all players use \hyperref[GD]{GD}, their empirical distribution of play converges to an approximate $\Phiprox$-equilibrium. \citep{ahunbay2024first}'s results do not hold for the empirical distribution of play but require sampling from the continuous curve extension of the dynamics.\footnote{\cite{ahunbay2024first} provides partial results for \emph{first-order} regret in the adversarial setting when the gradient fields satisfy an additional ``tangential" assumption. Under this assumption, the empirical distribution of play also satisfies \eqref{eq:Local CCE} and \eqref{eq:stationary CCE}. However, this assumption is somewhat restrictive and is not satisfied by natural deviations such as the projection based deviation 
$\Phiproj$.} Moreover, \citep{ahunbay2024first}'s results require all players to use the same step sizes, while our results do not.

\paragraph{Implications on Gradient Descent Dynamics} Our results show that the empirical distribution of play of the \hyperref[GD]{GD} dynamics converges to a refined subset of CCE we called proximal correlated equilibria. The general idea of providing tighter guarantees for \hyperref[GD]{GD} dynamics and refined class of CCE is also advocated by \citep{ahunbay2024first}, but our results and incentive guarantees are different as discussed above.

\section{Solution Concepts in Non-Concave Games: Existence and Complexity}
\label{app:preliminaries}
We present definitions of several solution concepts in the literature as well as the existence and computational complexity of each solution concept. 
\begin{definition}[Nash Equilibrium]
\label{def:Nash}
In a continuous game, a strategy profile $x \in \prod_{j=1}^n \X_j$ is a \emph{Nash equilibrium (NE)} if and only if for every player $i \in [n]$, \[
u_i(x_i', x_{-i}) \le u_i(x),  \forall x_i' \in \X_i\]
\end{definition}

\begin{definition}[Mixed Nash Equilibrium]
\label{def:mixed NE}
In a continuous game, a mixed strategy profile $p \in \prod_{j=1}^n \Delta(\X_j)$ (here we denote $\Delta(\X_i)$ as the set of probability measures over $\X_i$) is a \emph{mixed Nash equilibrium (MNE)} if and only if for every player $i \in [n]$, \[
u_i(p_i', p_{-i}) \le u_i(p),  \forall p_i' \in \Delta(\X_i)\]
\end{definition}

\begin{definition}[(Coarse) Correlated Equilibrium]\label{def:(C)CE}
In a continuous game, a distribution $\sigma$ over joint strategy profiles $\Pi_{i=1}^n \X_i$ is a \emph{correlated equilibrium (CE)} if and only if for all player $i \in [n]$, \[
\max_{\phi_i: \X_i \rightarrow \X_i} 
\-E_{x \sim \sigma} \InBrackets{u_i(\phi_i(x_i), x_{-i})} \le \-E_{x \sim \sigma} \InBrackets{u_i(x)}.\]
Similarly,  a distribution $\sigma$ over joint strategy profiles $\Pi_{i=1}^n \X_i$ is a \emph{coarse correlated equilibrium (CCE)}  if and only if for all player $i \in [n]$, \[
\max_{x_i' \in \X_i} 
\-E_{x \sim \sigma} \InBrackets{u_i(x_i', x_{-i})} \le \-E_{x \sim \sigma} \InBrackets{u_i(x)}.\]
\end{definition} 

\begin{definition}[Strict Local Nash Equilibrium]
\label{def:strict local NE}
In a continuous game, a strategy profile $x \in \prod_{j=1}^n \X_j$ is a \emph{strict local Nash equilibrium} if and only if for every player $i \in [n]$, there exists $\delta > 0$ such that \[
u_i(x_i', x_{-i}) \le u_i(x),  \forall x_i' \in B_{d_i}(x_i, \delta) \cap \X_i.\]
\end{definition}

\begin{definition}[Second-order Local Nash Equilibrium]
\label{def:2nd local NE}
Consider a continuous game where each utility function $u_i(x_i, x_{-i})$ is twice-differentiable with respect to $x_i$ for any fixed $x_{-i}$. A strategy profile $x \in \prod_{j=1}^n \X_j$ is a \emph{second-order local Nash equilibrium} if and only if for every player $i \in [n]$, $x_i$ maximizes the second-order Taylor expansion of its utility functions at $x_i$, or formally,  \[
\InAngles{ \nabla_{x_i} u_i(x), x_i' - x_i} + (x_i' - x_i)^\top \nabla^2_{x_i} u_i(x) (x_i' - x_i) \le 0,  \forall x_i' \in \X_i.\]
\end{definition}

\subsection{Existence}
Mixed Nash equilibria exist in continuous games, thus smooth games~\citep{debreu1952social, glicksberg1952further,fan1953minimax}. By definition, an MNE is also a CE and a CCE. This also proves the existence of CE and CCE. In contrast,  strict local Nash equilibria, second-order Nash equilibria, or (pure) Nash equilibria may not exist in a smooth non-concave game, as we show in the following example.
\begin{example}\label{ex:strict local NE}
    Consider a two-player zero-sum non-concave game: the action sets are $\X_1 = \X_2 = [-1, 1]$ and the utility functions are $u_1(x_1, x_2) = - u_2(x_1, x_2) = (x_1 - x_2)^2$. Let $x = (x_1, x_2) \in \X_1 \times \X_2$ be any strategy profile: if $x_1 = x_2$, then player 1 is not at a local maximizer; if $x_1 \ne x_2$, then player 2 is not at a local maximizer. Thus $x$ is not a strict local Nash equilibrium. Since the utility function is quadratic, we conclude that the game also has no second-order local Nash equilibrium. 
\end{example}

\subsection{Computational Complexity}
Consider a single-player smooth non-concave game with a quadratic utility function $f: \X \rightarrow \R$. The problem of finding a \emph{local} maximizer of $f$ can be reduced to the problem of computing a NE, a MNE, a CE, a CCE, a strict local Nash equilibrium, or a second-order local Nash equilibrium. Since computing a local maximizer or checking if a given point is a local maximizer is NP-hard~\citep{murty_np-complete_1987}, we know that the computational complexities of NE, MNE, CE, CCE, strict local Nash equilibria, and second-order local Nash equilibria are all NP-hard.

\subsection{Representation Complexity of Exact Equilibria}
\citet{karlin1959mathematical} present a two-player zero-sum non-concave game whose unique MNE has infinite support. Since in a two-player zero-sum game,  the marginal distribution of an (exact) CE or a (exact) CCE is an MNE, it also implies that the representation complexity of any CE or CCE is infinite. We present the example in \citet{karlin1959mathematical} here for completeness and also prove that the game is Lipschitz and smooth. 

\begin{example}[{\citep[Chapter 7.1, Example 3]{karlin1959mathematical}}]
\label{example:infintie}
We consider a two-player zero-sum game with action sets  $\X_1 = \X_2 = [0,1]$. Let $p$ and $q$ be two distributions over $[0,1]$. The only requirement for $p$ and $q$ is that their cumulative distribution functions are not finite-step functions. 
For example, we can take $p = q$ to be the uniform distribution.

Let $\mu_n$ and $\nu_n$ denote the $n$-th moments of $p$ and $q$, respectively. Define the utility function 
\begin{align*}
    u(x,y) = u_1(x,y) = -u_2(x,y) = \sum_{n=0}^\infty \frac{1}{2^n} (x^n - \mu_n)(y^n - \nu_n), \quad  0 \le x, y \le 1.
\end{align*}
\end{example}
\begin{claim}
    The game in \Cref{example:infintie} is $2$-Lipschitz and $6$-smooth, and $(p,q)$ is its unique (mixed) Nash equilibrium.
\end{claim}
\begin{proof}
    Fix any $y \in [0,1]$, since $|\frac{1}{2^n}(y^n-\nu_n) n x^{n-1}| \le \frac{n}{2^n}$, the series of $\nabla_x u(x,y)$ is uniformly convergent. We have $ |\nabla_x u(x,y)| \le \sum_{n=0}^\infty \frac{n}{2^n} \le 2, \quad y \in [0,1].$
    Similarly, we have $|\nabla^2_x u(x,y)| \le \sum_{n=0}^\infty \frac{n^2}{2^n} \le 6$ for all $y \in [0,1]$. By symmetry, we also have $|\nabla_y (x,y)| \le 2$ and $|\nabla^2_y(x,y)| \le 6$ for all $x, y \in [0,1]$. Thus, the game is $2$-Lispchitz and $6$-smooth.
    
    Since $|\frac{1}{2^n} (x^n - \mu_n)(y^n - \nu_n)| \le \frac{1}{2^n}$, the series of $u(x,y)$ is absolutely and uniformly convergent. We have 
    \begin{align*}
        \int_{0}^1 u(x,y) \mathrm{d} F_p(x) = \sum_{n=0}^\infty \frac{1}{2^n} (y^n - \nu_n) \int_0^1 (x^n - \mu_n) \mathrm{d} F_p(x) \equiv 0,\\
        \int_{0}^1 u(x,y) \textbf{}F_q(y) = \sum_{n=0}^\infty \frac{1}{2^n} (x^n - \mu_n) \int_0^1 (y^n - \nu_n) \mathrm{d} F_q(y) \equiv 0.
    \end{align*}
    In particular, $(p, q)$ is a mixed Nash equilibrium, and the value of the game is $0$. Suppose $(p',q')$ is also a mixed Nash equilibrium. Then $(p, q')$ is a mixed Nash equilibrium. Note that $p$ supports on every point in $[0,1]$. As a consequence, we have
    \begin{align*}
        0 \equiv \int_0^1 u(x,y) \mathrm{d}F_{q'}(y) =\sum_{n=0}^\infty \frac{1}{2^n} (x^n - \mu_n) (\nu'_n- \nu_n) 
    \end{align*}
for all $x\in[0,1]$, where $\nu_n'$ is the $n$-th moment of $q'$. Since the series vanished identically, the coefficients of each power of $x$ must vanish. Thus $\nu'_n = \nu_n$ and $q' = q$. Similarly, we have $p'= p$, and the mixed Nash equilibrium is unique.  
\end{proof}

\subsection{Representation Complexity of Approximate Equilibria}
We now turn to the representation complexity of \emph{approximate} equilibrium. We give a smooth non-concave game such that any MNE, CE, CCE has exponential representation complexity, i.e., the size of the support of any equilibrium is exponential in the dimension. This result holds even for approximate equilibrium and for two-player zero-sum games. As an important implication, our example rules out any $\poly(\log(1/\varepsilon))$-time algorithms for MNE, CE, CCE for non-concave games even if they are given non-convex optimization (or any other) oracles since even writing down any $\varepsilon$-approximate equilibrium requires $\Omega(1/\varepsilon)$ time.

\begin{theorem}\label{thm:representation of approx}
    For any $\alpha \in (0,1)$, there exists a $2d$-dimensional, $O(\frac{1}{\alpha})$-Lipschitz, and $O(\frac{d}{\alpha^2})$-smooth two-player zero-sum non-concave game such that
    \begin{itemize}
        \item[1.] For any $\varepsilon < \frac{\alpha^d}{2}$, any $\varepsilon$-MNE has a support size at least $(\frac{1}{\alpha})^d$;
        \item[2.] For any $\varepsilon < \frac{\alpha^d}{2}$, any $\varepsilon$-CE or $\varepsilon$-CCE supports on at least $(\frac{1}{\alpha})^d$ joint strategy profiles.
    \end{itemize}
\end{theorem}
\begin{proof}
    Let $\X_1 = \X_2 = B_d(1)$ be the $d$-dimension ball centered at the origin. The utility functions are defined as 
    \begin{align*}
        u_1(x_1, x_2) = -u_2(x_1, x_2) := \begin{cases}
            1+\cos\InParentheses{\frac{\pi}{\alpha^2}\InNorms{x_1 - x_2}^2}  & \InNorms{x_1 -x_2} \le \alpha \\
            0 & \mathrm{otherwise}
        \end{cases}
    \end{align*}
    We can verify that for $i \in \{1,2\}$ $\InNorms{\nabla_{x_i} u_1(x_i)} \le \frac{2\pi}{\alpha}$ and $\InNorms{H_{u_1}(x_i)}_2 \le \InNorms{H_{u_1}(x_i)}_F \le \frac{6\pi^2 d}{\alpha^2}$ ($\InNorms{\cdot}_F$ is the Frobenius norm). Thus the game is $O(\frac{1}{\alpha})$-Lispchitz and $O(\frac{d}{\alpha^2})$-smooth and satisfies \Cref{assumption:smooth games}. We note that the covering number of $B_d(1)$ using ball with radius $\alpha$ is at least $(\frac{1}{\alpha})^d$. This means that for $N < (\frac{1}{\alpha})^d$ points, there exists a point $x \in B_d(1)$ that is $\alpha$-away from all $N$ points.
    
    Let $\varepsilon < \frac{\alpha^d}{2}$ and $(p_1, p_2)$ be any $\varepsilon$-MNE of the game. We assume $p_1$ is supported on $k$ points denoted as $x_1^1, x_1^2, \ldots, x_1^k$ and $p_2$ is supported on $l$ points denoted as $x_2^1, x_2^2, \ldots, x_2^l$. For the sake of contradiction, we assume both $k$ and $l$ is smaller than $(\frac{1}{\alpha})^d$, the covering number of $B_d(1)$. Then there exists one point $y_2 \in \X_2$ that is $\alpha$-away from all points in $\{x_1^i\}_{i \in [k]}$, i.e., $\InNorms{x_1^i - y_2} \ge \alpha$ for all $i \in [k]$.  By definition, we have $u_2(p_1, y_2) = 0$\footnote{We slightly abuse notation here and denote $u(p,p')= \mathbb{E}_{x_1 \sim p, x_2 \sim p'}[u(x_1, x_2)]$.}. Since $(p_1, p_2)$ is an $\varepsilon$-MNE, we get $u_2(p_1,p_2) \ge u_2(p_1, y_2) -\varepsilon = -\varepsilon$. On the other hand, since $p_2$ supports on $l$ points, there exists one point $y_1 \in \{x_2^j\}_{j \in [l]}$ such that $p_2(y_1) \ge \frac{1}{l}$. Thus $u_1(p_1, p_2) \ge u_1(y_1, p_2) - \varepsilon \ge \frac{1}{l} - \varepsilon$. Combining the above gives $u_1(p_1,p_2) + u_2(p_1, p_2) \ge \frac{1}{l} - 2\varepsilon = \frac{1}{l} - \alpha^d > 0$ which contradicts the fact that the game is zero-sum. Consequently, any $\varepsilon$-MNE has support size at least $(\frac{1}{\alpha})^d$. Similar reasoning also give that any $\varepsilon$-CE ($\varepsilon$-CCE, respectively) of the game supports on at least $(\frac{1}{\alpha})^d$ joint strategy profiles. In fact, since the marginal distribution of a $\varepsilon$-CE ($\varepsilon$-CCE, respectively) is a $\varepsilon$-MNE in two-player zero-sum games, the support size of any $\varepsilon$-CE ($\varepsilon$-CCE, respectively) must be at least $(\frac{1}{\alpha})^d$. 
\end{proof}

An important implication of \Cref{thm:representation of approx} is that the representation complexity of an $\varepsilon$-MNE (CE, CCE, respectively) is at least $\Omega(1/\varepsilon)$ even for two-player zero-sum non-concave games with simple independent constraints. Therefore, there is no $\poly(\log(1/\varepsilon))$-time algorithm for computing an $\varepsilon$-CCE. This is in stark contrast to concave games where computing an $\varepsilon$-CCE or even $\varepsilon$-linear CE has a $\poly(\log(1/\varepsilon))$-time algorithm~\citep{daskalakis2024efficient}. We formally present the results below.

\begin{definition}[\textsc{$\varepsilon$-MixedNash}] Given $\varepsilon, L, G \in \mathbb{R}_+$, two circuits implementing a $L$-smooth and $G$-Lipschitz function $u: B_d(1) \times B_d(1) \rightarrow [0,1]$ and its gradients $\nabla u: B_d(1) \times B_d(1) \rightarrow \mathbb{R}^{2d}$, find a distribution $(p_1, p_2)$ such that 
\begin{align*}
    u(x_1, p_2) - \varepsilon \le u(p_1, p_2) \le u(p_1, x_2) + \varepsilon, \forall x_1, x_2 \in B_d(1).
\end{align*}
    
\end{definition}

\begin{corollary}
    \textsc{$\varepsilon$-MixedNash} requires $\Omega(1/\varepsilon)$ time even if $G = O(\frac{1}{\varepsilon})$ and $L = O(\frac{1}{\varepsilon^2})$ and $d=1$.
\end{corollary}
\begin{proof}
    Fix any $\varepsilon < \frac{1}{2}$. We consider the game $\+G$ defined in the proof of \Cref{thm:representation of approx} with $\alpha = 2\varepsilon < 1$. The game $\+G$ is $O(\frac{1}{\varepsilon})$-Lipschitz and $O(\frac{1}{\varepsilon^2})$-smooth. However, any $\varepsilon$-mixed Nash equilibrium has support size at least $\frac{1}{\alpha} = \frac{1}{2\varepsilon}$.
\end{proof}

\begin{corollary}
    \textsc{$\varepsilon$-MixedNash} requires $\Omega(2^d, 1/\varepsilon)$ time even if $G = O(1)$ and $L = O(d)$.
\end{corollary}
\begin{proof}
    Let $\varepsilon $ We consider the game in \Cref{thm:representation of approx} with $\alpha = \frac{1}{2}$ so that $G = O(1)$ and $L=O(d)$. Let $\varepsilon = \frac{1}{2^{d+1}}$. By \Cref{thm:representation of approx}, every $\varepsilon$-mixed Nash equilibrium has support size at least $2^d = \frac{1}{2\varepsilon}$.
\end{proof}
According to  \Cref{thm:representation of approx}, the same lower bounds also hold for approximate CE and CCE.

\section{Proof of \Cref{lemma:no-regret-2-CE}}
\label{app:lemma:no-regret-2-CE}
Since the utility $u_i$ of each player $i$ is $L$-smooth, we have 
\begin{align*}
    u_i(\phi(x_i), x_{-i}) - u_i(x_i) \le \InAngles{\nabla_{x_i} u_i(x), \phi(x_i) -  x_i } + \frac{L\delta^2}{2}, \forall x, \phi \in \Phi^{\+X_i} (\delta).
\end{align*}
Therefore, we have $\forall i\in [n], \forall \phi \in \Phi^{\+X_i}(\delta)$,
\begin{align*}
    &\-E_{x \sim \sigma} \InBrackets{ u_i(\phi(x_i), x_{-i}) - u_i(x_i) }\\
    &\le \-E_{x \sim \sigma} \InBrackets{\InAngles{\nabla_{x_i} u_i(x), \phi(x_i) -  x_i }} + \frac{L\delta^2}{2} \\
    & \le \varepsilon + \frac{L\delta^2}{2}.
\end{align*}
By definition, $\sigma$ is an $(\varepsilon + \frac{L\delta^2}{2})$-approximate $\Phi(\delta)$-equilibrium. This completes the proof of the first part.

Let each player $i \in [n]$ employ algorithm $\+A$ in a smooth game independently and produce iterates $\{x^t\}$. The averaged joint strategy profile $\sigma^T$ that chooses $x^t$ uniformly at random from $t\in [T]$ satisfies for any player $i \in [n]$,
\begin{align*}
    &\max_{\phi \in \Phi^{\X_i}(\delta)} \-E_{x \sim \sigma} \InBrackets{u_i(\phi(x_i), x_{-i})} - \-E_{x \sim \sigma} \InBrackets{u_i(x)}\\
    &= \max_{\phi \in \Phi^{\X_i}(\delta)} \frac{1}{T} \sum_{t=1}^T \InParentheses{ u_i(\phi(x^t_i), x^t_{-i}) - u_i(x^t)} \\
    &\le \max_{\phi \in \Phi^{\X_i}(\delta)} \frac{1}{T} \sum_{t=1}^T \InParentheses{ \InAngles{\nabla u_i(x), \phi(x_i) - x_i} + \frac{\delta^2L}{2}} \\
    &= \frac{\reg_{\Phi^{\X_i}(\delta)}^T}{T} + \frac{\delta^2L}{2}.
\end{align*}
Thus $\sigma^T$ is a $( \max_{i\in[n]}\{\reg_{\Phi^{\X_i}(\delta)}^T\}\cdot T^{-1} + \frac{\delta^2L}{2})$-approximate $\Phi(\delta))$-\lce. This completes the proof of the second part.

\section{Missing Details in \Cref{sec:proximal regret}}
\label{app:proofs proj regret}

\subsection{Online Gradient Descent Minimizes Symmetric Linear Swap Regret}
\label{sec:symmetric linear swap regret}
Let $\+X \subset \-R^d$. We consider symmetric affine endomorphism $\phi(x) = Ax + b$ such that $A \in \-R^{d \times d}$ is symmetric. We first show sufficient conditions under which $\phi = \prox_f$ for a quadratic function $f$.
\begin{proposition}\label{prop:affine transformation}
    For any symmetric affine endomorphism $\phi(x) = Ax +b$ over $\+X$, we have $\phi=\prox_f$ for $f(x) = \frac{1}{2} x^\top (A^{-1}-I)x - (A^{-1}b)^\top x$ if $A$ satisfies one of the two conditions below:
\begin{itemize}
\item $A$ is positive semidefinite (PSD) with the largest eigenvalue $\sigma_{\max}(A) \le 1$
\item $A$ is positive definite (PD) with the smallest eigenvalue $\sigma_{\min}(A) > \frac{1}{2}$
\end{itemize}
\end{proposition}
    \begin{proof}
         We first show that $\prox_f$ is well-defined. We note that $f$ is a quadratic function.
        \begin{itemize}
            \item if $A$ is PSD with $\sigma_{\max}(A) \le 1$, then $\nabla^2f=A^{-1}-I$ is PSD. Therefore, $f$ is convex;
            \item if $A$ is PD with $\sigma_{\min}(A) > \frac{1}{2}$, then $\InNorms{\nabla^2 f}_2=\InNorms{A^{-1}-I}_2 \in (\frac{1}{\sigma_{\max}}-1, \frac{1}{\sigma_{\min}}-1) \in (-1,1)$. Therefore, $f$ is $L$-smooth with $L < 1$.
        \end{itemize}
        Therefore, $\prox_f$ is well-defined. Fix any $x \in \+X$. Recall that 
        \begin{align*}
            \prox_f(x) = \argmin_{x' \in \+X} F_x(x'): = f(x') + \frac{1}{2}\InNorms{x' - x}^2
        \end{align*}
        Let us denote $x^* = \phi(x) = Ax + b \in \+X$. Then
        \begin{align*}
            \nabla F_x(x^*) &= (A^{-1} - I) x^* - A^{-1} b + x^* - x \\
            & = A^{-1} (x^* - b)  - x\\
            & = x - x = 0.
        \end{align*}
        Since $F_x(x')$ is a strictly convex function, we know $x^* = Ax + b \in \+X$ is its unique minimizer. Thus $\prox_f(x) = Ax+b = \phi(x)$. 
    \end{proof}
    Given any symmetric affine endomorphism $\phi(x) = Ax +b$, we consider a new symmetric affine endomorphism $\phi_\alpha := (1-\alpha) \mathrm{Id} + \alpha \phi$ such that $\phi_\alpha(x) = A_\alpha x+ \alpha b$ where $A_\alpha = (1-\alpha)I + \alpha A$. Now for any $\alpha \in (0, \frac{1}{2(1+\InNorms{A})})$, we have the eigenvalues of $A_\alpha$ are at least $1-\alpha - \alpha \InNorms{A} > \frac{1}{2}$. Thus $A_\alpha$ is PD and satisfies item 2 in \Cref{prop:affine transformation} and can be instantiated as proximal regret with a smooth $f$ as defined in \Cref{prop:affine transformation}. By \Cref{theorem:GD proximal regret}, we know \hyperref[GD]{GD} minimizes $\reg_{\phi_\alpha}(T)$. Since
    \begin{align*}
       \reg_{\phi}(T) \le \sum_{t=1}^T \InAngles{\nabla \ell^t(x^t) ,x^t - \phi(x^t)} = \frac{1}{\alpha} \sum_{t=1}^T \InAngles{\nabla \ell^t(x^t) ,x^t - \phi_\alpha(x^t)},
       \end{align*}
    \hyperref[GD]{GD} also minimizes $\reg_{\phi}(T)$.

\subsection{Differences between External Regret and $\Phiproj$-regret}\label{sec:diff external-proj}
The following two examples show that $\Phiproj(\delta)$-regret is incomparable with external regret for convex loss functions. A sequence of actions may suffer high $\reg^T$ but low $\regproj^T$ (Example~\ref{ex:H-R, L-LCCR}), and vise versa (Example~\ref{ex:H-LLR, L-R}). Since $\Phiproj(\delta)$-regret is a special case of the proximal regret, we know the proximal regret is \emph{strictly} more general than the external regret.
\begin{example}
\label{ex:H-R, L-LCCR}
    Let $f^1(x) = f^2(x) = |x|$ for $x \in \X = [-1, 1]$. Then the $\Phiproj(\delta)$-regret of the sequence $\{x^1 =  \frac{1}{2},  x^2 =-\frac{1}{2}\}$ for any $\delta \in (0, \frac{1}{2})$ is $0$. However, the external regret of the same sequence is $1$. By repeating the construction for $\frac{T}{2}$ times, we conclude that there exists a sequence of actions with $\regproj^T = 0$ and $\reg^T = \frac{T}{2}$ for all $T \ge 2$.
\end{example}
\begin{example}
\label{ex:H-LLR, L-R}
    Let $f^1(x) = -2x$ and $f^2(x) = x$ for $x \in \X = [-1, 1]$. Then the $\Phiproj(\delta)$-regret of the sequence $\{x^1 = \frac{1}{2},  x^2 =0\}$ for any $\delta \in (0, \frac{1}{2})$ is $\delta$. However, the external regret of the same sequence is $0$. By repeating the construction for $\frac{T}{2}$ times, we conclude that there exists a sequence of actions with $\regproj^T = \frac{\delta T}{2}$ and $\reg^T = 0$ for all $T \ge 2$.
\end{example}

\subsection{Lower Bounds for $\Phiproj$-Regret}\label{sec:lower bound proj regret}
\begin{theorem}[Lower bound for $\Phiproj(\delta)$-regret against convex losses] \label{thm:lower Bound for convex function} 
For any $T \ge 1$, $\Dx>0$, $0 < \delta \le \Dx$, and $G\geq 0$, there exists a distribution $\mathcal{D}$ on  $G$-Lipschitz linear loss functions $f^1, \ldots, f^T$ over $\+X = [-\Dx, \Dx]$ such that for any online algorithm, its $\Phiproj(\delta)$-regret on the loss sequence satisfies 
$
    \mathbb{E}_{\mathcal{D}} \InBrackets{\regproj^T} = \Omega(\delta G\sqrt{T}). 
$
\end{theorem}

This lower bound suggests that \hyperref[GD]{GD} achieves near-optimal $\Phiproj(\delta)$-regret for convex losses.
For $L$-smooth \emph{non-convex} loss functions, we provide another $\Omega(\delta^2LT)$ lower bound for algorithms that satisfy the linear span assumption. The \emph{linear span} assumption states that the algorithm produces $x^{t+1} \in \{\Pi_{\X}[\sum_{i\in[t]} a_i\cdot x^i +b_i\cdot \nabla f^i(x^i)]: a_i, b_i \in \R, \forall i \in [t]\}$ as essentially the linear combination of the previous iterates and their gradients. 
Many online algorithms, such as online gradient descent and optimistic gradient, satisfy the linear span assumption. 
Combining with \Cref{lemma:no-regret-2-CE}, this lower bound suggests that \hyperref[GD]{GD} attains nearly optimal $\Phiproj(\delta)$-regret, even in the non-convex setting, among a natural family of gradient-based algorithms. 
\begin{proposition}[Lower bound for $\Phiproj(\delta)$-regret against non-convex losses]
\label{prop:lower bound for non-convex function}
For any $T \ge 1$, $\delta \in (0,1)$, and $L\geq 0$, there exists a 
sequence of $L$-Lipschitz and $L$-smooth non-convex loss functions $f^1, \ldots, f^T$ on $\+X = [-1, 1]$ such that for any algorithm that satisfies the linear span assumption, its $\Phiproj(\delta)$-regret on the loss sequence is $ \regproj^T \ge \frac{\delta^2LT}{2}$.
\end{proposition}

\subsubsection{Proof of \Cref{thm:lower Bound for convex function}}
Our proof technique comes from the standard one used in multi-armed bandits \citep[Theorem 5.1]{auer2002nonstochastic}. Suppose that $f^t(x) = g^t x$. 
We construct two possible environments. In the first environment, $g^t=G$ with probability $\frac{1+\varepsilon}{2}$ and $g^t=-G$ with probability $\frac{1-\varepsilon}{2}$; in the second environment, $g^t=G$ with probability $\frac{1-\varepsilon}{2}$ and $g^t = -G$ with probability $\frac{1+\varepsilon}{2}$. We use $\bbE_i$ and $\bbP_i$ to denote the expectation and probability measure under environment $i$, respectively, for $i=1, 2$. Suppose that the true environment is uniformly chosen from one of these two environments. Below, we show that the expected regret of the learner is at least $\Omega(\delta G\sqrt{T})$. 

Define $N_+=\sum_{t=1}^T \mathbb{I}\{x^t\geq 0\}$ be the number of times $x^t$ is non-negative, and define $f^{1:T}=(f^1 ,\ldots, f^T)$. Then we have 
\begin{align}
   \left| \bbE_1[N_+] - \bbE_2[N_+] \right| 
   &= \left| \sum_{f^{1:T}}\Big(\bbP_1(f^{1:T})\bbE\left[N_+\mid f^{1:T}\right] - \bbP_2(f^{1:T})\bbE\left[N_+\mid f^{1:T}\right]\Big)\right| \tag{enumerate all possible sequences of $f^{1:T}$} \\
   &\leq T \sum_{f^{1:T}}\left|\bbP_1(f^{1:T}) - \bbP_2(f^{1:T})\right| \nonumber \\
   &= T \|\bbP_1 - \bbP_2\|_{\mathrm{TV}}\nonumber  \\
   &\leq T\sqrt{(2\ln 2) \KL(\bbP_1, \bbP_2)}    \tag{Pinsker's inequality} \\
   &= T\sqrt{(2\ln 2) T \cdot \KL\left(\text{Bernoulli}\left(\frac{1+\varepsilon}{2}\right), \text{Bernoulli}\left(\frac{1-\varepsilon}{2}\right)\right)} \nonumber \\
   &= T\sqrt{(2\ln 2) T \varepsilon \ln \frac{1+\varepsilon}{1-\varepsilon}} \leq T\sqrt{(4\ln 2) T\varepsilon^2}.    \label{eq: KL dist}
\end{align}
In the first environment, we consider the regret with respect to $v=\delta$. Then we have 
\begin{align*}
    \bbE_1 \left[\regproj^T\right] 
    &\geq \bbE_1 \left[\sum_{t=1}^T f^t(x^t) - f^t(\Pi_{\X}[x^t -\delta]) \right] =  \bbE_1 \left[\sum_{t=1}^T g^t (x^t - \Pi_{\X}[x^t -\delta]) \right] \\
    &= \bbE_1 \left[\sum_{t=1}^T \varepsilon  G(x^t - \Pi_{\X}[x^t -\delta]) \right] \geq \varepsilon\delta G\bbE_1  \left[\sum_{t=1}^T \mathbb{I}\{x^t \geq 0\} \right] = \varepsilon\delta G\bbE_1  \left[N_+\right],  
\end{align*}
where in the last inequality, we use the fact that if $x^t\geq 0$, then $x^t - \Pi_{\X}[x^t-\delta] = x^t - (x^t-\delta) = \delta$ because $D\geq \delta$. 
In the second environment, we consider the regret with respect to $v=-\delta$. Then similarly, we have 
\begin{align*}
    \bbE_2 \left[\regproj^T\right] 
    &\geq \bbE_2 \left[\sum_{t=1}^T f^t(x^t) - f^t(\Pi_{\X}[x^t +\delta]) \right] = \bbE_2 \left[\sum_{t=1}^T g^t (x^t - \Pi_{\X}[x^t +\delta]) \right] \\
    &= \bbE_2 \left[\sum_{t=1}^T -\varepsilon G(x^t - \Pi_{\X}[x^t + \delta]) \right] \geq \varepsilon\delta G\bbE_2  \left[\sum_{t=1}^T \mathbb{I}\{x^t < 0\} \right] = \varepsilon\delta G\left( T - \bbE_2  \left[N_+\right]\right). 
\end{align*}
Summing up the two inequalities, we get  
\begin{align*}
    \frac{1}{2}\left(\bbE_1 \left[\regproj^T\right] + \bbE_2 \left[\regproj^T\right]\right) 
    &\geq \frac{1}{2}\left(\varepsilon \delta GT + \varepsilon\delta G(\bbE_{1}[N_+] - \bbE_{2}[N_+])\right) \\
    &\geq \frac{1}{2}\left(\varepsilon\delta GT - \varepsilon\delta GT\varepsilon \sqrt{(4\ln 2) T}\right).     \tag{by \eqref{eq: KL dist}}
\end{align*}
Choosing $\varepsilon = \frac{1}{\sqrt{(16\ln 2)T}}$, we can lower bound the last expression by $\Omega(\delta G\sqrt{T})$. The theorem is proven by noticing that $\frac{1}{2}\left(\bbE_1 \left[\regproj^T\right] + \bbE_2 \left[\regproj^T\right]\right)$ is the expected regret of the learner. 

\subsubsection{Proof of \Cref{prop:lower bound for non-convex function}}
\begin{proof}
    Consider $f : [-1,1] \rightarrow \R$ such that $f(x) = - \frac{L}{2}x^2$ and let $f^t = f$ for all $t \in [T]$. Then any first-order methods that satisfy the linear span assumption with initial point $x^1 = 0$ will produce $x^t = 0$ for all $t \in [T]$. The $\Phiproj(\delta)$-regret is thus $\sum_{t=1}^T (f(0) - f(\delta)) = \frac{\delta^2LT}{2}$.    
\end{proof}

\subsection{Improved Proximal Regret in the Game Setting and Proof of \Cref{thm:game regret of OG} }
We first prove an $O(\sqrt{T})$ upper bound for \ref{OG} in the adversarial setting.  
\begin{theorem}[Adversarial Regret Bound for \ref{OG}]
\label{thm:adversaril regret of OG}
    Let $\+X \subseteq \-R^d$ be a closed convex set and $\{\ell^t: \+X \rightarrow \-R\}$ be a sequence of convex loss functions. For any convex function $f \in \+F_{\mathrm{lsc,c}}(\+X)$, denote $p^t = \prox_f(x^t)$, \ref{OG} guarantees for all $T \ge 1$,
    \begin{align*}
        \sum_{t=1}^T \InAngles{\nabla \ell^t(x^t), x^t - \prox_f(x^t)} \le \frac{D^2 + 2B_f - \InNorms{w^T-p^T}^2}{2\eta_T} + \sum_{t=1}^T\eta_t\InNorms{g^t - g^{t-1}}^2 - \sum_{t=1}^{T} \frac{1}{2\eta_t} \InNorms{x^t - w^{t}}^2,
    \end{align*}
    where $D = \max_{0\le t\le T-1} \InNorms{w^t-p^{t+1}}$ and $B_f = \max_{t \in [T]}f(p^t)) - \min_{t \in [T]} f(p^t)$. If the step size is constant $\eta_t = \eta$, then the above two bounds hold with $D = \InNorms{w^0 - p^1}$ and $B_f = f(p^1) - f(p^T)$.
\end{theorem}
\begin{remark}
   Compared to the regret of \hyperref[GD]{GD} which has dependence on $\sum_{t=1}^T\eta_t \InNorms{g^t}^2$, the main improvement of \ref{OG} is the dependence is only $\sum_{t=1}^T\eta_t \InNorms{g^t - g^{t-1}}^2$. 
\end{remark}

\begin{proof}
Fix a function $f \in \+F_{\lscc}(\+X)$. We defined $p^t = \prox_f(x^t)$. Following standard analysis of \ref{OG}~\citep{rakhlin2013optimization}, we have
\begin{align}
    & \sum_{t=1}^T \ell^t(x^t) - \ell^t(p^t) \le \sum_{t=1}^T \InAngles{\nabla \ell^t(x^t), x^t - p^t}  \nonumber \\
    & \le \sum_{t=1}^T \frac{1}{2\eta_t} \InParentheses{ \InNorms{w^{t-1} - p^t}^2 - \InNorms{w^t - p^t}^2 } + \eta_t \InNorms{g^t - g^{t-1}}^2 - \frac{1}{2\eta_t}\InParentheses{  \InNorms{x^t - w^t}^2 + \InNorms{x^t - w^{t-1}}^2}  \nonumber \\
    & = \frac{\InNorms{w^0 - p^1}^2}{2\eta_1} + \sum_{t=1}^{T-1} \frac{1}{2\eta_t}\InParentheses{ \InNorms{w^{t} - p^{t+1}}^2 - \InNorms{w^{t} - p^{t}}^2} + \sum_{t=1}^{T-1} \InParentheses{ \frac{1}{2\eta_{t+1}} - \frac{1}{2\eta_t}} \InNorms{w^t-p^{t+1}}^2\nonumber\\
    & \quad \quad - \frac{1}{2\eta_T} \InNorms{w^T - p^T}^2 + \sum_{t=1}^T \eta_t \InNorms{g^t - g^{t-1}}^2 - \sum_{t=1}^T\frac{1}{2\eta_t}\InParentheses{\InNorms{x^t - w^t}^2 + \InNorms{x^t - w^{t-1}}^2} \nonumber \\
    &\le  \frac{D^2 - \InNorms{w^T-p^T}^2}{2\eta_T} + \sum_{t=1}^{T-1} \frac{1}{2\eta_t}\InParentheses{ \InNorms{w^{t} - p^{t+1}}^2 - \InNorms{w^{t} - p^{t}}^2}+ \sum_{t=1}^T \eta_t \InNorms{g^t - g^{t-1}}^2 \nonumber  \\
    &\quad \quad -\sum_{t=1}^T \frac{1}{2\eta_t}\InParentheses{  \InNorms{x^t - w^t}^2 + \InNorms{x^t - w^{t-1}}^2}.\label{eq:OG-1}
\end{align}
In the last inequality, we use $\{\eta_t\}$ is non-increasing and $D:= \max_{0\le t\le T-1} \InNorms{w^t - p^{t+1}}$. If $\eta_t = \eta$ is constant, then the above inequality holds with $D = \InNorms{w^0 - p^1}$.

Now we apply a similar analysis from \Cref{lemma:telescope} to 
upper bound the term $\InNorms{w^{t} - p^{t+1}}^2 - \InNorms{w^{t} - p^{t}}^2$. Since $p^{t+1} = \prox_f(x^{t+1})$, we know there exists $v \in \partial f(p^{t+1})$ such that $x^{t+1} - v - p^{t+1} \in N_\+X(p^{t+1})$. This implies $\InAngles{x' - p^{t+1},x^{t+1} - v - p^{t+1} } \le 0$ for any $x' \in \+X$. 
\begin{align*}
    & \InNorms{w^{t} - p^{t+1}}^2 - \InNorms{w^{t} - p^{t}}^2 \\
     &= 2\InAngles{p^{t} - p^{t+1},  w^{t} - p^{t+1}} - \InNorms{p^t -  p^{t+1}}^2 \\
     &= 2\InAngles{p^{t} - p^{t+1}, v} + 2\InAngles{p^{t} - p^{t+1},  w^{t} - v - p^{t+1}} - \InNorms{p^{t} - p^{t+1}}^2 \\
     &= 2\InAngles{p^{t} - p^{t+1}, v} + 2\InAngles{p^{t} - p^{t+1},  x^{t+1} - v - p^{t+1}} + 2\InAngles{p^{t} - p^{t+1}, w^{t} -  x^{t+1}}- \InNorms{p^{t} - p^{t+1}}^2 \\
     &\le 2\InAngles{p^{t} - p^{t+1}, v}  + 2\InAngles{p^{t} - p^{t+1}, w^{t} -  x^{t+1}}- \InNorms{p^{t} - p^{t+1}}^2 \\
     & \le 2(f(p^t) - f(p^{t+1})) + \InNorms{w^t - x^{t+1}}^2,
\end{align*}
where in the second last-inequality we use $\InAngles{p^{t} - p^{t+1},  x^{t+1} - v - p^{t+1}} \le 0$; in the last inequality, we use convexity of $f$ and the basic inequality $2\InAngles{a,b} - \InNorms{b}^2 \le \InNorms{a}^2$.

Now we combine the above two inequalities and get
\begin{align*}
&\sum_{t=1}^T \InAngles{\nabla \ell^t(x^t), x^t - p^t} \\
&\le \frac{D^2 - \InNorms{w^T-p^T}^2}{2\eta_T} + \sum_{t=1}^{T-1}\frac{1}{2\eta_t}\InParentheses{f(p^t) - f(p^{t+1})} + \sum_{t=1}^T \eta_t \InNorms{g^t - g^{t-1}}^2  -\sum_{t=1}^T \frac{1}{2\eta_t} \InNorms{x^t - w^t}^2.
\end{align*}
Similar as \Cref{theorem:GD proximal regret}, we can use $B_f:= \max_{t \in [T]} f(p^t) - \min_{t \in [T]} f(p^t)$ to telescope the second term and get

\begin{align*}
&\sum_{t=1}^T \InAngles{\nabla \ell^t(x^t), x^t - p^t} \\
&\le \frac{D^2 + 2B_f - \InNorms{w^T-p^T}^2}{2\eta_T} + \sum_{t=1}^T \eta_t \InNorms{g^t - g^{t-1}}^2 -\sum_{t=1}^T \frac{1}{2\eta_t} \InNorms{x^t - w^t}^2. 
\end{align*}
When the step size is constant, then the above inequality holds with $B_f= f(p^1) - f(p^T)$.
\end{proof}
\subsubsection{Proof of \Cref{thm:game regret of OG}}
\begin{proof}
Let us fix any player $i \in [n]$ in the smooth game. In every step $t$, player $i$'s loss function $\ell_i^t: \X_i \rightarrow \R$  is $\InAngles{- \nabla_{x_i}u_i(x^t), \cdot }$ determined by their utility function $u_i$ and all players' actions $x^t$. Therefore,  their gradient feedback is $g^t = -\nabla_{x_i} u_i(x^t)$. 
For all $t \ge 2$, we have
\begin{align*}
    \InNorms{g^t - g^{t-1}}^2 &= \InNorms{\nabla u_i(x^t) - \nabla u_i(x^{t-1})}^2 \\
    & \le L^2\InNorms{x^t - x^{t-1}}^2 \\
    & = L^2\sum_{i=1}^n \InNorms{x^t_i - x^{t-1}_i}^2 \\
    & \le 3 L^2\sum_{i=1}^n \InParentheses{ \InNorms{x^t_i - w^t_i}^2 + \InNorms{w^t_i - w^{t-1}_i}^2 + \InNorms{w^{t-1}_i - x^{t-1}_i}^2} \\
    & \le 3n L^2 \eta^2 G^2,
\end{align*}
where we use $L$-smoothness of the utility function $u_i$ in the first inequality; we use the update rule of \ref{OG} and the fact that gradients are bounded by $G$ in the last inequality.

Applying the above inequality to the regret bound obtained in \Cref{thm:adversaril regret of OG}, the individual $\Phiproj(\delta)$-regret of player $i$ is upper bounded by
\begin{align*}
     \sum_{t=1}^T \InAngles{\nabla_{x_i} u_i(x^t), \prox_f(x^t_i)-x^t_i }\le \frac{D^2 + 2B_f}{\eta} +  2\eta G^2 + 3n L^2 G^2 \eta^3 T, \forall f \in \+F_{\lscc} 
\end{align*}
Choosing $\eta = T^{-\frac{1}{4}}$, the left hand side is bounded by $ \le (D^2+2B_f + 4nL^2G^2) T^{\frac{1}{4}}$ for any $f \in \+F_{\lscc}$. Using \Cref{lemma:no-regret-2-CE}, we have the empirical distribution of played strategy profiles converge to an $(\varepsilon + \frac{\delta^2L}{2})$-approximate $\Phiproxeq(\delta))$-\lce in $O(1/\varepsilon^{\frac{4}{3}})$ iterations.
\end{proof}

\subsection{Generalization to Bregman Setup}\label{sec:bregman proximal}
A \emph{distance-generating function} is $\phi:\+X \rightarrow \-R$ be a continuous differentiable and strictly convex function. The \emph{Bregman divergence} associated with $\phi$ is defined as 
\begin{align*}
    D_\phi(x|y) = \phi(x) - \phi(y) - \InAngles{\nabla \phi(y), x - y}, \forall x,y \in \+X.
\end{align*}

\begin{definition}[Bregman Proximal Operator]
    The \emph{Bregman proximal operator} of $f$ associated with $\phi$ is 
    \begin{align*}
        \prox^\phi_f(x) = \argmin_{y \in \+X} \{f(y) + D_\phi(y|x) \}.
    \end{align*}
\end{definition}

Either of the following conditions on $f$ and $\phi$ guarantees $\prox_f^\phi$ is single-valued and well-defined. 
\begin{itemize}
    \item $f \in \+F_{\mathrm{lsc, c}}$ is convex and lower semicontinuous.
    \item $\phi$ is $1$-strongly convex and $f \in \+F_{\mathrm{sm}}$ is $L$-smooth with $L < 1$.
\end{itemize}
For simplicity, we assume $\phi$ is $1$-strongly convex with respect to some norm $\InNorms{\cdot}$\footnote{Here the norm $\InNorms{\cdot}$ is not necessarilly $\ell_2$-norm and we denote its dual norm as $\InNorms{\cdot}_\star$.} in this section\footnote{For $\alpha$-strongly convex $\phi$, we can use $\frac{1}{\alpha}\phi$.}. Then $\prox_f$ is well defined for both $\+F_{\mathrm{lsc, c}}$ and $\+F_{\mathrm{sm}}$. Examples of $1$-strongly convex distance generated functions include
\begin{itemize}
    \item $\phi(x) = \sum_{i=1}^d x_i \log x_i$ defined on $\-R^d_{+}$. Then $D_\phi(x|y) = \sum_{i=1}^d x_i \log\frac{x_i}{y_i} + x_i - y_i$ is the KL divergence defined for all $x \ge 0$ and $y > 0$.
    \item $\phi(x) = \frac{1}{2}\InNorms{x}^2$. Then $D_\phi(x|y) = \frac{1}{2}\InNorms{x-y}^2$ is the squared Euclidean norm.
\end{itemize}
The proximal operator in \Cref{dfn:proximal operator} is a special case of Bregman proximal operator with $\phi(x) = \frac{1}{2}\InNorms{x}^2$. 

Similarly, we define the proximal regret associated with $\phi$ and $f$ as follows:
\begin{align}
    \reg^\phi_f(T) := \sum_{t=1}^T \ell^t(x^t) - \ell^t(\prox^\phi_f(x^t)). \tag{Bregman proximal regret}
\end{align}
We show that \hyperref[MD]{Mirror Descent (MD)} associated with $\phi$ has Bregman proximal regret $\reg^\phi_f(T) = O(\sqrt{T})$ for all $f \in \+F_{\mathrm{lsc, c}} \cup \+F_{\mathrm{sm}}$ simultaneously.

\begin{algorithm}[!ht]\label{MD}
    \KwIn{strategy space $\+X$, distance generating function $\phi$, step sizes $\eta > 0$} 
    \caption{Mirror Descent (MD)}
    Initialize $x^1 \in \+X$ arbitrarilly\\
    \For{$t = 1,2, \ldots,$}{
    play $x^t$ and receive $\nabla \ell^t(x^t)$.\\
    update $x^{t+1} = \argmin_{x \in \+X}\{ \InAngles{ \eta \nabla \ell^t(x^t), x} +D_\phi(x|x^t) \}$.
    }
\end{algorithm}

\begin{theorem}\label{theorem:MD proximal regret}
    Let $\+X \subseteq \-R^d$ be a closed convex set and $\phi$ a $1$-strongly convex distance generating function w.r.t. norm $\InNorms{\cdot}$. Let $\{\ell^t: \+X \rightarrow \-R\}$ be a sequence of convex loss functions that is $G$-Lipschitz w.r.t. $\InNorms{\cdot}_\star$. Let $\{x^t\}$ be the sequence generated by  \hyperref[MD]{Mirror Descent} with $\phi$ and denote $p^t = \prox^\phi_f(x^t)$. Then we have for all $f \in \+F_{\mathrm{lsc,c}}(\+X) \cup \+F_{\mathrm{sm}}(\+X)$, for all $T \ge 1$,
    \begin{align*}
        \sum_{t=1}^T \ell^t(x^t) - \ell^t(\prox^\phi_f(x^t)) \le \sum_{t=1}^T \InAngles{\nabla \ell^t(x^t), x^t - \prox^\phi_f(x^t)} \le \frac{D_\phi(p^1|x^1)+ f(p^1)- f(p^{T+1})}{\eta} + \frac{\eta G^2 T}{2}
    \end{align*}
\end{theorem}
\begin{proof}
    Denote $p^t = \prox^\phi_f(x^t)$. By convexity of $\ell^t$, we have $\sum_{t=1}^T \ell^t(x^t) - \ell^t(p^t) \le \sum_{t=1}^T \InAngles{\nabla \ell^t(x^t), x^t - p^t}$. By standard analysis of \hyperref[MD]{Mirror Descent}, we have
    \begin{align*}
        \sum_{t=1}^T \InAngles{\nabla \ell^t(x^t), x^t - p^t} &\le \sum_{t=1}^T \frac{1}{\eta} \InParentheses{D_\phi(p^t|x^t) - D_\phi(p^t|x^{t+1})} + \frac{\eta G^2 T}{2} \\
        &\le \frac{\eta G^2 T}{2} + \frac{D_\phi(p^1|x^1)}{\eta} + \frac{1}{\eta}\sum_{t=1}^{T-1} D_\phi(p^{t+1}|x^{t+1}) - D_\phi(p^t|x^{t+1})
    \end{align*}
    We focus on the term $D_\phi(p^{t+1}|x^{t+1}) - D_\phi(p^t|x^{t+1})$ that does not telescope. Recall that $p^{t+1} = \prox_f(x^{t+1}) = \argmin_{x \in \+X} \{f(x) + D_\phi(x|x^{t+1})\}$. We slightly abuse the notation and denote $\partial f(p^{t+1})$ the (sub)gradient such that $x^{t+1} - \partial f(p^{t+1}) - p^{t+1} \in N_\+X(p^{t+1})$. 
    \begin{align*}
        &D_\phi(p^{t+1}|x^{t+1}) - D_\phi(p^t|x^{t+1}) \\
        &= \InAngles{\nabla \phi(x^{t+1}) - \nabla \phi(p^{t+1}), p^t - p^{t+1}} - D_\phi(p^t|p^{t+1}) \tag{three-point identity} \\
        &= \InAngles{\partial f(p^{t+1}), p^t-p^{t+1}} + \InAngles{\nabla \phi(x^{t+1}) - \partial f(p^{t+1}) - \nabla \phi(p^{t+1}), p^t - p^{t+1}} - D_\phi(p^t|p^{t+1}) \\
        &\le \InAngles{\partial f(p^{t+1}), p^t-p^{t+1}} -D_\phi(p^t|p^{t+1}) \tag{$x^{t+1} - \partial f(p^{t+1}) - p^{t+1} \in N_\+X(p^{t+1})$} \\
        &\le  \InAngles{\partial f(p^{t+1}), p^t-p^{t+1}} -\frac{1}{2}\InNorms{p^t-p^{t+1}}^2\\
        &\le f(p^t) - f(p^{t+1}),
    \end{align*}
    where the last inequality holds for convex $f$ because $\InAngles{\partial f(p^{t+1}), p^t-p^{t+1}} \le f(p^t) - f(p^{t+1})$; the last inequality holds for $1$-smooth $f$ because $f(p^t) \ge f(p^{t+1}) + \InAngles{\partial f(p^{t+1}), p^t-p^{t+1}} -\frac{1}{2} \InNorms{p^t-p^{t+1}}^2$.

    Combining the above two inequalities gives
    \begin{align*}
        \sum_{t=1}^T \InAngles{\nabla \ell^t(x^t), x^t - p^t} &\le \frac{\eta G^2 T}{2} + \frac{D_\phi(p^1|x^1)}{\eta} + \frac{1}{\eta}\sum_{t=1}^{T-1} f(p^t) - f(p^{t+1})  \\
        &= \frac{D_\phi(p^1|x^1)+ f(p^1) -f(p^{T+1})}{\eta} + \frac{\eta G^2 T}{2}.
    \end{align*}
    This completes the proof.
\end{proof}

\section{Missing Details in \Cref{sec:convex-phi}}
\label{app:convex-phi}

\begin{algorithm}[!ht]
    \KwIn{$x_{1} \in \X$, $K \ge 2$, a no-external-regret algorithm $\mathfrak{R}_{\Phi}$ against linear reward over $\Delta(\Phi)$} 
    \KwOut{A $\Conv(\Phi)$-regret minimization algorithm over $\+X$}
    \caption{$\Conv(\Phi)$-regret minimization for Lipschitz smooth non-concave rewards}
    \label{alg:convex phi-reg}
    \Fn{\textsc{NextStrategy()}}{
    $p^t \leftarrow$ $\mathfrak{R}_\Phi$.\textsc{NextStrategy}(). Note that $p_t$ is a distribution over $\Phi$.\\
    $x_k \leftarrow \phi_{p^t}(x_{k-1})$, for all $2 \le k \le K$ \\
    \textbf{return} $x^t \leftarrow$ uniformly at random from $\{x_1, \ldots, x_K\}$. \\
    }
    \Fn{\textsc{ObserveReward}$(\nabla_x u^t(x^t))$}{
    $u^t_\Phi(\cdot) \leftarrow$ a linear reward over $\Delta(\Phi)$ with $u^t_{\Phi}(\phi) = \InAngles{\nabla_x u^t(x^t), \phi(x^t) - x^t}$ for all $\phi \in \Phi$. \\
    $\mathfrak{R}_{\Phi}$.\textsc{ObserveReward}($u^t_\Phi(\cdot)$).
    }
\end{algorithm}

\subsection{Proof of \Cref{theorem:convex finite-phi-regret}}
\begin{proof}
    For a sequence of strategies $\{x^t\}_{t \in [T]}$, its $\Conv(\Phi)$-regret is 
    \begin{align*}
        \reg^T_{\Conv(\Phi)}&=  \max_{\phi \in \Conv(\Phi)}\left\{ \sum_{t=1}^T \InParentheses{ u^t(\phi(x^t)) - u^t(x^t)} \right\}\\
        &= \underbrace{\max_{p \in \Delta(\Phi)}\left\{ \sum_{t=1}^T u^t(\phi_p(x^t)) -  u^t(\phi_{p^t}(x^t)) \right\}}_{\text{I: external regret over $\Delta(\Phi)$}} +  \underbrace{\sum_{t=1}^T  u^t(\phi_{p^t}(x^t)) - u^t(x^t)}_{\text{II: approximation error of fixed point}}.
    \end{align*}
    \paragraph{Bounding External Regret over $\Delta(\Phi)$} We can define a new reward function $f^t(p) := u^t(\phi_p(x^t))$ over $p \in \Delta(\Phi)$. Since $u^t$ is non-concave, the reward $f^t$ is also non-concave, and it is computationally intractable to minimize external regret. We use locality to avoid computational barriers. Here we use the fact that $\Phi = \Phi(\delta)$ contains only $\delta$-local strategy modifications. Then by $L$-smoothness of $u^t$, we know for any $p \in \Delta(\Phi)$
    \begin{align*}
        \left| u^t(\phi_{p}(x^t) - u^t(x^t) - \InAngles{\nabla u^t(x^t), \phi_{p}(x^t) - x^t}) \right| \le \frac{L}{2} \InNorms{\phi_{p}(x^t) - x^t}^2 \le \frac{\delta^2 L}{2}.
    \end{align*}
    Thus we can approximate the non-concave optimization problem by a linear optimization problem over $\Delta(\Phi)$ with only second-order error $\frac{\delta^2 L}{2}$. Here we use the notation $a = b \pm c$ to mean $b - c\le a \le b+ c$.
    \begin{align*}
         u^t(\phi_{p}(x^t) - u^t(x^t) &= \InAngles{\nabla u^t(x^t), \phi_{p}(x^t) - x^t} \pm  \frac{\delta^2 L}{2} \\
         & = \InAngles{\nabla u^t(x^t), \sum_{\phi \in \Phi}p(\phi)\phi(x^t) - x^t } \pm   \frac{\delta^2 L}{2} \\
         & = \sum_{\phi \in \Phi}p(\phi) \InAngles{\nabla u^t(x^t), \phi(x^t) - x^t   } \pm   \frac{\delta^2 L}{2}. 
    \end{align*}
    We can then instantiate the external regret $\mathfrak{R}_\Phi$ as the Hedge algorithm over reward \[f^t(p) = \sum_{\phi \in \Phi}p(\phi) \InAngles{\nabla u^t(x^t), \phi(x^t) - x^t}\] and get 
    \begin{align*}
        &\max_{p \in \Delta(\Phi)}\left\{ \sum_{t=1}^T u^t(\phi_p(x^t)) -  u^t(\phi_{p^t}(x^t)) \right\} \\
        &\le \max_{p \in \Delta(\Phi)}\left\{ \sum_{t=1}^T \sum_{\phi\in \Phi} (p(\phi) - p^t(\phi)) \InAngles{\nabla u^t(x^t), \phi(x^t) - x^t} ) \right\} + \delta^2LT \\
        &\le 2G\delta\sqrt{T\log|\Phi|} + \delta^2 LT,
    \end{align*}
    where we use the fact that $ \InAngles{\nabla u^t(x^t), \phi(x^t) - x^t} \le \InNorms{\nabla u^t(x^t)} \cdot \InNorms{\phi(x^t) - x^t} \le G\delta$.

    \paragraph{Bounding error due to sampling from a fixed point in expectation} We choose $x_{1}$ as an arbitrary point in $\+X$. Then we recursively apply $\phi_{p^t}$ to get 
    \begin{align*}
        x_k = \phi_{p^t}(x_{k-1}) = \sum_{\phi \in \Phi} p^t(\phi) \phi(x_{k-1}), \forall 2 \le k \le K.
    \end{align*}
    We denote $\mu^t = \mathrm{Uniform}\{x_k : 1\le k \le K\}$. Then the strategy $x^t \sim \mu^t$ is sampled from $\mu^t$. We have that $\mu^t$ is an approximate fixed-point in expectation / stationary distribution in the sense that 
    \begin{align*}
        \-E_{\mu^t}\InBrackets{u^t(\phi_{p^t}(x^t)) - u^t(x^t)} &= \frac{1}{K}\sum_{k=1}^K u^t(\phi_{p^t}(x_k) - u^t(x_k)) \\
        & = \frac{1}{K} \InParentheses{ u^t(\phi_{p^t}(x_{K})) - u^t(x_1)} \\
        & \le \frac{1}{K}.
    \end{align*}
    Thanks to the boundedness of $u^t$, we can use Hoeffding-Azuma's inequality to conclude that 
    \begin{align}\label{eq:cc azuma}
        \Pr\InBrackets{\sum_{t=1}^T  \InParentheses{u^t(\phi_{p^t}(x^t)) - u^t(x^t) - \frac{1}{K}} \ge \varepsilon} \le \exp\InParentheses{-\frac{\varepsilon^2}{8T}}.
    \end{align}
    for any $\varepsilon > 0$.
    Combining the above with $\varepsilon = \sqrt{8T\log(1/\beta)}$ and $K = \sqrt{T}$, we get with probability at least $1-\beta$,
    \begin{align*}
        \reg^T_{\Conv(\Phi)}&\le 2G\delta\sqrt{T\log|\Phi|} + \delta^2 LT + \sqrt{T} + \sqrt{8T\log(1/\beta)} \\
        & \le 8\sqrt{T} \InParentheses{ G\delta \sqrt{\log |\Phi|} + \sqrt{\log(1/\beta)}} + \delta^2L T.
    \end{align*}

    \paragraph{Convergence to $\Phi$-equilibrium} If all players in a non-concave continuous game employ \Cref{alg:phi-reg}, then we know for each player $i$,  with probability $1 - \frac{\beta}{n}$, its $\Phi^{\+X_i}$-regret is upper bounded by
    \begin{align*}
        8\sqrt{T} \InParentheses{ G\delta \sqrt{\log |\Phi^{\+X_i}|} + \sqrt{\log(n/\beta)}} + \delta^2L T.
    \end{align*}
    By a union bound over all $n$ players, we get with probability $1 - \beta$, every player $i$'s $\Phi^{\+X_i}$-regret is upper bounded by $8\sqrt{T} \InParentheses{ G\delta \sqrt{\log |\Phi^{+\X_i}|} + \sqrt{\log(n/\beta)}} + \delta^2L T$. Now by \Cref{theorem: no-phi-regret-2-phi-eq}, we know the empirical distribution of strategy profiles played forms an $(\varepsilon+\delta^2 L)$-approximate $\Phi = \Pi_{i=1}^n \Phi^{\mathcal{X}_i}$-equilibrium, as long as $T \ge  \frac{128(G^2\delta^2\log |\Phi^{\+X_i}| + \log(n/\beta))}{\varepsilon^2}$ iterations.
\end{proof}

\section{Missing details in \Cref{sec:phi-int-regret minization}}
\label{app:phi-int-regret minization}
We introduce a natural set of local strategy modifications and the corresponding local equilibrium notion. Given any set of (possibly non-local) strategy modifications $\Psi = \{ \psi: \X \rightarrow \X\}$, we define a set of \emph{local} strategy modifications as follows: for $\delta \le D_\X$ and $\lambda \in [0,1]$, each strategy modification $\phi_{\lambda, \psi}$ interpolates the input strategy $x$ with the modified strategy $\psi(x)$: formally, 
\begin{align*}
    \Phi^\X_{\Int, \Psi}(\delta):=\left\{ \phi_{\lambda, \psi}(x) := (1- \lambda)x + \lambda \psi(x): \psi \in \Psi, \lambda \le \delta / D_\X\right\}.
\end{align*}
Note that for any $\psi \in \Psi$ and $\lambda \le \frac{\delta}{D_\X}$, we have $\InNorms{\phi_{\lambda, \psi}(x) - x} = \lambda \InNorms{x - \psi(x)} \le \delta$, respecting the locality constraint. The induced $\Phi^\X_{\Int, \Psi}(\delta)$-regret can be written as 
\[   \reg^T_{\Int,\Psi, \delta} :=   \max_{\psi \in \Psi, \lambda \le \frac{\delta}{D_\X}} \sum_{t=1}^T \InParentheses{f^t(x^t) - f^t\InParentheses{(1-\lambda)x^t + \lambda \psi(x^t)}}.
\]
We now define the corresponding $\Phi_{\Int, \Psi}(\delta)$-equilibrium.
\begin{definition}
    Define $\Phi_{\Int, \Psi}(\delta) = \Pi_{j=1}^n  \Phi_{\Int, \Psi_j}^{\X_j}(\delta)$. In a continuous game, a distribution $\sigma$ over strategy profiles is an $\varepsilon$-approximate $ \Phi_{\Int, \Psi}(\delta)$-equilibrium if and only if for all player $i \in [n]$, 
    \[ 
    \max_{\psi \in \Psi_i, \lambda \le \delta/D_{\X_i}} 
    \-E_{x \sim \sigma} \InBrackets{u_i((1-\lambda)x_i + \lambda \psi(x_i), x_{-i})} \le \-E_{x \sim \sigma} \InBrackets{u_i(x)} + \varepsilon.
    \]
\end{definition}
Intuitively speaking, when a correlation device recommends strategies to players according to an $\varepsilon$-approximate $\Phi_{\Int, \Psi}(\delta)$-equilibrium, no player can increase their utility by more than $\varepsilon$ through a local deviation by interpolating with a (possibly global) strategy modification $\psi \in \Psi$. The richness of $\Psi$ determines the incentive guarantee provided by an $\varepsilon$-approximate $\Phi_{\Int, \Psi}(\delta)$-equilibrium and its computational complexity. When we choose $\Psi$ to be the set of all possible strategy modifications, the corresponding notion of local equilibrium—limiting a player's gain by interpolating with any strategy—resembles that of a \emph{correlated equilibrium}.

\paragraph{Computation of  $\varepsilon$-approximate 
$\Phi_{\Int, \Psi}(\delta)$-Equilibrium.}
By \Cref{lemma:no-regret-2-CE}, we know computing an $\varepsilon$-approximate $\Phi_{\Int, \Psi}(\delta)$-equilibrium reduces to minimizing $\Phi^\X_{\Int, \Psi}(\delta)$-regret against convex loss functions. We show that minimizing $\Phi^\X_{\Int, \Psi}(\delta)$-regret against convex loss functions further reduces to $\Psi$-regret minimization against linear loss functions. 

\begin{theorem}
\label{thm:int reduction}
    Let $\+A$ be an algorithm with $\Psi$-regret $\reg_{\Psi}^T(G,\Dx)$ for linear and $G$-Lipschitz loss functions over $\X$. Then, for any $\delta> 0$, the $\Phi^\X_{\Int, \Psi}(\delta)$-regret of $\+A$ for convex and $G$-Lipschitz loss functions over $\X$ is at most $\frac{\delta}{\Dx}\cdot \InBrackets{\reg_\Psi^T(G,\Dx)}^+$. 
\end{theorem}
\begin{proof}
    By definition and convexity of $f^t$, we get 
    \begin{align*}
        \max_{\phi \in \Phi^\X_{\Int, \Psi}(\delta)} \sum_{t=1}^T f^t(x^t) - f^t(\phi(x^t)) &= \max_{\psi\in \Psi, \lambda\le \frac{\delta}{D_\X}} \sum_{t=1}^T  f^t(x^t) - f^t((1-\lambda)x^t + \lambda \psi(x^t)) \\
        &\le \frac{\delta}{D_\X} \left[ \max_{\psi \in \Psi} \sum_{t=1}^T \InAngles{\nabla f^t(x^t), x^t - \psi(x^t)}  \right]^{+}.
    \end{align*} 
\end{proof}
Note that when $f^t$ is linear, the reduction is without loss. Thus, any worst-case $\Omega(r(T))$-lower bound for $\Psi$-regret implies a $\Omega(\frac{\delta}{D_\X}\cdot r(T))$ lower bound for $\Phi_{\Int,\Psi}(\delta)$-regret. Moreover, for any set $\Psi$ that admits efficient $\Psi$-regret minimization algorithms such as swap transformations over the simplex and more generally any set such that (i) all modifications in the set can be represented as linear transformations in some finite-dimensional space and (ii) fixed point computation can be carried out efficiently for any linear transformations~\citep{gordon2008no}, we also get an efficient algorithm for computing an $\varepsilon$-approximate $\Phi_{\Int, \Psi}(\delta)$-equilibrium in the first-order stationary regime.

\paragraph{CCE-like Instantiation} In the special case where $\Psi$ contains only \emph{constant} strategy modifications (i.e., $\psi(x) = x^*$ for all $x$), we get a coarse correlated equilibrium (CCE)-like instantiation of local equilibrium, which limits the gain by interpolating with any \emph{fixed} strategy. We denote the resulting set of local strategy modification simply as $\Phiint$. We can apply any no-external regret algorithm for efficient $\Phiint$-regret minimization and computation of $\varepsilon$-approximate $\Phiinteq(\delta)$-\lce in the first-order stationary regime as summarized in \Cref{thm:lce_int}.

The above $\Phiint(\delta)$-regret bound of $O(\sqrt{T})$ is derived for the adversarial setting. In the game setting, where each player employs the same algorithm, players may have substantially lower external regret~\citep{syrgkanis2015fast,chen2020hedging,daskalakis2021near-optimal,anagnostides2022near-optimal, anagnostides2022uncoupled, farina2022near} 
but we need a slightly stronger smoothness assumption than \Cref{assumption:smooth games}. This assumption is naturally satisfied by finite normal-form games and is also made for results about concave games~\citep{farina2022near}. Using \Cref{assumption:stronger smoothness} and \Cref{lemma:no-regret-2-CE}, the no-regret learning dynamics of \citep{farina2022near} that guarantees  $O(\log T)$ individual external regret in concave games can be applied to smooth non-concave games so that the individual  $\Phiint(\delta)$-regret of each player is at most $O(\log T) + \frac{\delta^2 LT}{2}$. This gives an algorithm with faster $\Tilde{O}(1/\varepsilon)$ convergence to an $(\varepsilon + \frac{\delta^2L}{2})$-approximate $\Phiinteq(\delta)$-\lce than \hyperref[GD]{Online Gradient Descent}.

\section{Beam-Search Local Strategy Modifications and Local Equilibria}\label{sec:beam search}
In \Cref{sec:proximal regret} and \Cref{sec:phi-int-regret minization}, we have shown that \hyperref[GD]{GD}  achieves near-optimal performance for both $\Phiint(\delta)$-regret and $\Phiprox(\delta)$-regret. In this section, we introduce another natural set of local strategy modifications, $\Phibeam(\delta)$, which is similar to $\Phiproj(\delta)$. Specifically, the set $\Phibeam(\delta)$ contains deviations that try to move as far as possible in a fixed direction (see \Cref{fig:beam-proj} for an illustration of the difference between $\phi_{\Beam, v}(x)$ and $\phi_{\Proj, v}(x)$): \[\Phibeam(\delta):= \{ \phi_{\Beam, v}(x) = x - \lambda^* v: v\in B_d(\delta), \lambda^* = \max \{\lambda: x-\lambda v\in \X, \lambda\in[0,1]\}
\}.\]
It is clear that $\InNorms{\phi_{\Beam, v}(x) - x} \le \InNorms{v} \le \delta$. We can similarly derive the notion of $\Phibeam$-regret and $(\varepsilon, \Phibeameq(\delta))$-\lce. Surprisingly, we show that \hyperref[GD]{GD}  suffers linear $\Phibeam(\delta)$-regret (proof deferred to \Cref{sec:proof beam search}).

\begin{figure}[t]
    \caption{Illustration of $\phi_{\Proj, v}(x)$ and $\phi_{\Beam, v}(x)$}
    \centering
    \includegraphics[width=.5\textwidth]{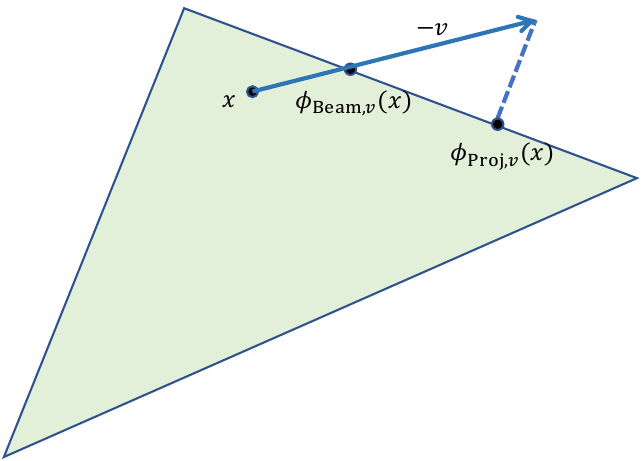}
    \label{fig:beam-proj}
\end{figure}

\begin{theorem}\label{thm:beam,GD linear regret}
    For any $\delta, \eta < \frac{1}{2}$ and $T \ge 1$, there exists a sequence of linear loss functions $\{f^t: \X \subseteq [0,1]^2 \rightarrow \R\}_{ t\in [T]}$ such that \hyperref[GD]{GD}  with step size $\eta$ suffers $\Omega(\delta T)$ $\Phibeam(\delta)$-regret.
\end{theorem}
Our results show that even for simple local strategy modification sets $\Phi(\delta)$, the landscape of efficient local $\Phi(\delta)$-regret minimization is already quite rich, and many basic and interesting questions remain open.

\subsection{Proof of Theorem~\ref{thm:beam,GD linear regret}}\label{sec:proof beam search} 

Let $\X \subset \mathbb{R}^2$ be a triangle region with vertices $A=(0,0)$, $B=(1,1)$, $C=(\delta, 0)$. Consider $v=(-\delta, 0)$. The initial point is $x_1=(0,0)$. 

The adversary will choose $\ell_t$ adaptively so that $x_t$ remains on the boundary of $\X$ and cycles clockwise (i.e., $A\rightarrow \cdots \rightarrow B \rightarrow \cdots \rightarrow C \rightarrow \cdots \rightarrow A\rightarrow \cdots$). To achieve this, the adversary will repeat the following three phases: 
\begin{enumerate}
    \item Keep choosing $\ell_t=u_{\overrightarrow{BA}}$ ($u_{\overrightarrow{BA}}$ denotes the unit vector in the direction of $\overrightarrow{BA}$) until $x_{t+1}$ reaches~$B$. 
    \item Keep choosing $\ell_t=u_{\overrightarrow{CB}}$ until $x_{t+1}$ reaches $C$. 
    \item Keep choosing $\ell_t=u_{\overrightarrow{AC}}$ until $x_{t+1}$ reaches $A$. 
\end{enumerate}

In Phase 1, $x_t\in \overline{AB}$. By the choice of $v=(-\delta, 0)$, we have $x_t - \phi_v(x_t)= (-\delta(1-x_{t,1}), 0) $, and the instantaneous regret is $\frac{\delta(1-x_{t,1})}{\sqrt{2}} \geq 0$. 

In Phase 2, $x_t\in \overline{BC}$. By the choice of $v=(-\delta, 0)$, we have $x_t - \phi_v(x_t)= (0, 0)$, and the instantaneous regret is $0$. 

In Phase 3, $x_t\in \overline{CA}$. By the choice of $v=(-\delta, 0)$, we have $x_t - \phi_v(x_t) = (-\delta + x_{t,1}, 0)$, and the instantaneous regret is $-\delta + x_{t,1} \leq 0$. 

In each cycle, the number of rounds in Phase 1 is of order $\Theta(\frac{\sqrt{2}}{\eta})$, the number of rounds in Phase 2 is between $O(\frac{1}{\eta})$ and $O(\frac{\sqrt{2}}{\eta})$, the number of rounds in Phase 3 is of order $\Theta(\frac{\delta}{\eta})$. 

Therefore, the cumulative regret in each cycle is roughly 
\begin{align*}
    &\frac{\sqrt{2}}{\eta} \times \frac{0.5\delta}{\sqrt{2}} + 0 + \frac{\delta}{\eta} \left(-0.5\delta\right) = \frac{0.5\delta - 0.5\delta^2}{\eta}. 
\end{align*}
On the other hand, the number of cycles is no less than $\frac{T}{\frac{\sqrt{2}}{\eta} + \frac{\sqrt{2}}{\eta} + \frac{\delta}{\eta}} = \Theta(\eta T)$.  
Overall, the cumulative regret is at least $\frac{0.5\delta- 0.5\delta^2}{\eta}\times \Theta(\eta T) = \Theta(\delta T)$ as long as $\delta < 0.5$.

\section{Hardness for Approximate $\Phiprojeq(\delta)/ \Phiinteq(\delta)$-Equilibrium when $\delta = D$}\label{sec:hardness in global regime}
In the first-order stationary regime $\delta \le \sqrt{2\varepsilon / L}$, $(\varepsilon, \delta)$-local Nash equilibrium is intractable, and we have shown polynomial-time algorithms for computing the weaker notions of $\varepsilon$-approximate $\Phiinteq(\delta))$-equilibrium and $\varepsilon$-approximate $\Phiprojeq(\delta))$-equilibrium.  A natural question is whether correlation enables efficient computation of $\varepsilon$-approximate $\Phi(\delta))$-equilibrium when $\delta$ is in the global regime, i.e., $\delta = \Omega(\sqrt{d})$. In this section, we prove both computational hardness and a query complexity lower bound for both notions in the global regime

To prove the lower bound results, we only require a single-player game. The problem of computing an $\varepsilon$-approximate $\Phi(\delta)$-equilibrium becomes: given scalars $\varepsilon, \delta, G, L > 0$ and a polynomial-time Turing machine $\+C_f$ evaluating a $G$-Lipschitz and $L$-smooth function $f : [0,1]^d \rightarrow [0,1]$ and its gradient $\nabla f: [0,1]^d \rightarrow \R^d$, we are asked to output a distribution $\sigma$ that is an $\varepsilon$-approximate $\Phi(\delta)$-equilibrium or $\perp$ if such equilibrium does not exist. 

\paragraph{Hardness of finding $\varepsilon$-approximate $\Phiint(\delta)$-equilibria in the global regime} When $\delta = \sqrt{d}$, which equals to the diameter $D$ of $[0,1]^d$, then the problem of finding an $\varepsilon$-approximate $\Phiint(\delta)$-equilibrium is equivalent to finding a $(\varepsilon, \delta)$-local minimizer of $f$: assume $\sigma$ is an $\varepsilon$-approximate $\Phiint(\delta)$-equilibrium of $f$, then there exists $x\in [0,1]^d$ in the support of $\sigma$ such that 
\[
f(x) - \min_{x^* \in [0,1]^d \cap B_d(x^*,\delta)} f(x^*) \le \varepsilon.
\]
Then hardness of finding an $\varepsilon$-approximate $\Phiint(\delta)$-equilibrium follows from hardness of finding a $(\varepsilon, \delta)$-local minimizer of $f$~\citep{daskalakis2021complexity}. The following Theorem is a corollary of Theorem 10.3 and 10.4 in \citep{daskalakis2021complexity}.
\begin{theorem}[Hardness of finding $\varepsilon$-approximate $\Phiint(\delta)$-equilibria in the global regime]
\label{thm:lce-hardness-2}
    In the worst case, the following two holds.
    \begin{itemize}
        \item Computing an $\varepsilon$-approximate $\Phiint(\delta)$-equilibrium for a game on $\X = [0,1]^d$ with $G = \sqrt{d}$, $L = d$, $\varepsilon \le \frac{1}{24}$, $\delta = \sqrt{d}$ is NP-hard. 
        \item $\Omega(2^d /d)$ value/gradient queries are needed to determine an $\varepsilon$-approximate $\Phiint(\delta)$-equilibrium for a game on $\X = [0,1]^d$ with $G = \Theta(d^{15})$, $L = \Theta(d^{22})$, $\varepsilon < 1$, $\delta = \sqrt{d}$. 
    \end{itemize}
\end{theorem}
\paragraph{Hardness of finding $\varepsilon$-approximate $\Phiproj(\delta)$-equilibria in the global regime} 
\begin{theorem}[Hardness of of finding $\varepsilon$-approximate $\Phiproj(\delta)$-equilibria in the global regime]
\label{thm:lce-hardness}
    In the worst case, the following two holds.
    \begin{itemize}
        \item Computing an $\varepsilon$-approximate $\Phiproj(\delta)$-equilibrium for a game on $\X = [0,1]^d$ with $G = \Theta(d^{15})$, $L = \Theta(d^{22})$, $\varepsilon < 1$, $\delta = \sqrt{d}$ is NP-hard. 
        \item $\Omega(2^d /d)$ value/gradient queries are needed to determine an $\varepsilon$-approximate $\Phiproj(\delta)$-equilibrium for a game on $\X = [0,1]^d$ with $G = \Theta(d^{15})$, $L = \Theta(d^{22})$, $\varepsilon < 1$, $\delta = \sqrt{d}$. 
    \end{itemize}
    
\end{theorem}

The hardness of computing $\varepsilon$-approximate $\Phiproj(\delta)$-equilibrium also implies a lower bound on $\Phiproj(\delta)$-regret in the global regime.
\begin{corollary}[Lower bound of $\Phiproj(\delta)$-regret against non-convex functions]\label{coro:lower bounds for global rerget}
    In the worst case, the $\Phiproj(\delta)$-regret of any online algorithm is at least $\Omega(2^d / d, T)$ even for loss functions $f: [0,1]^d \rightarrow [0,1]$ with $G, L = \poly(d)$ and $\delta = \sqrt{d}$.
\end{corollary}
The proofs of \Cref{thm:lce-hardness} and \Cref{coro:lower bounds for global rerget} can be found in the next two sections.

\subsection{Proof of \Cref{thm:lce-hardness}}
We will reduce the problem of finding an $\varepsilon$-approximate $\Phiproj(\delta)$-equilibrium in smooth games to finding a satisfying assignment of a boolean function, which is NP-complete.
\begin{fact}
    \label{fact:boolean SAT-lower bound}
    Given only \emph{black-box access} to a boolean formula $\phi: \{0,1\}^d \rightarrow \{0,1\}$, at least $\Omega(2^d)$ queries are needed in order to determine whether $\phi$ admits a satisfying assignment $x^*$ such that $\phi(x^*) = 1$. The term \emph{black-box access} refers to the fact that the clauses of the formula are not given, and the only way to determine whether a specific boolean assignment is satisfying is by querying the specific binary string. Moreover, the problem of finding a satisfying assignment of a general boolean function is NP-hard.
\end{fact}
We revisit the construction of the hard instance in the proof of \citep[Theorem 10.4]{daskalakis2021complexity} and use its specific structures. Given black-box access to a boolean formula $\phi$ as described in \Cref{fact:boolean SAT-lower bound}, following \citep{daskalakis2021complexity}, we construct the function $f_\phi(x): [0,1]^d \rightarrow [0,1]$ as follows:
\begin{itemize}
    \item[1.] for each corner $v \in V = \{0, 1\}^d$ of the $[0, 1]^d$ hypercube, we set $f_\phi(x) = 1 - \phi(x)$.
    \item[2.] for the rest of the points $x \in [0, 1]^d / V$, we set $f_\phi(x) = \sum_{v \in V} P_v(x) \cdot f_\phi(v)$ where $P_v(x)$ are non-negative coefficients defined in \citep[Definition 8.9]{daskalakis2021complexity}.
\end{itemize}
The function $f_\phi$ satisfies the following properties: 
\begin{itemize}
    \item[1.] if $\phi$ is not satisfiable, then $f_\phi(x) = 1$ for all $x \in [0, 1]^d$ since $f_\phi(v) = 1$ for all $v \in V$; if $\phi$ has a satisfying assignment $v^*$, then $f_\phi(v^*) = 0$.
    \item[2.] $f_\phi$ is $\Theta(d^{12})$-Lipschitz and $\Theta(d^{25})$-smooth. 
    \item[3.] for any point $x \in [0,1]^d$, the set $V(x): = \{ v \in V: P_v(x) \ne 0\}$ has cardinality at most $d+1$ while $\sum_{v \in V}P_v(x) = 1$; any value / gradient query of $f_\phi$ can be simulated by $d+1$ queries on $\phi$. 
\end{itemize}

In the case there exists a satisfying argument $v^*$, then $f_\phi(v^*) = 0$. Define the deviation $e$ so that $e[i] = 1$ if $v^*[i] = 0$ and $e[i] = -1$ if $v^*[i] = 1$. It is clear that $\InNorms{e} = \sqrt{d} = \delta$. By properties of projection on $[0,1]^d$, for any $x \in [0,1]^d$, we have $\Pi_{[0,1]^n}[x - v] = v^*$. Then any $\varepsilon$-approximate $\Phiproj(\delta)$-equilibrium $\sigma$ must include some $x^* \in \X$ with $f_\phi(x^*) < 1$ in the support, since $\varepsilon < 1$. In case there exists an algorithm $\+A$ that computes an $\varepsilon$-approximate $\Phiproj(\delta)$-equilibrium, $\+A$ must have queried some $x^*$ with $f_\phi(x^*) < 1$. Since $f_\phi(x^*) = \sum_{v \in V(x^*)} P_v(x^*) f_\phi(v) < 1$, there exists $\hat{v} \in V(x^*)$ such that $f_{\phi}(\hat{v}) = 0$. Since $|V(x^*)| \le d+1$, it takes addition $d+1$ queries to find $\hat{v}$ with $f_\phi(\hat{v}) = 0$. By  \cref{fact:boolean SAT-lower bound} and the fact that we can simulate every value/gradient query of $f_\phi$ by $d+1$ queries on $\phi$, $\+A$ makes at least $\Omega(2^d/ d)$ value/gradient queries.  

Suppose there exists an algorithm $\+B$ that outputs an $\varepsilon$-approximate $\Phiproj(\delta)$-equilibrium $\sigma$ in time $T(\+B)$ for $\varepsilon < 1$ and $\delta = \sqrt{d}$. We construct another algorithm $\+C$ for SAT that terminates in time $T(\+B) \cdot \poly(d)$.   $\+C$: (1) given a boolean formula $\phi$, construct $f_\phi$ as described above; (2) run $\+B$ and get output $\sigma$ (3) check the support of $\sigma$  to find $v \in \{0,1\}^d$ such that $f_\phi(v) = 0$; (3) if finds $v \in \{0,1\}^d$ such that $f_\phi(v) = 0$, then $\phi$ is satisfiable, otherwise $\phi$ is not satisfiable. Since we can evaluate $f_\phi$ and $\nabla f_\phi$ in $\poly(d)$ time and the support of $\sigma$ is smaller than $T(\+B)$, the algorithm $\+C$ terminates in time $O(T(\+B) \cdot \poly(d))$. The above gives a polynomial time reduction from SAT to finding an $\varepsilon$-approximate $\Phiproj(\delta)$-equilibrium and proves the NP-hardness of the latter problem.

\subsection{Proof of \Cref{coro:lower bounds for global rerget}}
Let $\phi:\{0,1\}^d \rightarrow \{0,1\}$ be a boolean formula and define $f_\phi: [0,1]^d \rightarrow [0,1]$ the same as that in \Cref{thm:lce-hardness}. We know $f_\phi$ is $\Theta(\poly(d))$-Lipschitz and $\Theta(\poly(d))$-smooth. 
 Now we let the adversary pick $f_\phi$ each time. For any $T \le O(2^d /d)$, in case there exists an online learning algorithm with $ \regproj^T < \frac{T}{2}$, then $\sigma := \frac{1}{T}\sum_{t=1}^T 1_{x^t}$ is an $(\frac{1}{2}, \delta)$-\lce{}. Applying \Cref{thm:lce-hardness} and the fact that in this case, $\regproj^T$ is non-decreasing with respect to $T$ concludes the proof. 
\notshow{
If $\phi$ is satisfiable with $\phi(v^*) = 1$, then any online algorithm with $\delta$-local regret smaller than $T$ implies
\begin{align*}
    T &> \sum_{t=1}^T f_\phi(x^t) -   \min_{\InNorms{v} \le \delta} \sum_{t=1}^T f_\phi(\Pi_{[0,1]^d}[x^t - v])\\
    &\ge \sum_{t=1}^T f_\phi(x^t) -    \sum_{t=1}^T f_\phi(v^*) \tag{Take $v$ such that $v[i] = -1$ if $v^*[i] = 1$ and $v[i] = 1$ if $v^*[i] = 0$} \\
    & = \sum_{t=1}^T f_\phi(x^t).
\end{align*}
Thus, at least one $x^t$ with $f_\phi(x^t) < 1$ exists. In case there exists an online learning algorithm with $\delta$-local regret $< T$ for some $T$ less than $O(2^d / d)$,  the following algorithm determines the satisfiability of $\phi$ using less than $O(2^d)$ queries:
\begin{itemize}
    \item[1.] simulate the online learning algorithm against loss function $f_\phi$ and answer each value / gradient query of $f_\phi(x^t)$ using at most $d + 1$ queries of $\phi$. 
    \item[2.] if before time $T$, the algorithm queries $x^t$ with $f_\phi(x^t) < 1$, then $\phi$ is satisfiable; otherwise, $\phi$ is not satisfiable.
\end{itemize}
Invoking \Cref{fact:boolean SAT-lower bound} and properties of the $f_\phi$ concludes the proof.}

\end{document}